\documentclass{article}
\usepackage[margin=1in]{geometry}
\usepackage[onehalfspacing]{setspace}
\usepackage{amsmath,mathtools}
\usepackage{pgfplots}
\usepackage{subcaption}
\usepackage{booktabs}
\usepackage[braket, qm]{qcircuit}
\usepackage{graphicx}
\usepackage{authblk}
\usepackage{xcolor}
\usepackage[linesnumbered,ruled,vlined]{algorithm2e}
\usepackage{qcircuit}
\usepackage{booktabs}
\usepackage{cite}

\usepackage{float}

\SetCommentSty{mycommfont}

\SetKwInput{KwInput}{Input}                
\SetKwInput{KwOutput}{Output}              
\SetKwProg{Fn}{Function}{}{}

\newcommand{\meqref}[1]{\text{Eq}.~\eqref{#1}}
\newcommand{\mref}[1]{Sec.\,\,\!$\ref{#1} $}
\newcommand{\mfig}[1]{Fig.\,\,\!$\ref{#1} $}

\usepackage[english]{babel}

\usepackage{pdflscape}
\usepackage{fancyvrb}
\usepackage{calc}
\usepackage{rotating}
\usepackage{pgf,tikz}
\usepackage{mathrsfs}
\usetikzlibrary{arrows}
\usepackage{mathtools}

\DeclarePairedDelimiter\floor{\lfloor}{\rfloor}
\usepackage{listings}
\usepackage[]{qcircuit}
\usepackage{braket}
\usepackage{mathtools}
\usepackage{varwidth}
\usepackage{caption}
\usepackage{subcaption}
\usepackage{graphicx}
\usepackage[export]{adjustbox}

\makeatletter
\def\paragraph{\@startsection{paragraph}{4}%
	\z@\z@{-\fontdimen2\font}%
	{\normalfont\bfseries}}
\makeatother

\usepackage{graphicx}
\usepackage{pgffor}
\usepackage{tikz,tikz-cd}
\usetikzlibrary{backgrounds,fit,decorations.pathreplacing}  
\usepackage{booktabs}
\usepackage{enumerate}

\usepackage{amssymb, amsmath, amsthm}

\usepackage{xypic}
\usepackage{scalerel}
\usepackage{ulem}

\usepackage{graphicx}              
\usepackage{amsmath}               
\usepackage{amsfonts}              
\usepackage{amsthm}                
\usepackage{xkeyval}               
\usepackage{amssymb}
\usepackage{enumerate}             
\usepackage{color}
\usepackage{breqn}                 
\usepackage{bm}                    

\allowdisplaybreaks[1]
\usepackage{nicefrac}
\usepackage{booktabs}

\usepackage{kpfonts}
\usepackage{stackengine}
\usepackage{calc}
\newlength\shlength
\newcommand\xshlongvec[2][0]{\setlength\shlength{#1pt}%
	\stackengine{-5.6pt}{$#2$}{\smash{$\kern\shlength%
			\stackengine{7.55pt}{$\mathchar"017E$}%
			{\rule{\widthof{$#2$}}{.57pt}\kern.4pt}{O}{r}{F}{F}{L}\kern-\shlength$}}%
	{O}{c}{F}{T}{S}}

\usepackage{hyperref}
\hypersetup{
	colorlinks   = true,    
	citecolor    = red      
}

\newcommand{\RN}[1]{%
	\textup{\uppercase\expandafter{\romannumeral#1}}%
}

\allowdisplaybreaks

\newtheorem{thm}{Theorem}[subsection]

\newtheorem{lem}[thm]{Lemma}

\newtheorem{cor}[thm]{Corollary}

\newtheorem{defn}[thm]{Definition} 
\newtheorem{example}[thm]{Example} 
\newtheorem{remark}[thm]{Remark}

\def\CC{{\mathbb C}}

\def\<{\langle}
\def\>{\rangle}

\numberwithin{equation}{section}

\newcommand{\abs}[1]{\lvert#1\rvert}

\begin{document}

 \title{A quantum approach for digital signal processing}

	\author[1]{Alok Shukla \thanks{Corresponding author.}}
	\author[2]{Prakash Vedula}
	\affil[1]{School of Arts and Sciences, Ahmedabad University, India}
	\affil[1]{alok.shukla@ahduni.edu.in}
	\affil[2]{School of Aerospace and Mechanical Engineering, University of Oklahoma, USA}
	\affil[2]{pvedula@ou.edu}
	

	\maketitle


	\date{}
	
	\maketitle

\begin{abstract}
We propose a novel quantum approach to signal processing, including a quantum algorithm for low-pass and high-pass filtering, based on the sequency-ordered Walsh-Hadamard transform. We present quantum circuits for performing the sequency-ordered Walsh-Hadamard transform, as well as quantum circuits for low-pass, high-pass, and band-pass filtering. Additionally, we provide a proof of correctness for the quantum circuit designed to perform the sequency-ordered Walsh-Hadamard transform.
The performance and accuracy of the proposed approach for signal filtering were illustrated using computational examples, along with corresponding quantum circuits, for DC, low-pass, high-pass, and band-pass filtering. 
Our proposed algorithm for signal filtering has a reduced gate complexity and circuit depth of $O (\log_2 N)$, compared to at least $O ((\log_2 N )^2)$ associated with Quantum Fourier Transform (QFT) based filtering (excluding state preparation and measurement costs). In contrast, classical Fast Fourier Transform (FFT) based filtering approaches have a complexity of $O (N \log_2 N )$. This shows that our proposed approach offers a significant improvement over QFT-based filtering methods and classical FFT-based filtering methods. Such enhanced efficiency of our proposed approach holds substantial promise across several signal processing applications by ensuring faster computations and efficient use of resources via reduced circuit depth and lower gate complexity.
\end{abstract}

\section{Introduction}\label{sec:intro}

Signal processing is an important field that deals with the manipulation and analysis of signals to extract
meaningful information or improve their quality. DC, low-pass and high-pass filtering are essential techniques in signal processing. These techniques are used in various applications such as audio processing,
communication systems, medical signal analysis, and more, enabling the extraction of relevant information
from complex signals.

Traditionally, a low-pass filter eliminates high-frequency components (based on Fourier decomposition), allowing only low-frequency components (DC and low-frequency signals) to pass through, which is useful in removing noise from signals. On the other hand, a high-pass filter allows high-frequency components to pass through, eliminating low-frequency noise or baseline drift from signals. Similar to frequency filtering, one can use sequency filtering based on Walsh basis functions in signal processing applications. Sequency filtering using Walsh basis functions involves selectively preserving the desired sequency components of a signal represented in the Walsh domain while eliminating the unwanted sequency components. The Walsh domain refers to the representation of a signal obtained using the Walsh-Hadamard transform. The notion of sequency and the Walsh-Hadamard transform will also be briefly reviewed in \mref{sec:Sequency-ordered-WH-transforms}.

The paper introduces a new quantum approach for signal filtering using the Walsh-Hadamard transform. It presents quantum circuits for low-pass, high-pass, and band-pass filtering, along with a quantum circuit for performing the sequency-ordered Walsh-Hadamard transform. The correctness of the approach is proven and its performance and accuracy are demonstrated through computational examples. The algorithm is computationally efficient, with gate complexity and circuit depth of $O (\log_2 N)$, making it more efficient compared to alternative quantum filtering methods based on QFT and classical FFT-based approaches. 
The proposed approach shows promising potential for faster and more resource-efficient signal processing applications.

We note that usually the Walsh-Hadamard transform in natural order appears in quantum algorithms (for example, Deutsch-Jozsa algorithm \cite{deutsch1992rapid}, Bernstein–Vazirani algorithm \cite{bernstein1993quantum} and its probabilistic generalization \cite{shukla2023generalization}, Simon's algorithm \cite{simon1997power}, Grover's algorithm \cite{grover1996fast}, Shor's Algorithm \cite{shor1999polynomial}, etc.), often to get a uniform superposition of quantum states at the beginning of the quantum algorithm. Walsh basis functions in natural order have also been used in the solution of non-linear ordinary differential equations \cite{SHUKLA2022127708}. However, for image and signal processing applications, the use of sequency-ordered Walsh-Hadamard transforms is often more convenient (see ~\cite{kuklinski1983fast, zarowski1985spectral}) compared to natural-ordered Walsh-Hadamard transforms. Sequency-ordered Walsh-Hadamard transforms offer a better representation of a signal in the transformed domain and they have better energy compaction properties. The hybrid classical-quantum algorithm approach for image processing discussed in \cite{Shukla2022} uses the Walsh-Hadamard transforms in sequency ordering. We note that in \cite{Shukla2022}, the quantum Walsh-Hadamard transform is performed in natural order, followed by classical computations to obtain the transforms in sequency ordering. In this and other applications, for example, cryptography~\cite{lu2016walsh}, solution of non-linear ordinary differential equations and partial differential equations~\cite{beer1981walsh,ahner1988walsh,gnoffo2014global, gnoffo2015unsteady, gnoffo2017solutions},
wherein the Walsh-Hadamard transform in sequency ordering is needed, it is desirable to have a quantum circuit to compute the Walsh-Hadamard transforms directly in sequency ordering.

We discuss Walsh series representation of a signal in \mref{Sec:fun_representation} and a quantum circuit that performs the sequency-ordered Walsh-Hadamard transform is given in \mref{Sec:seq_WH}. We also provide proof of correctness \mref{sec:Lemma}, establishing that the quantum circuit given in  \mref{Sec:seq_WH} produces the desired sequency-ordered Walsh-Hadamard transformed state. This circuit will be a key component in our approach to low-pass, high-pass and band-pass filtering of signals. In  \mref{Sec:low-pass-high-pass}, we present a quantum algorithm, Algorithm \ref{alg_filtering}, for low-pass and high-pass filtering in the sequency domain. 
We also provide computational examples, along with quantum circuits, to illustrate the application of our filtering algorithm for low-pass and high-pass filtering in \mref{Sec:low-pass-high-pass}. 
Our quantum approaches for DC and band-pass filtering using Walsh-Hadamard transforms are discussed in \mref{Sec:DC} and \mref{Sec:bandpass}, respectively. 
Quantum circuits and illustrative computational examples for DC and band-pass filtering are also provided in \mref{Sec:DC} and \mref{Sec:bandpass}.
Filtered signals and corresponding spectra obtained from our proposed quantum approach in all the cases (i.e., DC, low-pass, high-pass and band-pass filtering) match the expected results. 

The complexity analysis for our approach is presented in \mref{Sec:complexity}. 
The gate complexity and circuit depth of the quantum circuit based on Algorithm \ref{alg_filtering} are shown to be $O (\log_2 N)$, which is a significant improvement over existing quantum filtering methods based on Quantum Fourier Transform (QFT). Traditional QFT-based methods require at least $O((\log_2 N)^2)$ gates and circuit depth, while classical Fast Fourier Transform (FFT)-based approaches have a complexity of $O (N \log_2 N)$. Considering the discussion above, our proposed approach holds promise for faster and more resource-efficient computations in various signal-processing applications.

\subsection{Notation} \label{sec:notation}Before proceeding further, we fix some convenient notations used in the rest of the paper.  
	\begin{itemize}
		\item $ \oplus $ : $ x \oplus y $ will denote $ x + y \mod 2 $.
	    \item $ s \cdot x $ : For $ s = s_{n-1}\,s_{n-2}\,\ldots \, s_1\, s_0 $ and $ x =  x_{n-1}\,x_{n-2} \,\ldots \, x_1\, x_0 $ with $ s_i,\, x_i \in \{0,1\}$, $ s \cdot x $ will denote the bit-wise dot product of  $s $ and $ x $ modulo $ 2 $, i.e.,  $s \cdot x :=  s_{0}x_{0} + s_{1}x_{1}+ \ldots + s_{n-1}x_{n-1} \pmod 2 $.
	    \item  $ s(m) $ : For  $ s =  s_{n-1}\,s_{n-2}\,\ldots\, s_2\,s_1\,s_0 $ and $ 1 \leq m \leq n $, $ s(m) $ denotes the string formed by keeping only the $ m $ least significant bits of $ s $, i.e., $ s(m) = s_{m-1}\,s_{m-2}\,\ldots\, s_2\,s_1\,s_0 $. 
	    \item On a few occasions,  by abuse of notations, a non-negative integer $ s $, such that  $ s = \sum_{j=0}^{n-1} \, s_j 2^j $, will be used to represent the $ n $-bit string
	    $ s_{n-1}\,s_{n-2}\,\ldots\, s_2\,s_1\,s_0 $.   
	\end{itemize}

\section{Walsh-Hadamard transforms} \label{sec:Sequency-ordered-WH-transforms}
In this section, a brief review of the Walsh basis functions in both sequency and natural orderings will be provided. A review of Walsh-Hadamard transforms in sequency and natural orderings will also be presented. Subsequently, a quantum circuit for performing the Walsh-Hadamard transform in sequency ordering will be described. Interested readers may refer to \cite{beauchamp1975walsh} for further details on Walsh basis functions, Walsh-Hadamard transforms, and their applications. 

\subsection{Walsh basis functions in sequency and natural orderings} \label{Sec:WH_SN}
Walsh basis functions $ W_k (t) $ for $  k =0,~1,~2, ~\ldots~ N-1 $  in sequency order are defined as follows
\begin{align}
	W_0(t) &= 1 \quad \text{for } 0 \leq t \leq 1,  \nonumber\\
	W_{2k} (t) &= W_k(2t) + (-1)^k W_k (2t -1 ), \nonumber \\
	W_{2k+1} (t) &= W_k(2t) - (-1)^k W_k (2t -1 ), \nonumber\\
	W_k(t) &= 0 \quad \text{for } t < 0 \text{ and } t >1, \label{eq_def_Walsh}
\end{align}
where $ N $ is an integer of the form $ N = 2^n$. 
For $ N =8 $, Walsh basis functions in sequency order are shown in \mfig{fig_walsh_sequency}. 
The number of sign changes (or zero-crossings) for Walsh basis functions increases as the orders of the functions increase. It is worth noting that the number of sign changes varies in accordance with the order of the respective Walsh basis function, as illustrated in \mfig{fig_walsh_sequency} for the case of $N=8$.
For instance, the functions $W_0(t)$ to $W_7(t)$ exhibit a progressively increasing number of sign changes from $0$ to $7$, respectively. 
Analogous to the concept of frequency in Fourier analysis, the concept of sequency, as defined below, plays a fundamental role in signal analysis based on Walsh basis functions.

\begin{figure}
	\centering
	\begin{subfigure}{0.24\textwidth}
		\centering
		\includegraphics[width=\linewidth]{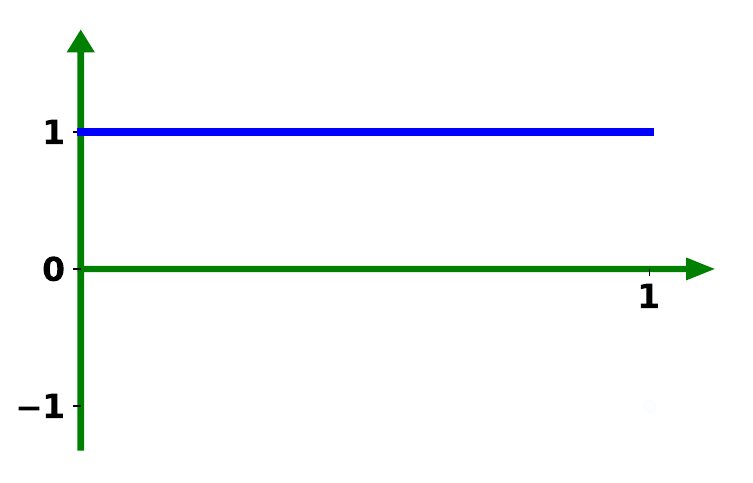}
		\caption{$W_0(t)$}
		\label{fig:WS0}
	\end{subfigure}
	\begin{subfigure}{0.24\textwidth}
		\centering
		\includegraphics[width=\linewidth]{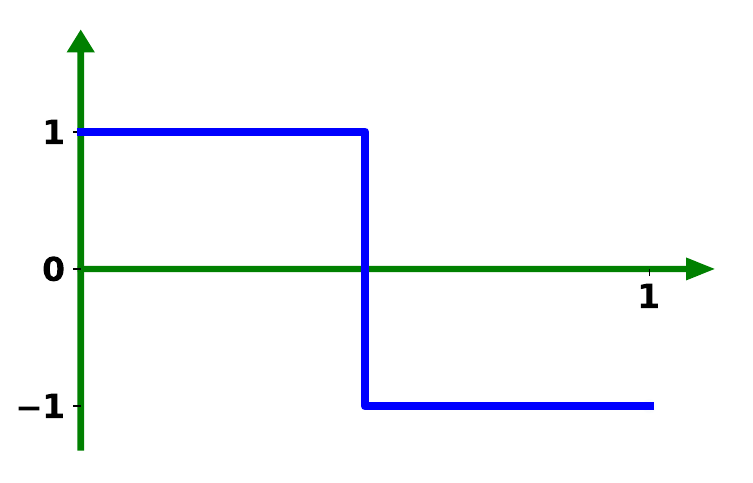}
		\caption{$W_1(t)$}
		\label{fig:WS1}
	\end{subfigure}
	\begin{subfigure}{0.24\textwidth}
		\centering
		\includegraphics[width=\linewidth]{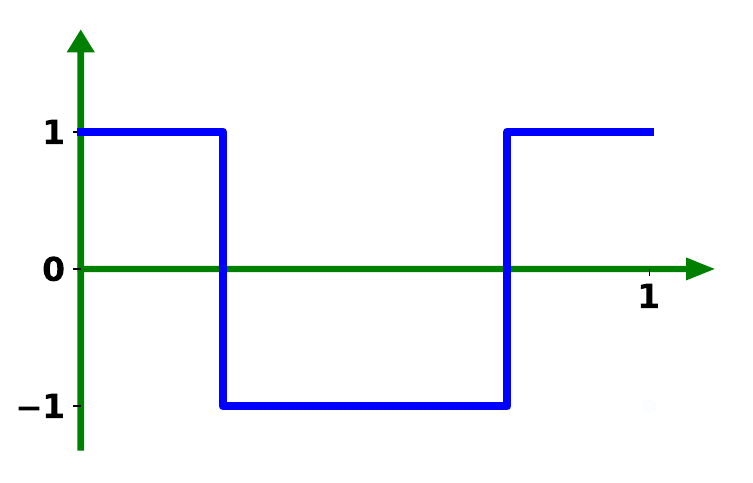}
		\caption{$W_2(t)$}
		\label{fig:WS2}
	\end{subfigure}
	\begin{subfigure}{0.24\textwidth}
		\centering
		\includegraphics[width=\linewidth]{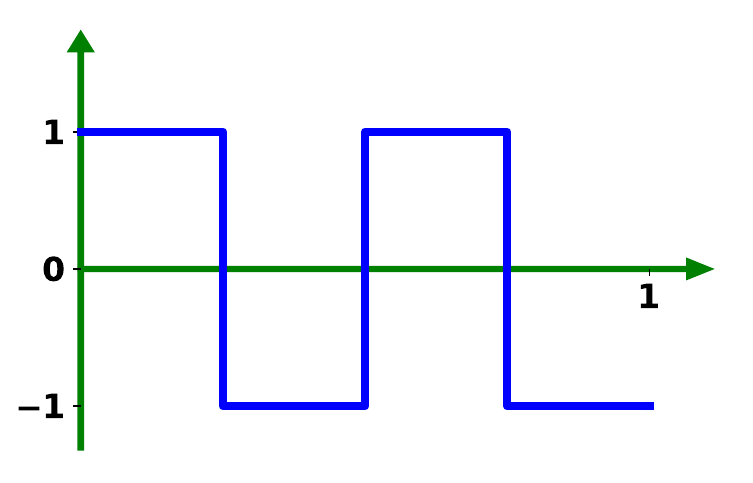}
		\caption{$W_3(t)$}
		\label{fig:WS3}
	\end{subfigure}
	\\
	\begin{subfigure}{0.24\textwidth}
		\centering
		\includegraphics[width=\linewidth]{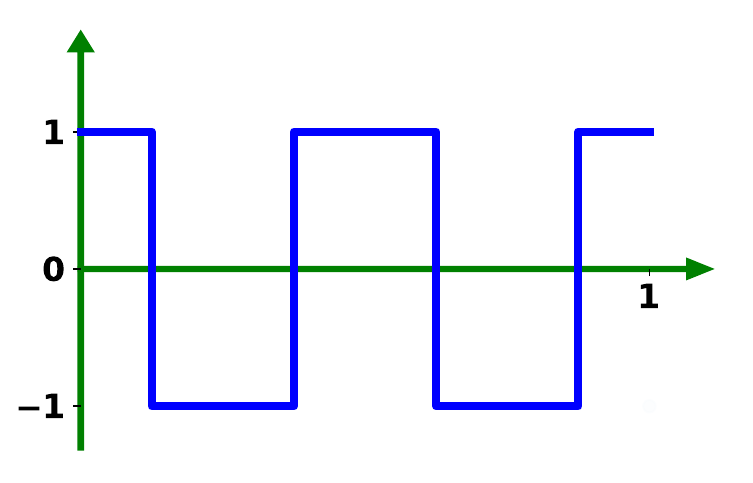}
		\caption{$W_4(t)$}
		\label{fig:WS4}
	\end{subfigure}
	\begin{subfigure}{0.24\textwidth}
		\centering
		\includegraphics[width=\linewidth]{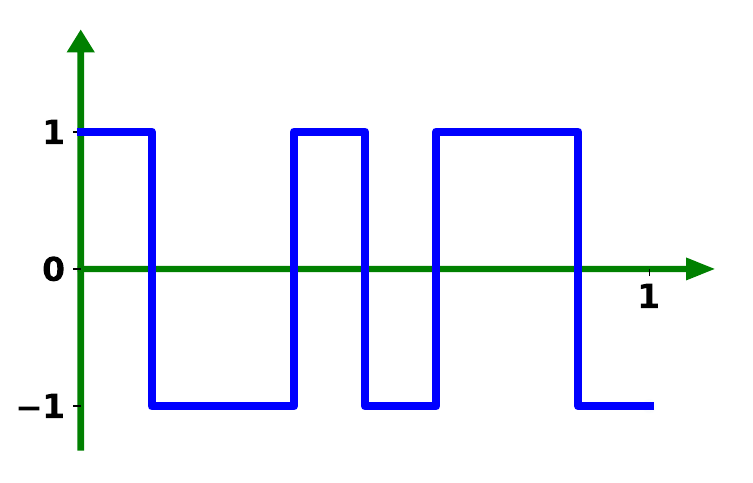}
		\caption{$W_5(t)$}
		\label{fig:WS5}
	\end{subfigure}
	\begin{subfigure}{0.24\textwidth}
		\centering
		\includegraphics[width=\linewidth]{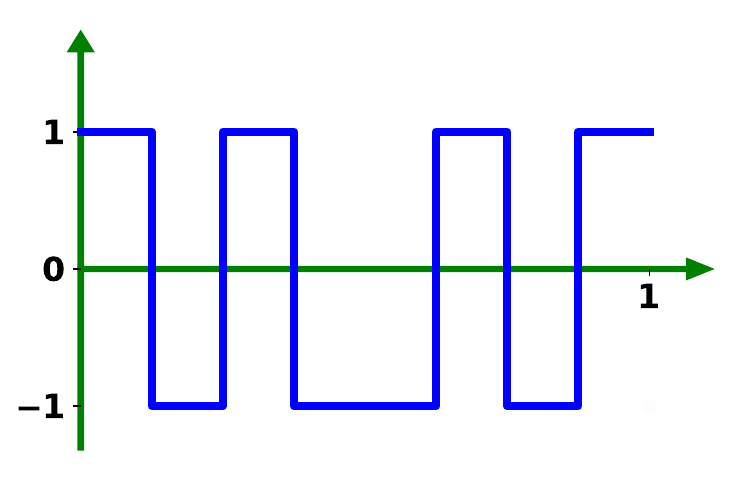}
		\caption{$W_6(t)$}
		\label{fig:WS6}
	\end{subfigure}
	\begin{subfigure}{0.24\textwidth}
		\centering
		\includegraphics[width=\linewidth]{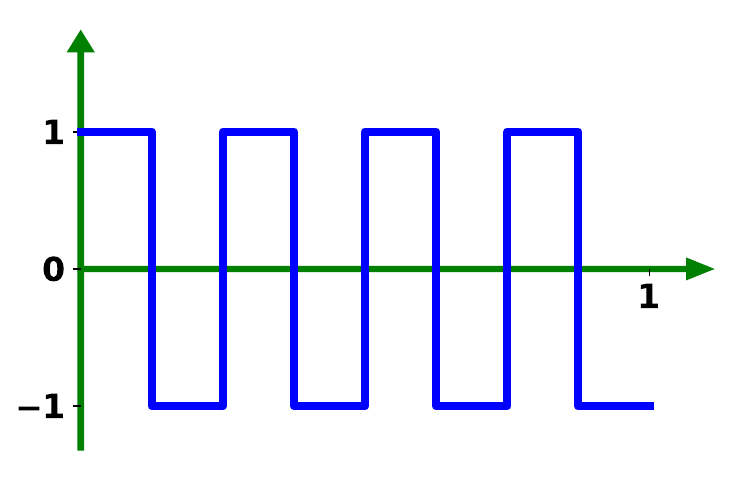}
		\caption{$W_7(t)$}
		\label{fig:WS7}
	\end{subfigure}
 \caption{Sequency-ordered Walsh basis functions for $N=8$.}
	\label{fig_walsh_sequency}
\end{figure}

\begin{defn} \label{def:seq}
The sequency for any vector 
		\begin{equation}\label{eq:def_sequence}
			 \bm{\mathcal{S}} = 	\left[\, F(0) \,\,\,  F(1) \,\,\, \cdots \cdots \,\,\, F(N-1) \,\right]^T,
		\end{equation}
  where $F(k) \in \{1,-1\}$ for all $k$, is 
  defined as 	\begin{equation}\label{Eq:defn_zero_crossings}
	\frac{1}{2} \sum_{k=0}^{N-2} \, \abs{   F(k+1) - F(k)}.
\end{equation}
We note that the sequency of the vector $ \bm{\mathcal{S}}$ represents the number of zero-crossings or sign changes in $ \bm{\mathcal{S}}$. We define the sequency of the row vector $ \bm{\mathcal{S}}^T$ to be the same as the sequency of the column vector $\bm{\mathcal{S}}$.
\end{defn}

A vector of length $ N $ can be obtained by sampling a Walsh basis function (ref.~\mref{Sec:fun_representation}). The sequency-ordered Walsh-Hadamard transform matrix  $H_N^S$ is obtained by arranging the vectors obtained from the sampling of Walsh basis functions as the columns of a matrix. The vectors are arranged in increasing order of sequency. The Walsh-Hadamard transform matrix of order $ N=8 $ in sequency order is
\begin{align*}
H^S_8 = 
	\frac{1}{\sqrt{8}} \,
	\begin{pmatrix*}[r]
		1 & 1 & 1 & 1 & 1 & 1 & 1 & 1  \\
		1 & 1 & 1 & 1 & -1 & -1 & -1 & -1  \\
		1 & 1 & -1 & -1 & -1 & -1 & 1 & 1  \\
		1 & 1 & -1 & -1 & 1 & 1 & -1 & -1  \\
		1 & -1 & -1 & 1 & 1 & -1 & -1 & 1  \\
		1 & -1 & -1 & 1 & -1 & 1 & 1 & -1  \\
		1 & -1 & 1 & -1 & -1 & 1 & -1 & 1  \\
		1 & -1 & 1 & -1 & 1 & -1 & 1 & -1  \\
	\end{pmatrix*}.
\end{align*} 
In contrast, the Walsh-Hadamard transform matrix of order $ N =2^n $ in natural order is given by $ H_N = H^{\otimes n}$,
where 
\begin{align*}
	H = \frac{1}{\sqrt{2}} \, \begin{pmatrix*}[r]
		1 & 1  \\
		1 & -1  \\
	\end{pmatrix*}.
\end{align*}

\begin{figure}
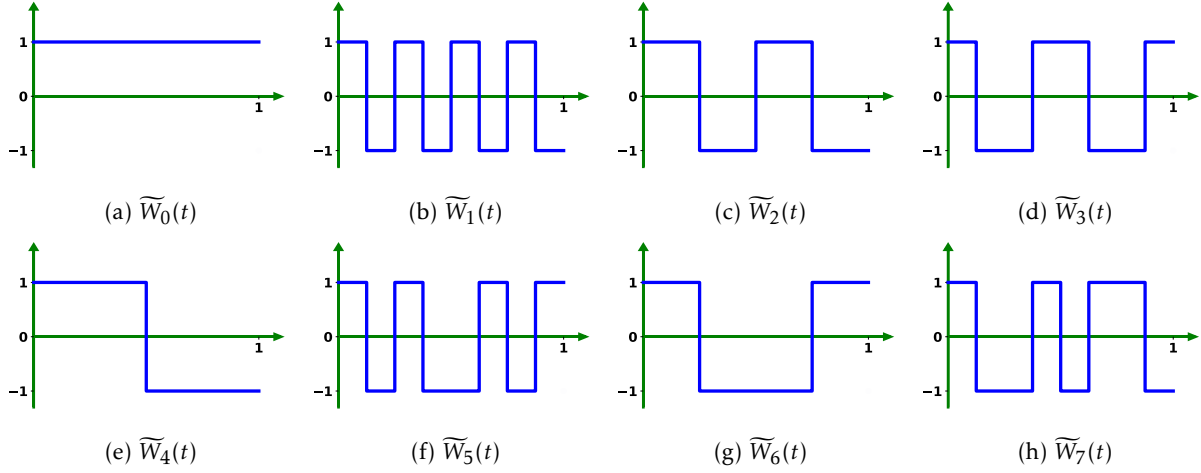

\centering
\begin{subfigure}{0.24\textwidth}
\centering
    \includegraphics[width=\linewidth]{W_N0.pdf}
    \caption{$\widetilde{W}_0(t)$}
    \label{fig:W0}
\end{subfigure}
\begin{subfigure}{0.24\textwidth}
\centering
    \includegraphics[width=\linewidth]{W_N1.pdf}
    \caption{$\widetilde{W}_1(t)$}
    \label{fig:W1}
\end{subfigure}
\begin{subfigure}{0.24\textwidth}
\centering
    \includegraphics[width=\linewidth]{W_N2.pdf}
    \caption{$\widetilde{W}_2(t)$}
    \label{fig:W2}
\end{subfigure}
\begin{subfigure}{0.24\textwidth}
\centering
    \includegraphics[width=\linewidth]{W_N3.pdf}
    \caption{$\widetilde{W}_3(t)$}
    \label{fig:W3}
\end{subfigure}
\\
\begin{subfigure}{0.24\textwidth}
\centering
    \includegraphics[width=\linewidth]{W_N4.pdf}
    \caption{$\widetilde{W}_4(t)$}
    \label{fig:W4}
\end{subfigure}
\begin{subfigure}{0.24\textwidth}
\centering
    \includegraphics[width=\linewidth]{W_N5.pdf}
    \caption{$\widetilde{W}_5(t)$}
    \label{fig:W5}
\end{subfigure}
\begin{subfigure}{0.24\textwidth}
\centering
    \includegraphics[width=\linewidth]{W_N6.pdf}
    \caption{$\widetilde{W}_6(t)$}
    \label{fig:W6}
\end{subfigure}
\begin{subfigure}{0.24\textwidth}
\centering
    \includegraphics[width=\linewidth]{W_N7.pdf}
    \caption{$\widetilde{W}_7 (t)$}
    \label{fig:W7}
\end{subfigure}
\caption{Walsh Hadamard basis functions in natural ordering for $N=8$.}
 \label{fig_walsh_natural}
\end{figure}

The Walsh-Hadamard matrix in natural order for $ N=8 $ is 
\begin{align*}
	H_8 =	\frac{1}{\sqrt{8}} \,
	\begin{pmatrix*}[r]
		1 & 1 & 1 & 1 & 1 & 1 & 1 & 1  \\
		1 & -1 & 1 & -1 & 1 & -1 & 1 & -1  \\
		1 & 1 & -1 & -1 & 1 & 1 & -1 & -1  \\
		1 & -1 & -1 & 1 & 1 & -1 & -1 & 1  \\
		1 & 1 & 1 & 1 & -1 & -1 & -1 & -1  \\
		1 & -1 & 1 & -1 & -1 & 1 & -1 & 1  \\
		1 & 1 & -1 & -1 & -1 & -1 & 1 & 1  \\
		1 & -1 & -1 & 1 & -1 & 1 & 1 & -1  \\
	\end{pmatrix*}.
\end{align*}
Walsh basis functions (or Walsh-Hadamard basis functions) in natural order can be obtained using the rows/columns of the Walsh-Hadamard matrix $ H_N$. 
Let  $\widetilde{{\bf W}}_{k}$  denote the $k$-th column of the matrix $H_N$. Then $\widetilde{{\bf W}}_{k}$ is the $k$-th Walsh basis vector in natural ordering. From this discrete basis vector, the corresponding Walsh basis function $\widetilde{W}_{k}(t)$ can be obtained. 
For $ N =8 $, Walsh basis functions in natural ordering are shown in \mfig{fig_walsh_natural}.

\begin{table}[]
	\centering
	\begin{tabular}{@{}c@{}}
		\toprule
		Index, $k$ \hspace{3cm} $ \widetilde{{\bf W}}_{k}^T $ \hspace{3cm} Sequency \\ 
     \midrule
  	$\begin{matrix*}[r]
		0 \hspace{2cm} [1 & 1 & 1 & 1 & 1 & 1 & 1 & 1 ] \hspace{2cm} 0 \\ 
		1 \hspace{2cm} [1 & -1 & 1 & -1 & 1 & -1 & 1 & -1]  \hspace{2cm} 7 \\ 
		2 \hspace{2cm} [1 & 1 & -1 & -1 & 1 & 1 & -1 & -1 ] \hspace{2cm} 3 \\ 
		3 \hspace{2cm} [1 & -1 & -1 & 1 & 1 & -1 & -1 & 1 ] \hspace{2cm} 4 \\ 
		4 \hspace{2cm} [1 & 1 & 1 & 1 & -1 & -1 & -1 & -1  ] \hspace{2cm} 1 \\ 
		5 \hspace{2cm} [1 & -1 & 1 & -1 & -1 & 1 & -1 & 1 ] \hspace{2cm} 6 \\ 
		6 \hspace{2cm} [1 & 1 & -1 & -1 & -1 & -1 & 1 & 1  ] \hspace{2cm} 2\\ 
		7 \hspace{2cm} [1 & -1 & -1 & 1 & -1 & 1 & 1 & -1  ] \hspace{2cm} 5
	\end{matrix*}$ 
    \\ \bottomrule
	\end{tabular}
  \caption{Sequency (or number of sign changes) for Walsh basis vectors in natural ordering  $\widetilde{{\bf W}}_{k} $ for $k=0$ to $k=7$. The middle column shows the transpose of each of the Walsh basis vectors in natural ordering.}
	\label{tab:zero-crosings-table}
\end{table}

Table \ref{tab:zero-crosings-table} shows the sequencies (ref.~\ref{def:seq}) for the Walsh basis vectors $\widetilde{{\bf W}}_{k} $  (in natural ordering) for $k=0$ to $k=7$. The second column in the table shows the rows of the matrix $H_8$  (which represent the Walsh basis vectors in natural ordering  $\widetilde{{\bf W}}_{k}^T $ for $k=0$ to $k=7$ for $N=8$). The last column shows the number of sign changes in the corresponding  $\widetilde{{\bf W}}_{k} $ (i.e., the sequency of $\widetilde{{\bf W}}_{k} $).  

Walsh-Hadamard transforms in natural and sequency orders can also be defined in terms of their actions on the computational basis vectors. 
Let $ N=2^n $ be a positive integer. Let $ V$ be the $ N $ dimensional complex vector space generated by the computational basis states $ \{ \ket{0}, \, \ket{1}, \, \ldots \,,\, \ket{N-1} \} $.  We note that the Walsh-Hadamard transform in natural order can be defined as a linear transformation $ H_N : V \to V $ such that the action of $ H_N = H^{\otimes n}$ on the computational basis state $ \ket{j} $, with $ 0 \leq j \leq N-1 $ is given by
\begin{equation}\label{eq:natual:hadamard}
	H_N \, \ket{j} = \frac{1}{\sqrt{N}} \sum_{k=0}^{N-1} \, (-1)^{k \cdot j} \, \ket{k}. 
\end{equation}
Here $ k \cdot j $ denotes the bit-wise dot product of $ k $ and $ j$. 
It is clear that the $(k,j)$-th element of the corresponding matrix of this transformation will be given by $(-1)^{k \cdot j}$.
It can be checked that this matrix is symmetric. Further, this matrix is the same as the Walsh-Hadamard transform matrix $H_N$ and the vector formed by the entries in the $s$-th column 
is the $s$-th Walsh basis vector in natural order $\widetilde{{\bf W}}_{s}$. Equivalently,  
\begin{align} \label{eq:defWS}
  \widetilde{{\bf W}}_{s}  =  [F(0),\, F(1),\, \ldots\, ,\, F(N-1)]^T, 
\end{align}
where $F(k) = (-1)^{k \cdot s}$, for $k=0,\, 1, \ldots\,,\, N-1$.

Next, we note that the Walsh-Hadamard transform in sequency order can be defined as a linear transformation $ H^S_N: V \to V $ acting on the basis state $ \ket{j} $, with $ 0 \leq j \leq N-1 $, as follows (see \cite{beauchamp1975walsh}). 
\begin{equation}\label{eq:sequency:hadamard}
	H_N^S \, \ket{j} = \frac{1}{\sqrt{N}} \sum_{k=0}^{N-1} \, (-1)^{ \sum_{r=0}^{n-1} \, k_{n-1-r} (j_r \oplus j_{r+1}) } \, \ket{k},
\end{equation}
where $ k = k_{n-1}\,k_{n-2}\,\ldots \, k_1\, k_0 $ and $ j =  j_{n-1}\,j_{n-2} \,\ldots \, j_1\, j_0 $, are binary representations of $ k $ and $ j $ respectively, with $ k_i,\, j_i \in \{0,1\} $ for $ i=0,\,1,\, \ldots ,\,n-1$, and $ j_{n} = 0 $.

\section{Walsh series representation of a signal} \label{Sec:fun_representation}

Any piecewise continuous complex-valued function (or signal) $f(t)$ defined on the interval $[0,1]$ can be approximated through a two-step process: first by discretizing the function and then expressing the discretized version as a linear combination of Walsh basis functions. 
One common method of discretizing $f(t)$ on $[0,1]$ is as follows: the interval $[0,1]$ is uniformly divided into $N$ subintervals, and then the function $f(t)$ is discretized by sampling it at the midpoints of these subintervals. This process yields a vector ${\bf{f}} = [f_0, \, f_1,\, \ldots, f_{N-1} ]^T$, where $f_k = f(t_k)$ and $t_k= \frac{2k+1}{2N}$ represents the midpoint of the $k$-th subinterval.
By following this approach, one can approximate the original function $f(t)$ using the discrete vector ${\bf{f}}$ and then express this discretized version as a linear combination of discretized Walsh basis functions. In the following, we will consider the discretized version ${\bf W}_k$ of the Walsh basis function $ W_k (t) $  as defined in \eqref{eq_def_Walsh}. We note that ${\bf W}_k$ is the $k$-th column of the Walsh-Hadamard matrix $H^S_N$ (ref.~\meqref{eq:sequency:hadamard}) in sequency order.
       
        Let us consider the complex vector space $\CC^N $. 
        Let $  \langle \, \cdot \, | \, \cdot \,  \rangle $ be the (Hermitian) inner product on $\CC^N $ defined as
        \begin{align}
             \langle \, {\bf{f}} \, | \, {\bf{g}}  \, \rangle  = \frac{1}{\sqrt{N}} \sum_{k=0}^{N-1} \, \overline{f_k} \,  g_k.
        \end{align}
        Two distinct Walsh basis functions are orthogonal and satisfy the following.
        \begin{align}
            \langle {\bf{W_j}} | {{\bf W}_k} \rangle = \sqrt{N} \delta_{jk},
        \end{align}
        where   
        $$
\delta_{jk}=\begin{cases}
			1, & \text{if $j=k$}\\
            0, & \text{otherwise.}
		 \end{cases}
$$
The vector  $\bf{f}$, which is a discretized version of the function $f(t)$, can be expressed as a linear combination of discretized Walsh basis functions as follows
\begin{align} \label{Eq:Walsh_Fourier_Series}
     {\bf{f}}  = \frac{1}{\sqrt{N}} \,  \sum_{k=0}^{N-1} \hat{f}_k \, {\bf W}_k,  \end{align}
where  $\hat{f}_k =  \langle {\bf W}_k | \bf{f} \rangle$.
The vector $ {\hat{\bf{f}}} = [\hat{f}_0,\, \hat{f}_1,\, \ldots,\, \hat{f}_{N-1}]^{T}$ can be computed by performing the Walsh-Hadamard transform (in sequency order) as follows
       \begin{align}
   {\hat{\bf{f}}} = H_N^S \bf{f}.
       \end{align}
We observe that \meqref{Eq:Walsh_Fourier_Series} provides the Walsh series representation of the discrete signal $\bf{f}$. The coefficient $\hat{f}_k$ contains information about the contribution of the $k$-th sequency component $\bf{W}_k$ in the discrete signal $\bf{f}$. The vector $ {\hat{\bf{f}}} = [\hat{f}_0,\, \hat{f}_1,\, \ldots,\hat{f}_{N-1}]^{T}$ represents the sequency spectrum of $\bf{f}$. Similar to how a time-domain signal is analyzed using its frequency spectrum in Fourier analysis, one can utilize the sequency spectrum to analyze the signal via the Walsh analysis approach.

\begin{remark}
    The above formulation of the Walsh series representation of signals is similar to the formulation using Fourier basis functions. Analogous to \meqref{Eq:Walsh_Fourier_Series}, the vector ${\bf f}$ can be expressed as
\begin{align} \label{Eq:Fourier_Series}
     {\bf{f}}  = \frac{1}{\sqrt{N}} \,  \sum_{k=0}^{N-1} \hat{F}_k \, {\bf U}_k,  \end{align}
where  ${\bf U}_k = \left[ \left. \exp \left( i\frac{2 \pi  }{N} kn\right) \,\, \right|  \,  n = 0,\, 1,\,\ldots\, , \, N-1  \right]^T $ and $\hat{F}_k =  \langle {\bf U}_k | \bf{f} \rangle$.
The vector $ {\hat{\bf{F}}} = [\hat{F}_0,\, \hat{F}_1,\, \ldots,\, \hat{F}_{N-1}]^{T}$ is the discrete Fourier transform of the vector ${\bf f}$.
\end{remark}

\begin{figure}[htp]
	\centering
	\begin{subfigure}{0.32\textwidth}
		\centering
		\includegraphics[width=\linewidth]{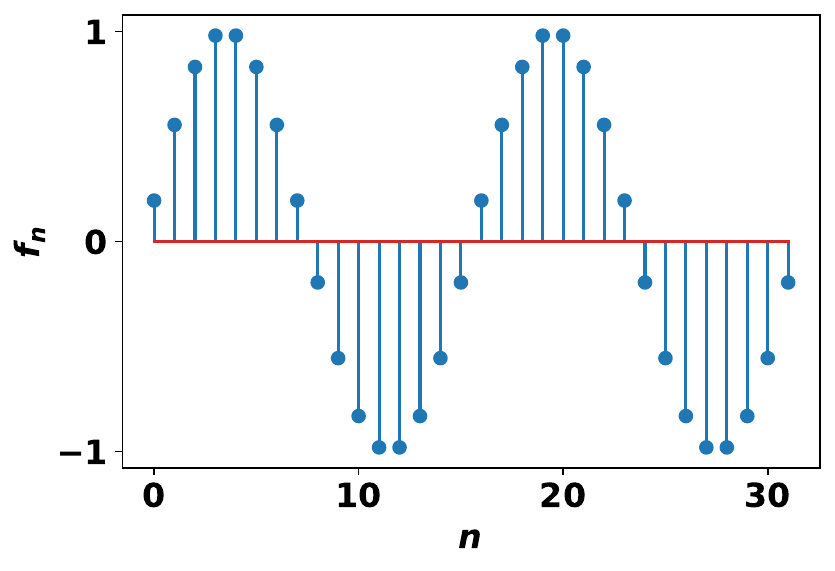}
		\caption{Signal: Sinusoidal}
		\label{Sig:1}
	\end{subfigure}
	\begin{subfigure}{0.32\textwidth}
		\centering
		\includegraphics[width=\linewidth]{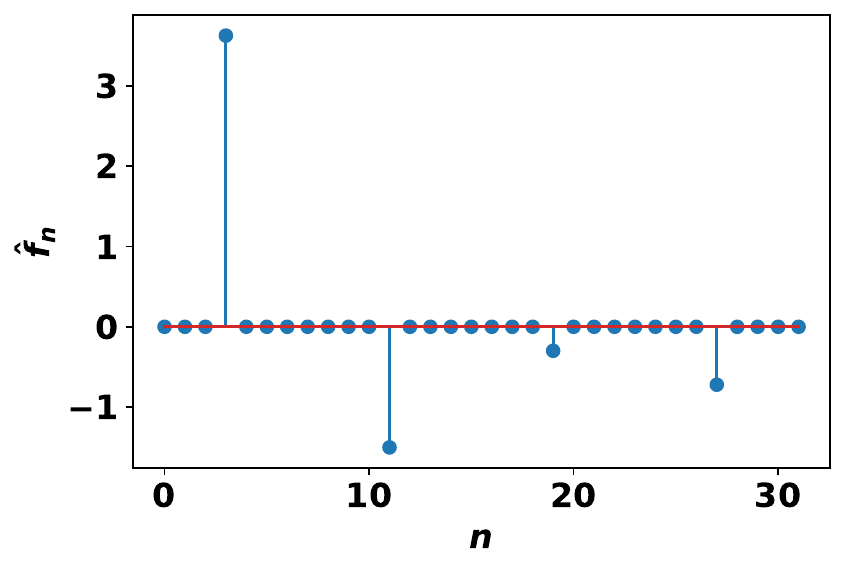}
		\caption{Sequency spectrum}
		\label{Seq:1}
	\end{subfigure}
	\begin{subfigure}{0.32\textwidth}
		\centering
		\includegraphics[width=\linewidth]{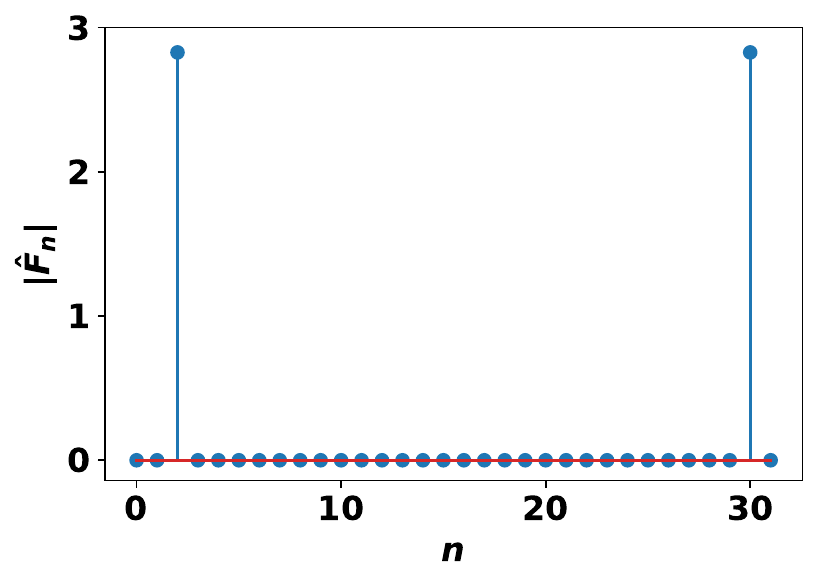}
		\caption{Frequency spectrum}
		\label{Freq:1}
	\end{subfigure}
		\\ 
  \vspace{0.15cm}
	\begin{subfigure}{0.32\textwidth}
		\centering
		\includegraphics[width=\linewidth]{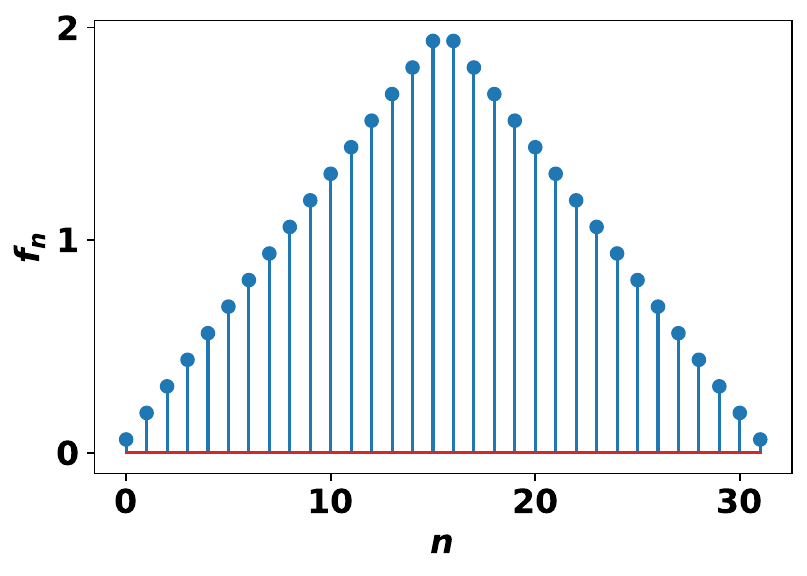}
		\caption{Signal: Triangular}
		\label{Sig:2}
	\end{subfigure}
	\begin{subfigure}{0.32\textwidth}
		\centering
		\includegraphics[width=\linewidth]{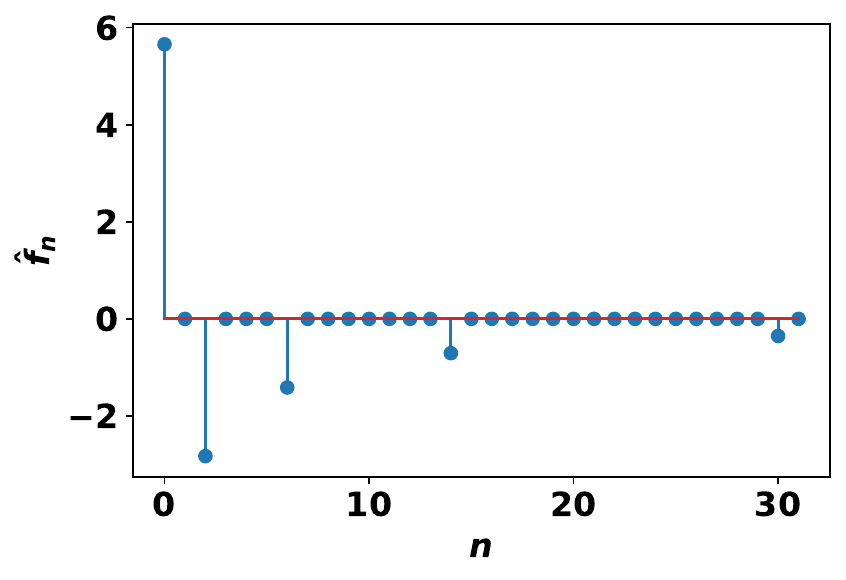}
		\caption{Sequency spectrum}
		\label{Seq:2}
	\end{subfigure}
	\begin{subfigure}{0.32\textwidth}
		\centering
		\includegraphics[width=\linewidth]{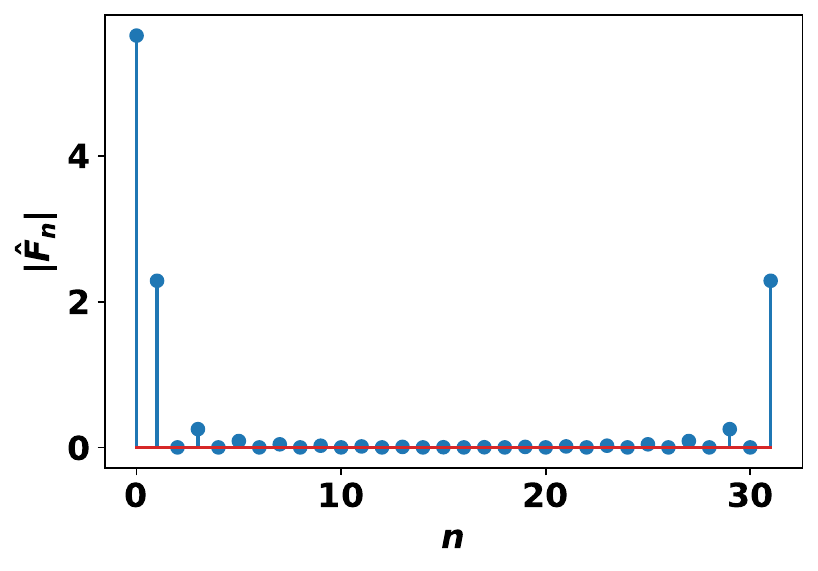}
		\caption{Frequency spectrum}
		\label{Freq:2}
	\end{subfigure}
 \\
   \vspace{0.15cm}
	\begin{subfigure}{0.32\textwidth}
		\centering
		\includegraphics[width=\linewidth]{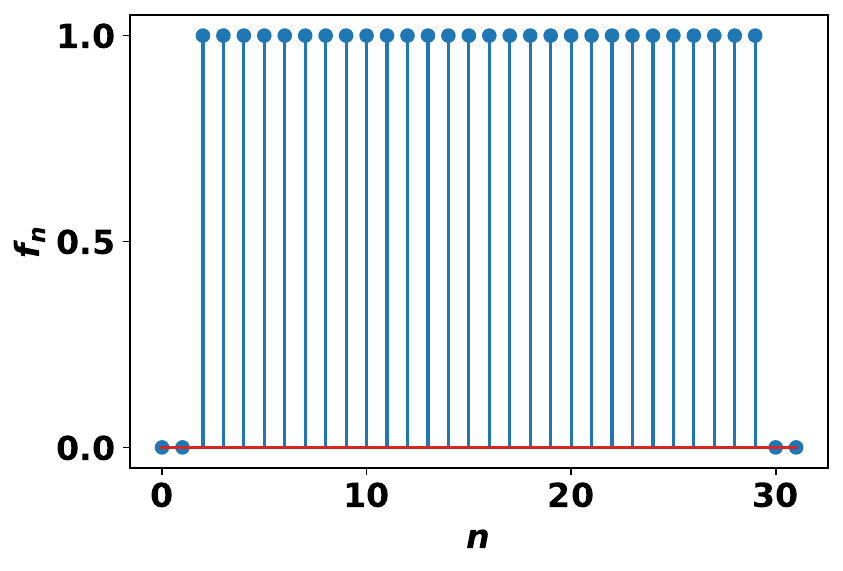}
		\caption{Signal: Rectangular}
		\label{Sig:3}
	\end{subfigure}
	\begin{subfigure}{0.32\textwidth}
		\centering
		\includegraphics[width=\linewidth]{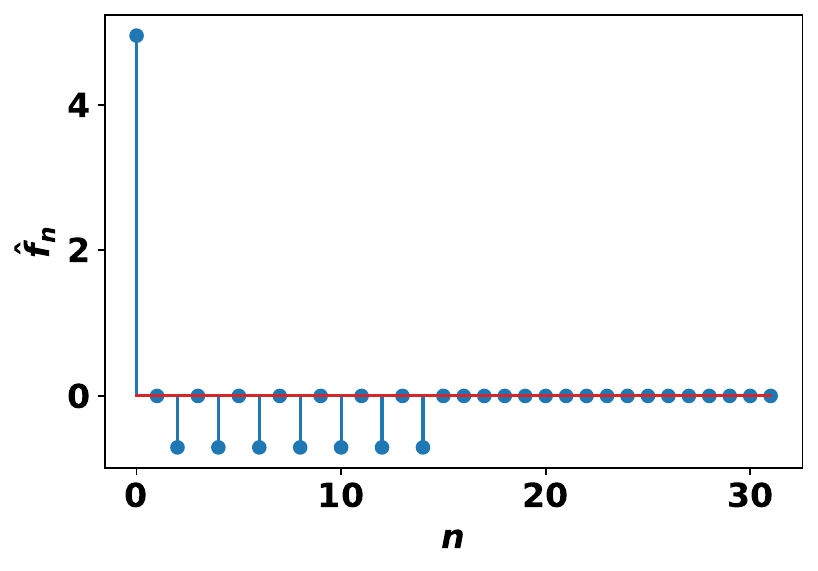}
		\caption{Sequency spectrum}
		\label{Seq:3}
	\end{subfigure}
	\begin{subfigure}{0.32\textwidth}
		\centering
		\includegraphics[width=\linewidth]{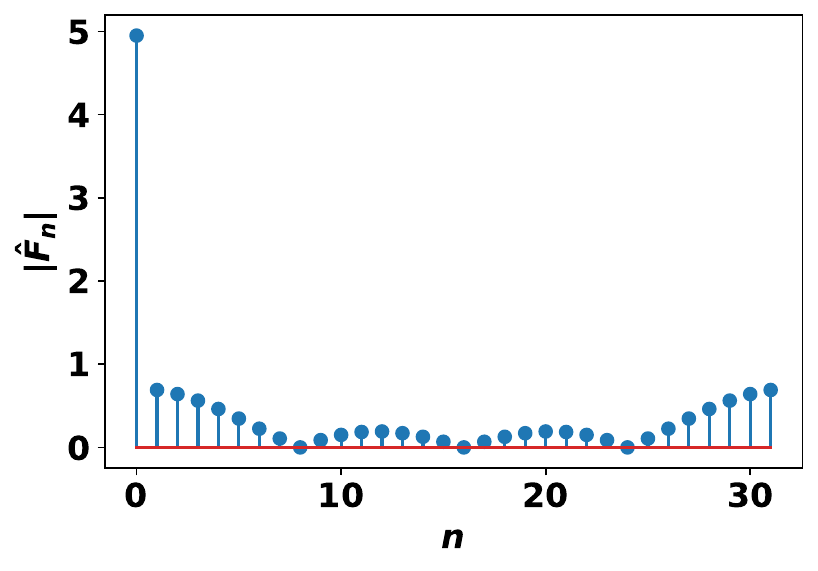}
		\caption{Frequency spectrum}
		\label{Freq:3}
	\end{subfigure}
 \\
   \vspace{0.15cm}
	\begin{subfigure}{0.32\textwidth}
		\centering
		\includegraphics[width=\linewidth]{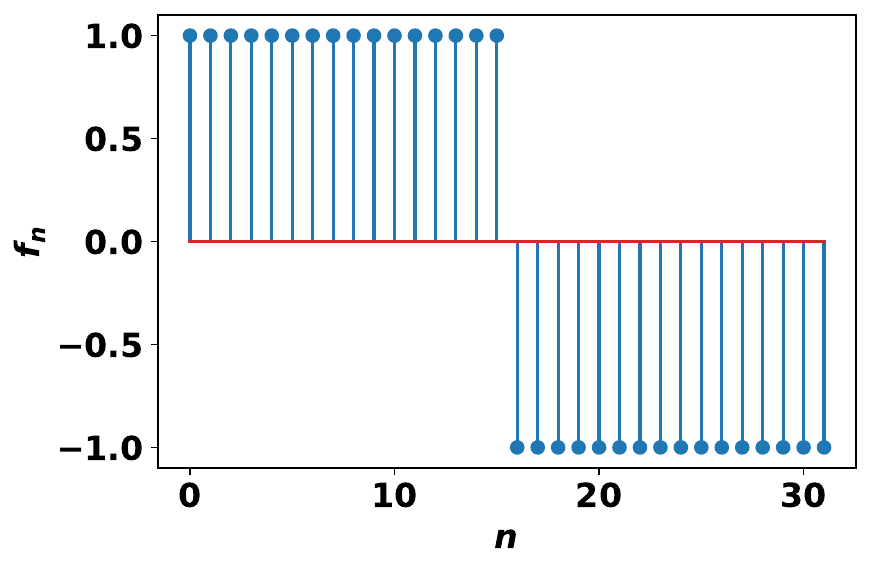}
		\caption{Signal: Square}
		\label{Sig:4}
	\end{subfigure}
	\begin{subfigure}{0.32\textwidth}
		\centering
		\includegraphics[width=\linewidth]{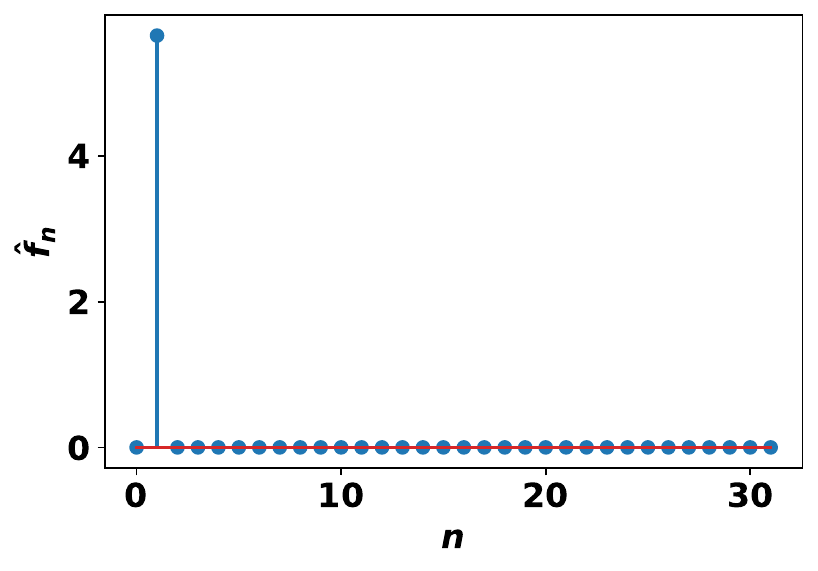}
		\caption{Sequency spectrum}
		\label{Seq:4}
	\end{subfigure}
	\begin{subfigure}{0.32\textwidth}
		\centering
		\includegraphics[width=\linewidth]{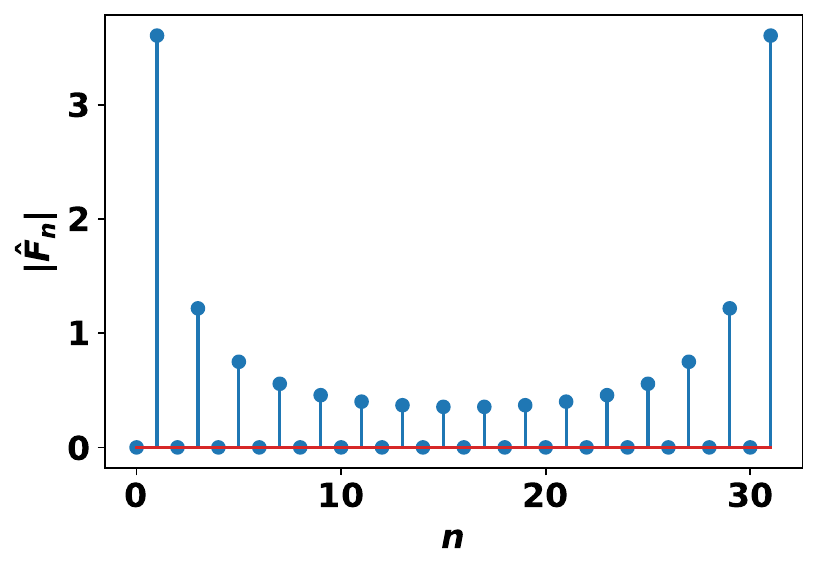}
		\caption{Frequency spectrum}
		\label{Freq:4}
	\end{subfigure}

 \caption{
 Signal waveforms (sinusoidal, triangular, rectangular, square) with corresponding sequency and Fourier spectra (along each row).
}
	\label{fig_sequency_spectrum}
\end{figure}

\begin{figure}[htp]
	\centering
	\begin{subfigure}{0.32\textwidth}
		\centering
		\includegraphics[width=\linewidth]{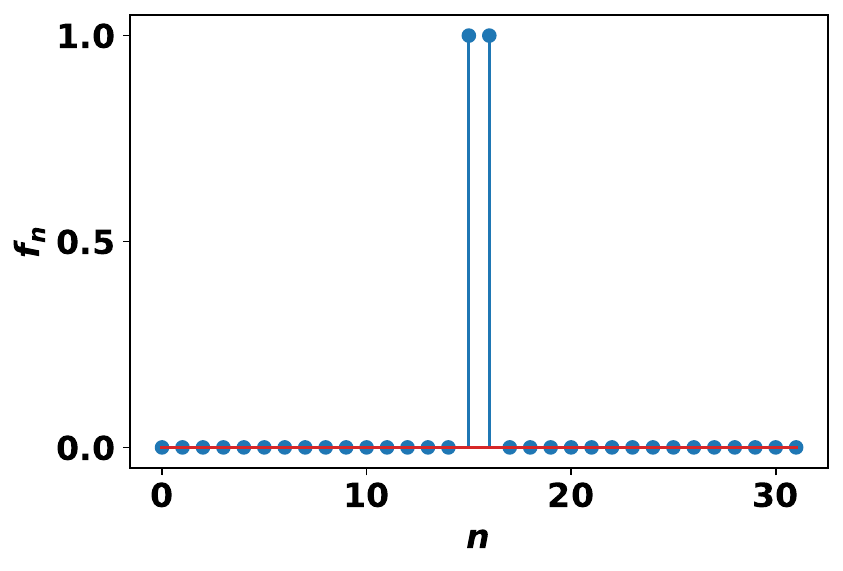}
		\caption{Signal}
		\label{fig:1}
	\end{subfigure}
	\begin{subfigure}{0.32\textwidth}
		\centering
		\includegraphics[width=\linewidth]{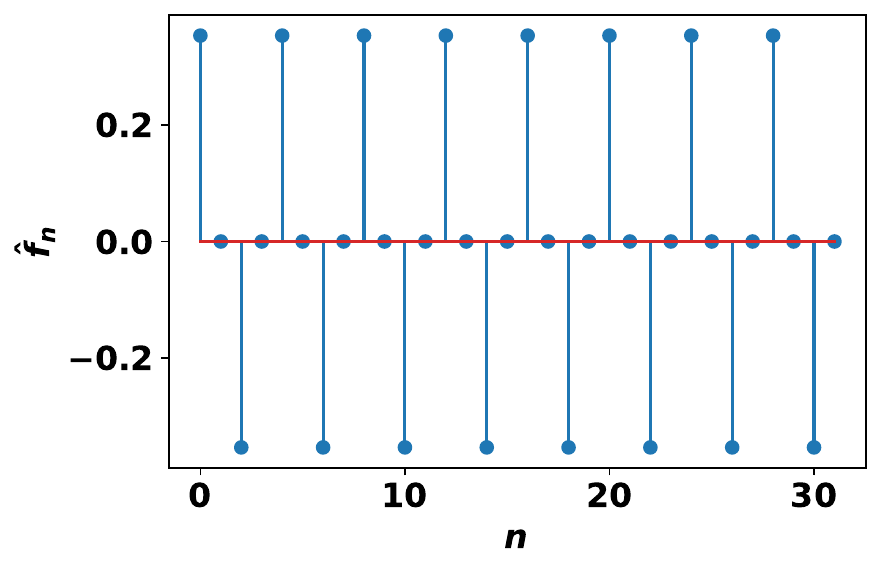}
		\caption{Sequency spectrum}
		\label{fig:2}
	\end{subfigure}
	\begin{subfigure}{0.32\textwidth}
		\centering
		\includegraphics[width=\linewidth]{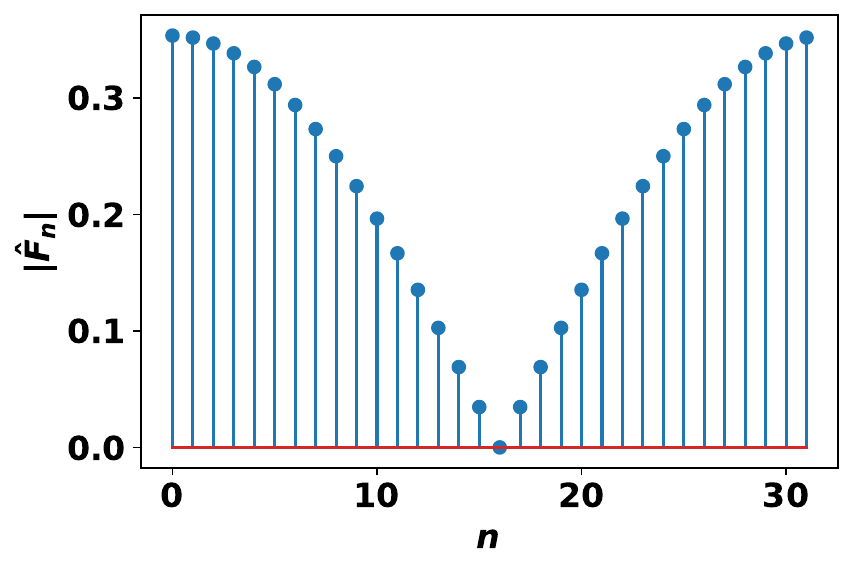}
		\caption{Frequency Spectrum}
		\label{fig:3}
	\end{subfigure}
		\\
    \vspace{0.15cm}
	\begin{subfigure}{0.32\textwidth}
		\centering
		\includegraphics[width=\linewidth]{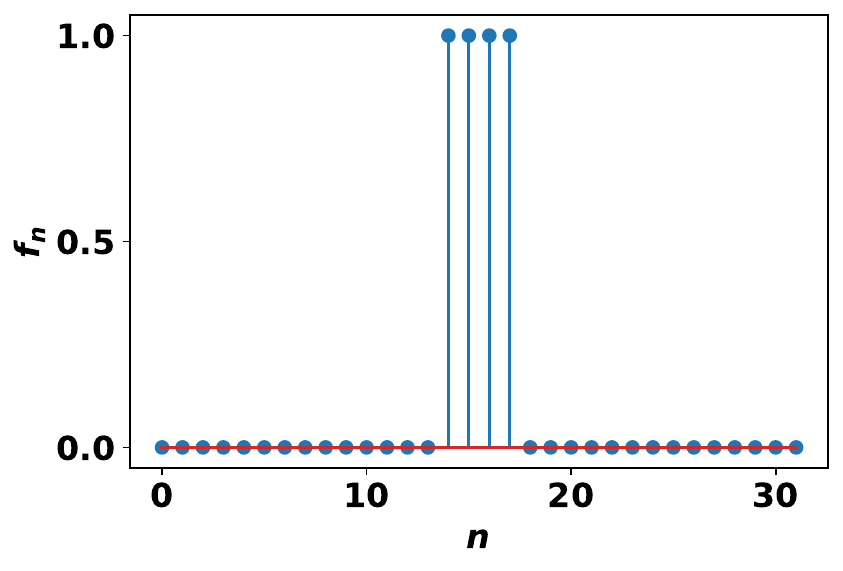}
		\caption{Signal}
		\label{fig:4}
	\end{subfigure}
	\begin{subfigure}{0.32\textwidth}
		\centering
		\includegraphics[width=\linewidth]{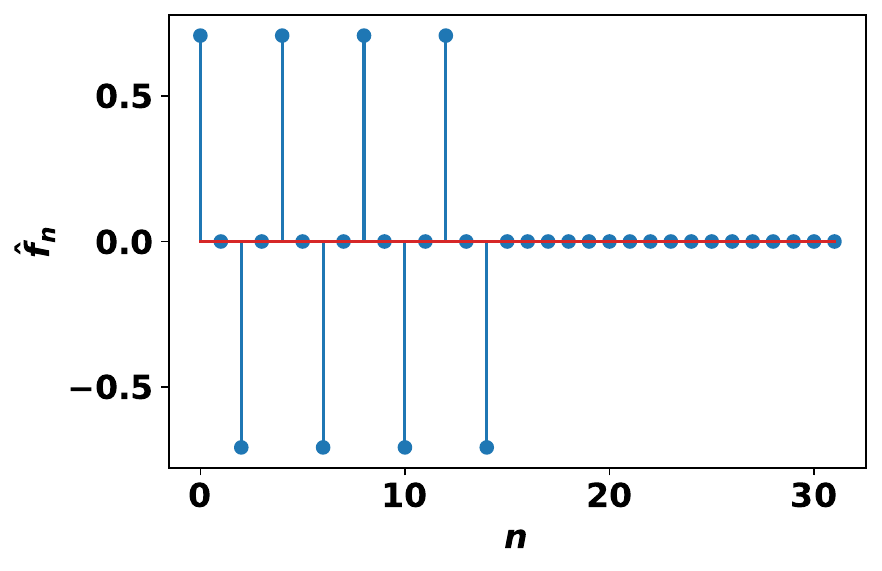}
		\caption{Sequency spectrum}
		\label{fig:5}
	\end{subfigure}
	\begin{subfigure}{0.32\textwidth}
		\centering
		\includegraphics[width=\linewidth]{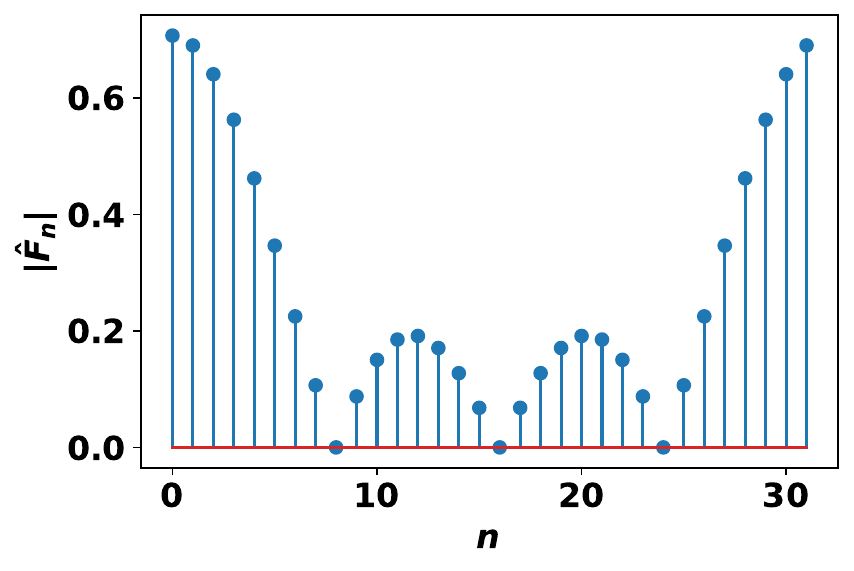}
		\caption{Frequency Spectrum}
		\label{fig:6}
	\end{subfigure}
 \\
   \vspace{0.15cm}
	\begin{subfigure}{0.32\textwidth}
		\centering
		\includegraphics[width=\linewidth]{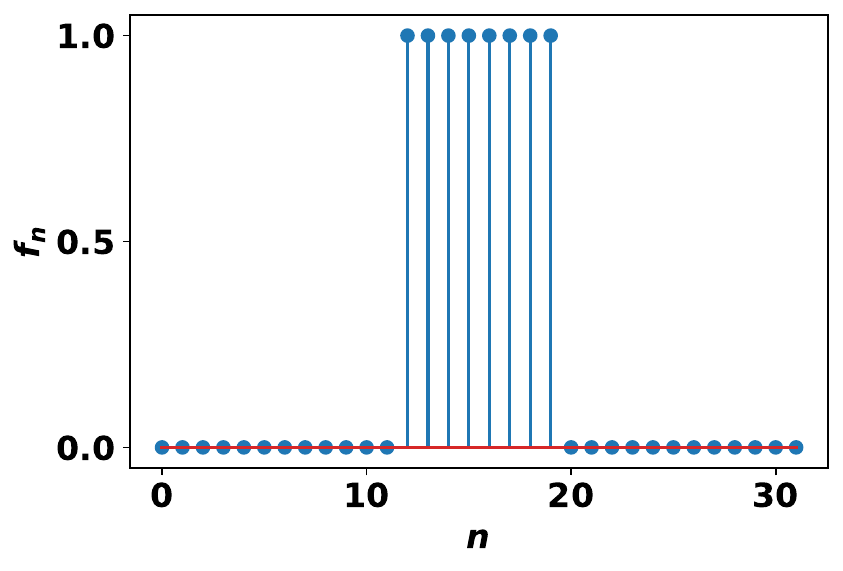}
		\caption{Signal}
		\label{fig:7}
	\end{subfigure}
	\begin{subfigure}{0.32\textwidth}
		\centering
		\includegraphics[width=\linewidth]{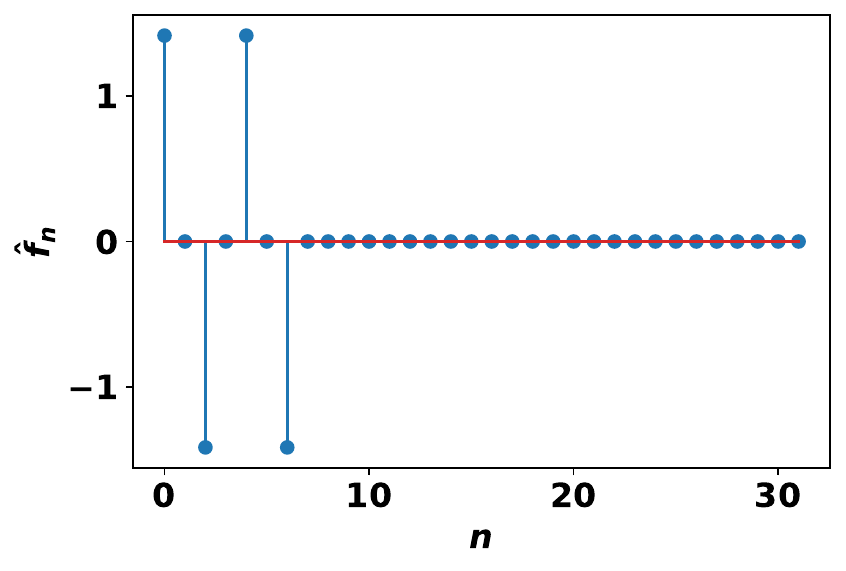}
		\caption{Sequency spectrum}
		\label{fig:8}
	\end{subfigure}
	\begin{subfigure}{0.32\textwidth}
		\centering
		\includegraphics[width=\linewidth]{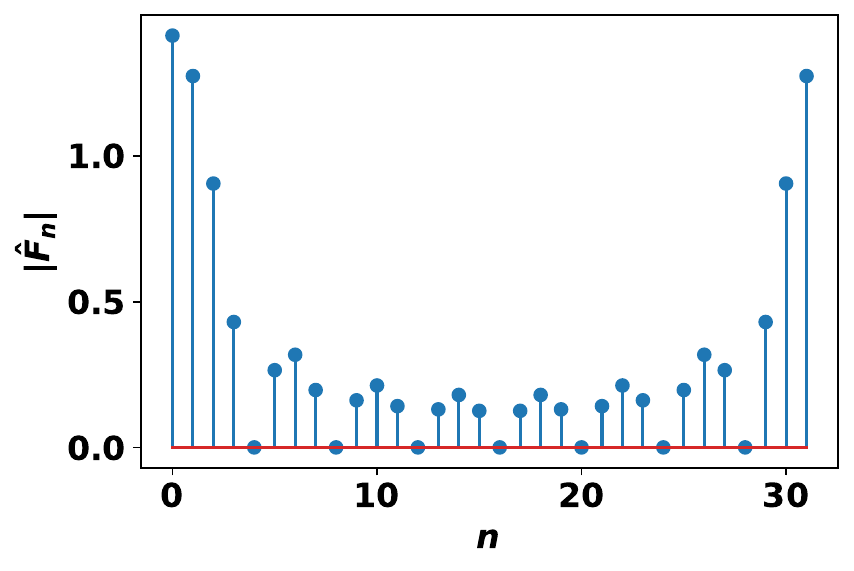}
		\caption{Frequency Spectrum}
		\label{fig:9}
	\end{subfigure}
  \\
    \vspace{0.15cm}
	\begin{subfigure}{0.32\textwidth}
		\centering
		\includegraphics[width=\linewidth]{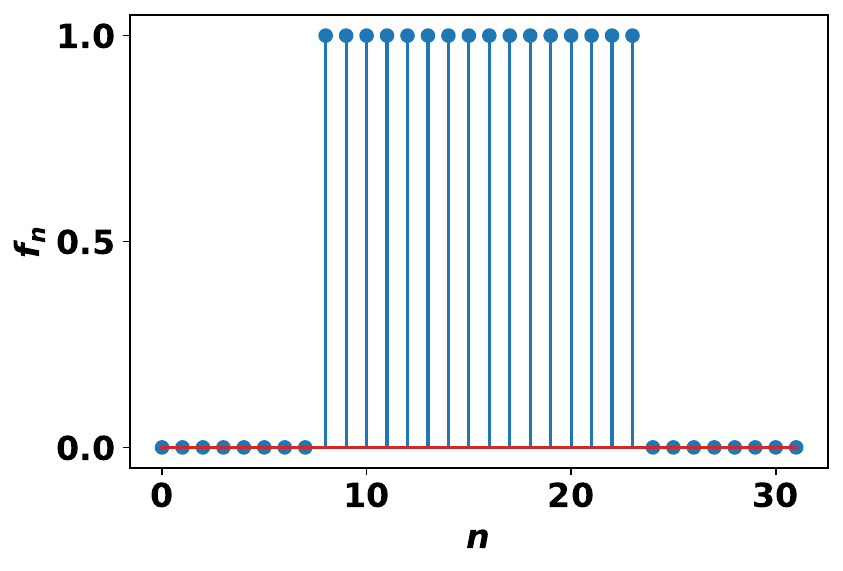}
		\caption{Signal}
		\label{fig:10}
	\end{subfigure}
	\begin{subfigure}{0.32\textwidth}
		\centering
		\includegraphics[width=\linewidth]{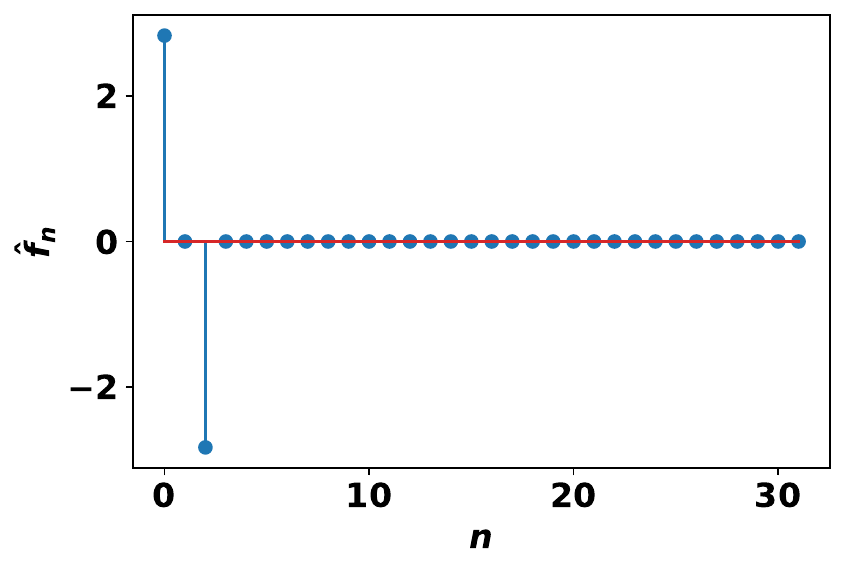}
		\caption{Sequency spectrum}
		\label{fig:11}
	\end{subfigure}
	\begin{subfigure}{0.32\textwidth}
		\centering
		\includegraphics[width=\linewidth]{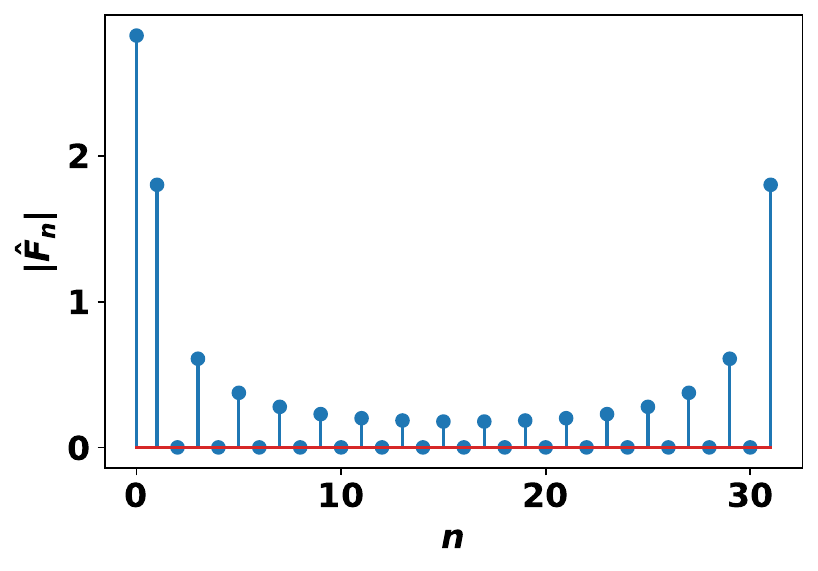}
		\caption{Frequency Spectrum}
		\label{fig:12}
	\end{subfigure}
  \caption{
  Rectangular pulses of different widths with corresponding sequency and Fourier spectra (along each row). 
}
	\label{fig_sequency_spectrum_rect_pulse}
\end{figure}

In \mfig{fig_sequency_spectrum}, several signals (ref.~Sub-figures \ref{Sig:1}, \ref{Sig:2}, \ref{Sig:3} and \ref{Sig:4}) and their corresponding sequency spectra (ref.~Sub-figures \ref{Seq:1}, \ref{Seq:2}, \ref{Seq:3} and \ref{Seq:4}) and frequency spectra (ref.~Sub-figures \ref{Freq:1}, \ref{Freq:2}, \ref{Freq:3} and \ref{Freq:4}) are shown for comparison. We observe from these figures that the Fourier series is particularly well suited for representing periodic functions,  which have continuous and smooth waveforms (ref.~Sub-figures \ref{Sig:1}, \ref{Seq:1}, \ref{Freq:1}). On the other hand, Walsh basis functions are more suitable for representing signals with sharp transitions and discontinuities, such as square waves (ref.~Sub-figures \ref{Sig:4}, \ref{Seq:4}, \ref{Freq:4}).

Rectangular pulses of different widths are shown in \mfig{fig_sequency_spectrum_rect_pulse}, along with corresponding sequency and Fourier spectra (along each row). Note that as the width of the rectangular pulse increases, its representation in the sequency domain requires fewer sequency components. It is clear from \mfig{fig_sequency_spectrum_rect_pulse} that for rectangular pulses, the signal representation in sequency domain is considerably simpler than that in the frequency domain (based on Fourier analysis).

\section{Quantum circuit for sequency ordered Walsh-Hadamard transforms} \label{Sec:seq_WH}

\begin{figure}
	\centering
	\includegraphics{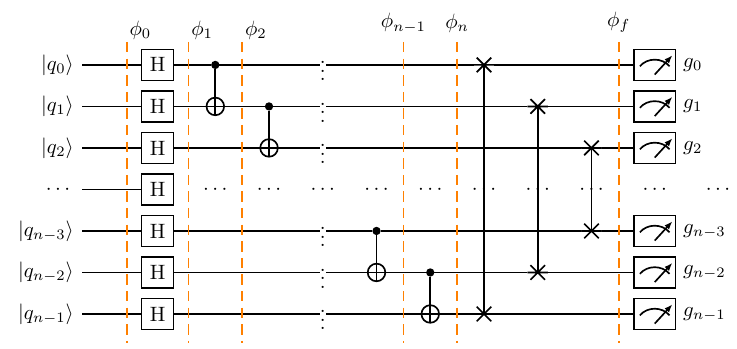}
	\caption{Quantum circuit for computing the Walsh-Hadamard transform in sequency ordering.} \label{fig:sequency}
\end{figure}
We note that some of the discussions in this section, including Lemma \ref{sec:Lemma}, are based on \cite{shukla2022quantum}.
A schematic quantum circuit for performing the Walsh-Hadamard transform in sequency order is shown in \mfig{fig:sequency}. 
In order to show that the quantum circuit in \mfig{fig:sequency} computes the Walsh-Hadamard transform in sequency ordering, we will compute the output when the input is a computational basis state $ \ket{\phi_0} =  \ket{j} $, with $ 0 \leq j  \leq N-1 $ and $ N =2^n $. 
We have 
\begin{align}
	\ket{\phi_1 } = H^{\otimes n} \ket{\phi_0} &=  \frac{1}{\sqrt{N}}  \sum_{s=0}^{N-1}   \, (-1)^{j \cdot s} \, \ket{s},  \nonumber \\
	&= \frac{1}{\sqrt{N}}  \sum_{s=0}^{N-1}   \,  \, (-1)^{ \sum_{k=0}^{n-1} \, j_k s_k} \, \ket{s}. \label{eq:phif}
\end{align}

Let us denote the action of the remaining part of the quantum circuit (i.e., between the barriers labeled $\phi_1$ to $\phi_f$) by $U_Z$. It will be shown in \mref{Sec:UZ} that for any computational basis state $\ket{s}$, the action of $U_z$ results in  $ U_z \ket{s} = \ket{g} $, where $ g = \sum_{k=0}^{n-1} \, g_k 2^k $ and 
$ g_{n-1} = s_0 $, $  g_k = \, s_0 \, \oplus \, s_1 \, \oplus \, \ldots   \, \oplus \, s_{n-1-k}$, for $ k = n-2, \, \ldots, \, 1,\, 0 $ (see  \meqref{eq:gk}  and \meqref{eq:g}). 
Further, it is easy to see that 
\begin{equation}\label{eq:stog}
	s_k = g_{n-k} \, \oplus \, g_{n-k-1}, \quad k=0,\, 1,\, \ldots,\, n-1,  
\end{equation}
where it is assumed that $ g_n = 0 $. Also, the map sending $ \ket{s} $ to $ \ket{g} $ defined in (see \meqref{eq:gk}  and  \meqref{eq:g}) is invertible. It can be easily checked and it also follows from the fact that the transformation that sends  $ \ket{s} $ to $ \ket{g} $ is implemented via unitary quantum gates. Therefore, the summation that runs through $ s=0 $ to $ s= n-1 $ in   \meqref{eq:phif}  can be replaced by the summation running through $ g =0 $ to $ g=n-1 $. It follows that
\begin{align} \label{eq:phif_final}
	\ket{\phi_f }  &= \frac{1}{\sqrt{N}}  \sum_{g=0}^{N-1}   \, (-1)^{\sum_{k=0}^{n-1} \, j_k (g_{n-k}  \, \oplus \,  g_{n-k-1} )}   \, \ket{g}  \nonumber \\
	&= \frac{1}{\sqrt{N}}  \sum_{g=0}^{N-1}   \, (-1)^{\sum_{r=0}^{n-1} \, j_{n-1-r} (g_{r} \,  \oplus \, g_{r+1} )}   \, \ket{g}. 
\end{align}
From  \meqref{eq:sequency:hadamard} and  \meqref{eq:phif_final},  it follows that the quantum circuit shown in \mfig{fig:sequency} can be used for computing the Walsh-Hadamard transform in sequency order.

\subsection{The action of $U_Z$} \label{Sec:UZ}

\begin{figure}
	\centering
	\includegraphics[scale=1.2]{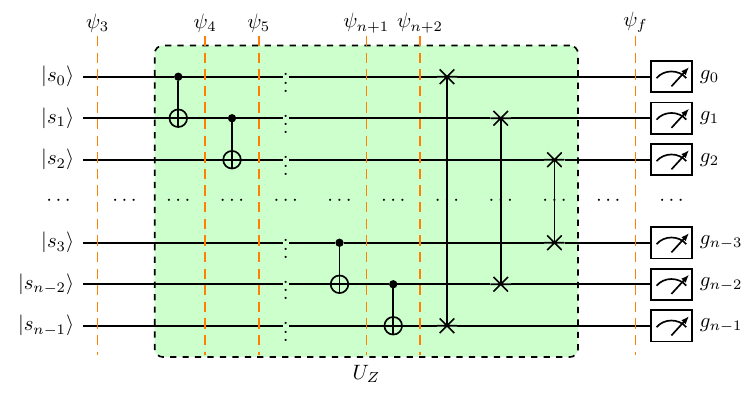}
	\caption{ The quantum circuit shown in this figure (whose action is denoted by the unitary $U_Z$) is part of the quantum circuit between the barriers labeled $\phi_1$ to $\phi_f$ in \mfig{fig:sequency}.
 } \label{fig:last-part}
\end{figure}

The quantum circuit shown in \mfig{fig:last-part} corresponds to the unitary $U_Z$ (i.e., part of of the quantum between the barriers labeled $\phi_1$ to $\phi_f$ in \mfig{fig:sequency}). 
Assume $ \ket{s} = \ket{s_{n-1}} \, \otimes \, \ket{s_{n-2}} \, \otimes \, \cdots \, \otimes \,  \ket{s_{1}} \, \otimes \ket{s_{0}}  $, i.e., assume the decimal representation of $ s $ to be  $ \sum_{m=0}^{n-1} \, s_m 2^m  $. 
The quantum circuit in \mfig{fig:last-part} contains $ n-1 $ CNOT and $ \floor{\frac{n}{2}} $ swap gates. The swap gates could be applied in parallel. However, for clarity, they are not shown to be parallel in  \mfig{fig:last-part}. 
The swap gates in \mfig{fig:last-part} simply reverse the order of qubits. The action of $ U_Z $ in the quantum circuit in \mfig{fig:last-part} can be described as  
 $ U_Z \ket{\psi_3} = \ket{g}$, where  \begin{equation}\label{eq:g}
 	g = \sum_{j=0}^{n-1} \, g_j 2^j,
 \end{equation}
where,
\begin{align*}
	g_{n-1} & = \, s_0, \\
	g_{n-2} & =  \, s_0 \, \oplus \, s_1, \\ 
	\cdots&  \,\cdots \\
	g_1 &=  \, s_0 \, \oplus \, s_1 \, \oplus \, s_2 \, \oplus \, \ldots  \, \oplus \, s_{n-2}, \\
	g_0 &= \, \, s_0 \, \oplus \, s_1 \, \oplus \, s_2 \, \oplus \, \ldots  \, \oplus \, s_{n-2} \, \oplus \, s_{n-1}.
\end{align*}  
It means,
\begin{align}\label{eq:gk}
	g_{n-1} = s_0,  \quad \text{and} \quad	g_k &= \, \, s_0 \, \oplus \, s_1 \, \oplus \, \ldots   \, \oplus \, s_{n-1-k}, \quad \text{for $ k = n-2 $ to $ 0 $}.
\end{align}
Finally, all the qubits are measured and stored in a classical register. The classical bits resulting from these measurements are labeled $ g_0 $ to $ g_{n-1} $ in \mfig{fig:last-part}. It is interesting to note that  $ g $ is the number of zero-crossings of the vector $\widetilde{{\bf W}}_{s}$ defined in \meqref{eq:defWS}. To see this, we note that using \meqref{Eq:defn_zero_crossings} one can compute the number of zero-crossings of the vector $\widetilde{{\bf W}}_{s}$ by the following expression,
\[
 \frac{1}{2} \sum_{k=0}^{N-2} \, \abs{   (-1)^{s \cdot (k+1)} - (-1)^{s \cdot k}}.
\]
Further, on using Corollary \ref{cor}, it is easy to see that $ g = Z_n(s) = \frac{1}{2} \sum_{k=0}^{N-2} \, \abs{(-1)^{s \cdot (k+1)} - (-1)^{s \cdot k}}$ is the number of number of zero-crossings of the vector $\widetilde{{\bf W}}_{s}$. It follows that  $ g $ is the number of zero-crossings (or sequency) of  $\widetilde{{\bf W}}_{s}$ (which is the $s$-th Walsh basis vector in natural ordering).

\subsection{A lemma} \label{sec:Lemma} 
 In this section, we will prove Lemma~\ref{lemma} and its corollary (Corollary \ref{cor}) to show that  $ g = \sum_{j=0}^{n-1} \, g_j 2^j $ (see  \meqref{eq:g}  and  \meqref{eq:gk}) gives the number of of zero-crossings of the vector $ \bm{\mathcal{S}} $ given by 
 \[
			 \bm{\mathcal{S}} = 	\left[\, F(0) \,\,\,  F(1) \,\,\, \cdots \cdots \,\,\, F(N-1) \,\right]^T,
 \]
 with $ F(x) = (-1)^{f(x)} $ and $ f(x) = s \cdot x $ (ref.\,\,\!\meqref{eq:def_sequence}). 
 
\begin{lem} \label{lemma}
	For an integer $ m $ with $ 1 < m \leq n $, and for 
	$ x = x_{n-1}\,x_{n-2} \,\ldots \, x_1\, x_0 $ (or equivalently, $ x $ with a decimal representation  $  x = \sum_{j=0}^{m-1} \, x_j 2^j  $),
	 with $ x_j \in \{0,1\} $,   define $ Z_m(x) $ as
	\begin{equation}\label{eq:def_Z_k}
		Z_m(x) := \frac{1}{2} \sum_{k=0}^{2^m-2} \, \abs{   (-1)^{x \cdot (k+1)} - (-1)^{x \cdot k}}.
	\end{equation}
Then, we have 
\begin{equation}\label{eq:lemma}
Z_m(s(m)) = 2 Z_{m-1} (s(m-1)) + (s_0 \, \oplus \, s_1 \, \oplus \, s_2 \, \oplus \, \ldots \, \oplus \,  s_{m-1}).	
\end{equation}
Here $ s(m) = s_{m-1}\,s_{m-2}\,\ldots \, s_1\, s_0 $
and $ s_0 \, \oplus \, s_1 \, \oplus \, s_2 \, \oplus \, \ldots \, \oplus \,  s_{m-1}= (s_0 + s_1 + s_2 + \ldots + s_{m-1}) \pmod 2 $. (ref.\,\,\!\mref{sec:notation}). 
\end{lem}

\begin{proof}
Let $ M = 2^m $. We put $ x=s(m) $ and split the summation on the right side of  \meqref{eq:def_Z_k}  into three different parts as follows.
\begin{align}	\label{eq:lem_all_terms}
	Z_m(s(m)) \nonumber  = &\frac{1}{2} \left(\sum_{k=0}^{\frac{M}{2}-2} \, \abs{   (-1)^{s(m) \cdot (k+1)} - (-1)^{s(m) \cdot k}}\right) + \frac{1}{2} \abs{  (-1)^{s(m) \cdot \left(\frac{M}{2}\right)} - (-1)^{s(m) \cdot \left(\frac{M}{2} -1 \right) }} \\&+ \frac{1}{2} \sum_{k=\frac{M}{2}}^{M-2} \, \abs{   (-1)^{s(m) \cdot (k+1)} - (-1)^{s(m) \cdot k}}.
\end{align} 
We note that $ \frac{M}{2} = 2^{m-1}$ represents the $ m $-bit string $ 10\ldots0 $. Therefore, $ s(m) \cdot \left(\frac{M}{2}\right) = s_{m-1} $. Similarly,  $ \frac{M}{2} -1 = 2^{m-1}-1$ represents the $ m $-bit string $ 011\ldots1 $, hence  $ s(m) \cdot \left(\frac{M}{2} -1 \right) =  s_0 \, \oplus \, s_1 \, \oplus \, s_2 \, \oplus \, \ldots \, \oplus \,  s_{m-2}$.
Therefore, the middle term in the above summation reduces to 
\begin{align}\label{eq:lem_second_term}
	\frac{1}{2} \abs{  (-1)^{s(m) \cdot \left(\frac{M}{2}\right)} - (-1)^{s(m) \cdot \left(\frac{M}{2} -1 \right) }} & = \frac{1}{2} \abs{  (-1)^{s_{m-1}} - (-1)^{ s_0 \, \oplus \, s_1 \, \oplus \, \ldots \, \oplus \,  s_{m-2} }} \nonumber \\ & = s_0 \, \oplus \, s_1 \, \oplus \, s_2 \, \oplus \, \ldots \, \oplus \,  s_{m-1}.
\end{align}
Next we show that the last and the first terms are equal. 
\begin{align}\label{eq:lem_last_term}
\frac{1}{2} \sum_{k=\frac{M}{2}}^{M-2} \, \abs{   (-1)^{s(m) \cdot (k+1)} - (-1)^{s(m) \cdot k}} & = 	\frac{1}{2} \sum_{k=0}^{\frac{M}{2}-2} \, \abs{   (-1)^{s(m) \cdot (\frac{M}{2} + k+1)} - (-1)^{s(m) \cdot (\frac{M}{2} + k)}}  \nonumber \\
  & =   \frac{1}{2} \sum_{k=0}^{\frac{M}{2}-2} \, \abs{   (-1)^{s_{m-1}}  \left((-1)^{s(m) \cdot (  k+1)} - (-1)^{s(m) \cdot (k)}  \right)  } \nonumber \\
  & =  \frac{1}{2} \sum_{k=0}^{\frac{M}{2}-2} \, \abs{    \left((-1)^{s(m) \cdot (  k+1)} - (-1)^{s(m) \cdot (k)}  \right)  }. 
\end{align}
From \meqref{eq:lem_all_terms}, \meqref{eq:lem_second_term} and \meqref{eq:lem_last_term} it follows that
\begin{align}
Z_m(s(m)) & = \left( \sum_{k=0}^{\frac{M}{2}-2} \, \abs{   (-1)^{s(m) \cdot (k+1)} - (-1)^{s(m) \cdot k}}\right) + (s_0 \, \oplus \, s_1 \, \oplus \, \ldots \, \oplus \,  s_{m-1}) \nonumber \\
& =  \left( \sum_{k=0}^{\frac{M}{2}-2} \, \abs{   (-1)^{s(m-1) \cdot (k+1)} - (-1)^{s(m-1) \cdot k}}\right) + (s_0 \, \oplus \, s_1 \, \oplus \, \ldots \, \oplus \,  s_{m-1}).  
\end{align}
The last step follows because as $ k $ runs through $ 0$ to  $ \frac{M}{2} -2 = 2^{m-1} -2 $, the computations of  $ s(m) \cdot (k+1) $ and $ s(m) \cdot k $  involve only the $ m-1 $ least significant bits of $ s $, allowing one to write $ s(m) \cdot (k+1) = s(m-1) \cdot (k+1) $ and $ s(m) \cdot k = s(m-1) \cdot k $. 
Hence, we obtain
\begin{align*}
Z_m(s(m)) = 2 Z_{m-1} (s(m-1)) + (s_0 \, \oplus \, s_1 \, \oplus \, s_2 \, \oplus \, \ldots \, \oplus \,  s_{m-1}),
\end{align*}
and the proof is complete.
\end{proof}

\begin{cor}\label{cor}
	Let  $ s = s_{n-1}\,s_{n-2}\,\ldots \, s_1\, s_0 $ with $ s_j \in \{0,1\} $. If $ 	Z_n(s) $ is defined as 
	\[
		Z_n(s) := \frac{1}{2} \sum_{k=0}^{2^n-2} \, \abs{   (-1)^{s \cdot (k+1)} - (-1)^{s \cdot k}},
	\]
	then 
	\begin{equation}\label{eq:cor}
		Z_n(s) = \sum_{k=0}^{n-1} \, g_k 2^k,
	\end{equation}
where $ g_{n-1} = s_0 $ and  $ g_k =  s_0 \, \oplus \, s_1 \, \oplus \, \ldots \,  \oplus \, s_{n-1-k}\,$ for $ k=n-2$ to $ k=0 $.
\end{cor}
\begin{proof}
	It is easy to see from  \meqref{eq:lemma} that 
	\begin{align*}
		Z_1(s(1)) &  = s_0  = g_{n-1}\\ 
		Z_2(s(2)) & = 2 Z_1(s(1)) + (s_0 \, \oplus \,  s_1) = 2 g_{n-1} + g_{n-2} \\ 
		Z_3(s(3)) & = 2 Z_2(s(2)) + (s_0 \, \oplus \,  s_1 \, \oplus s_2) = 2^2 g_{n-1} + 2 g_{n-2} + g_{n-3}. 
	\end{align*}
A simple induction argument, whose details we skip, and the observation that $s(n) = s$  (see \mref{sec:notation}), shows that
\[
	Z_n(s) = \sum_{k=0}^{n-1} \, g_k 2^k.
\]
This completes the proof. 
\end{proof}
Clearly, $ Z_n(s) $ computes the number of zero-crossings of the vector $  \bm{\mathcal{S}} $ (ref.\,\,\!\meqref{eq:def_sequence} and  \meqref{Eq:defn_zero_crossings}).
In the following, we give an example to illustrate the steps for computing $ Z_n(s) $.
\begin{example}\label{ex:two}
	Let $ n = 3 $ and $ s = 5 $ (or equivalently $ s = 101 $ as a binary string). 
We have $$  \bm{\mathcal{S}} = \left[(-1)^{s\cdot 0} \,\,\, (-1)^{s\cdot 1} \,\,\, \ldots \ldots \,\,\, (-1)^{s\cdot 7} \right]^T =  
\begin{bmatrix}
     1 & -1 &  1 &  -1 & -1 & 1 & -1 & 1 
\end{bmatrix}^T .$$ 
Let $ N= 2^n = 8 $. 
As $ s = 101 $, we have $ s_0 = 1 $, $ s_1 =0 $ and $ s_2 = 1 $. The number of zero-crossings of the vector $  \bm{\mathcal{S}} $ is given by 
\begin{align}\label{eq:eaxample_one}
	Z_3(s)) = Z_3(s(3)) = 2 Z_2 (s(2)) + (s_0 \, \oplus \, s_1 \, \oplus \, s_2) = 2 Z_2(s(2)) + (1 \, \oplus \, 0 \, \oplus \, 1) = 2 Z_2(s(2)).
\end{align}
We have, $ s(2) = 01 $. Therefore,
\begin{equation}\label{eq:example_two}
	Z_2(s(2)) = 2 Z_1(s(1)) + (s_0 \, \oplus \, s_1 \,) = 2 Z_1(s_0) + ( 1 \, \oplus \, 0 \,) = 2 (1) + 1 = 3.
\end{equation}
From \meqref{eq:eaxample_one}  and  \meqref{eq:example_two},  one gets $ Z_3(s)) = 6 $.
\end{example}

\section{Application: Sequency Filtering} \label{Sec:application}
  
Signal filtering involves manipulating signals to extract information or enhance quality. DC, low-pass, and high-pass filtering are key in applications like audio processing, communication, and medical image analysis. Low-pass filters are often used to reduce noise by keeping low-frequency components of the signal while eliminating the high-frequency noise from it. High-pass filters eliminate baseline drift and low-frequency components. Sequency filtering, similar to frequency filtering, retains desired sequency components in the Walsh domain while attenuating or eliminating the unwanted sequency components. The Walsh domain refers to the representation of a signal obtained using the Walsh-Hadamard transform, as noted earlier. 

In \mfig{fig:signal_and_spectrum_f} and \mfig{fig:signal_and_spectrum_g}, example signals $f(t)$ (\mfig{fig:signal_f_test}) and $g(t)$ (\mfig{fig:signal_g}) and their corresponding  spectra $\widehat{f}(n)$ (\mfig{fig:spectrum_f_test})  and $\widehat{g}(n)$ (\mfig{fig:spectrum_g}) are shown. 
In the following, these two signals will be used to illustrate DC, low-pass, high-pass and band-pass sequency filtering using our proposed quantum approach for signal filtering. 

\begin{figure}[H] 
\centering
      \begin{subfigure}{0.48\textwidth}
		\centering
 \includegraphics[width=\textwidth]{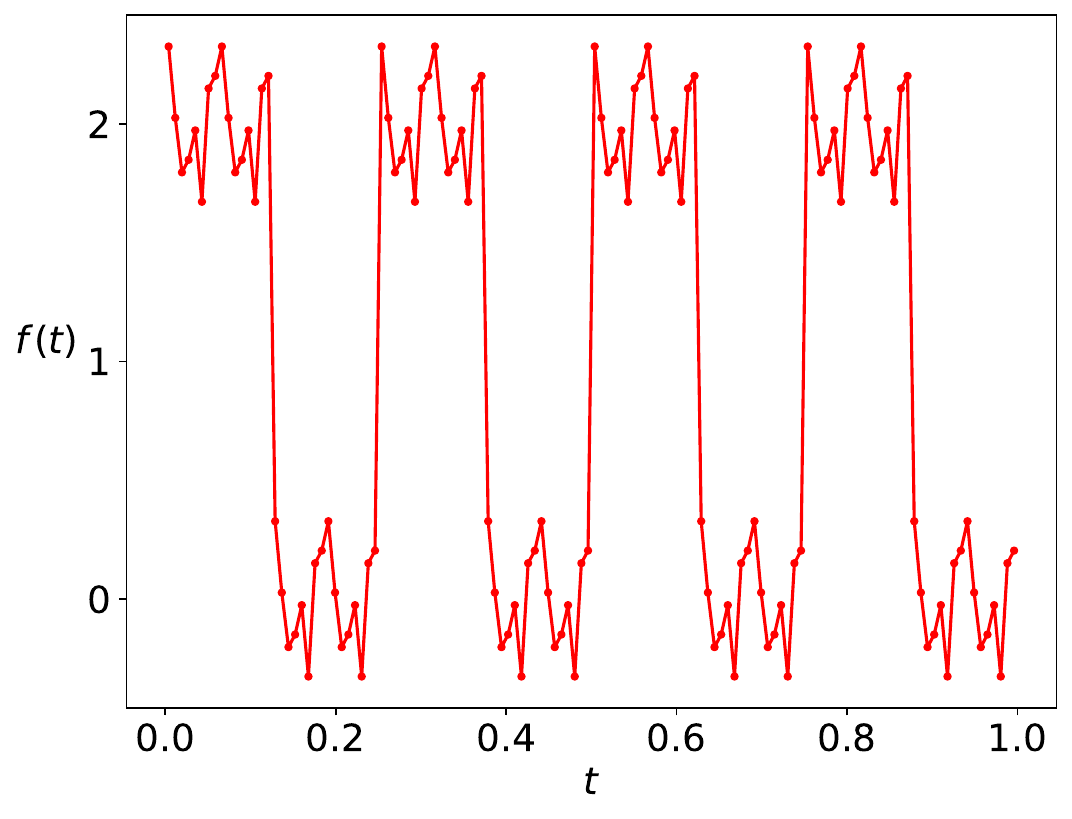}
		\caption{Signal}
		 \label{fig:signal_f_test}
	\end{subfigure}
\begin{subfigure}{0.48\textwidth}
		\centering
 \includegraphics[width=\textwidth]{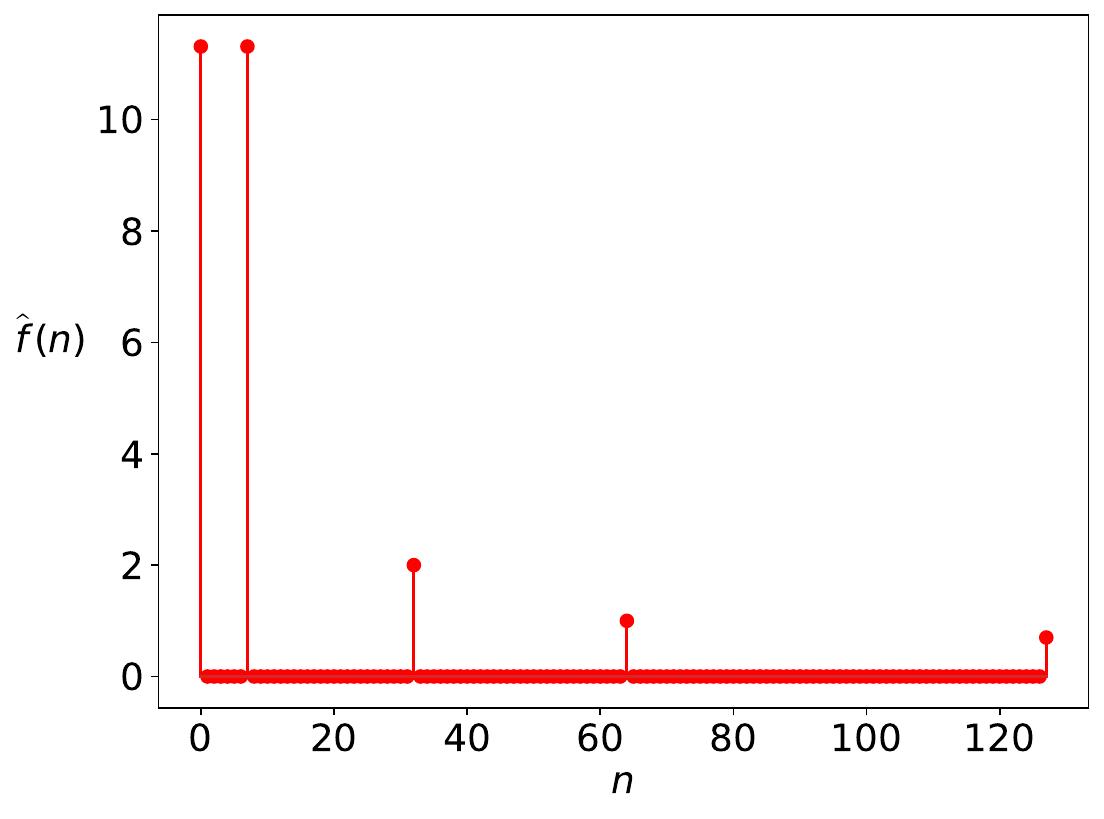}
		\caption{Sequency spectrum}
		 \label{fig:spectrum_f_test}
\end{subfigure}
    \caption{Signal $f(t)$ is shown on the left ((a)) and its sequency spectrum  $\widehat{f}(n)$ is shown on the right ((b)).  
    }
    \label{fig:signal_and_spectrum_f}
\end{figure}

\begin{figure}[H] 
\centering
      \begin{subfigure}{0.48\textwidth}
		\centering
 \includegraphics[width=\textwidth]{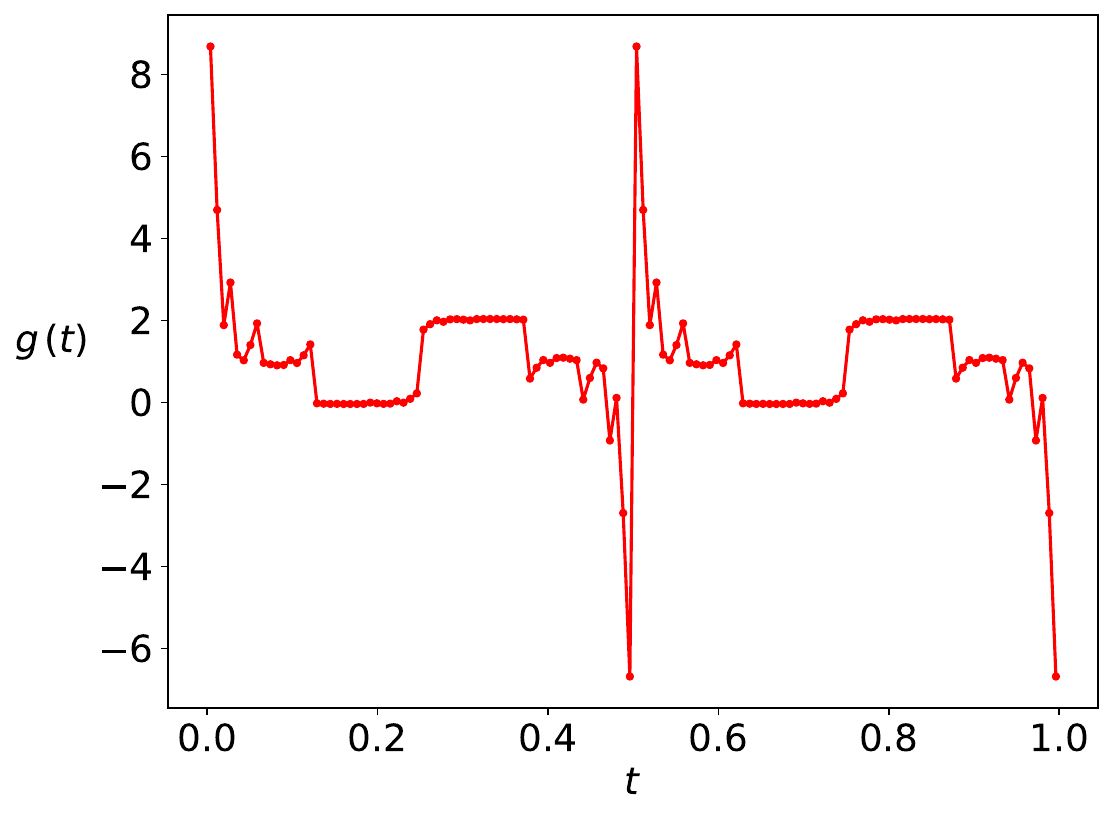}
		\caption{Signal}
		 \label{fig:signal_g}
	\end{subfigure}
\begin{subfigure}{0.48\textwidth}
		\centering
 \includegraphics[width=\textwidth]{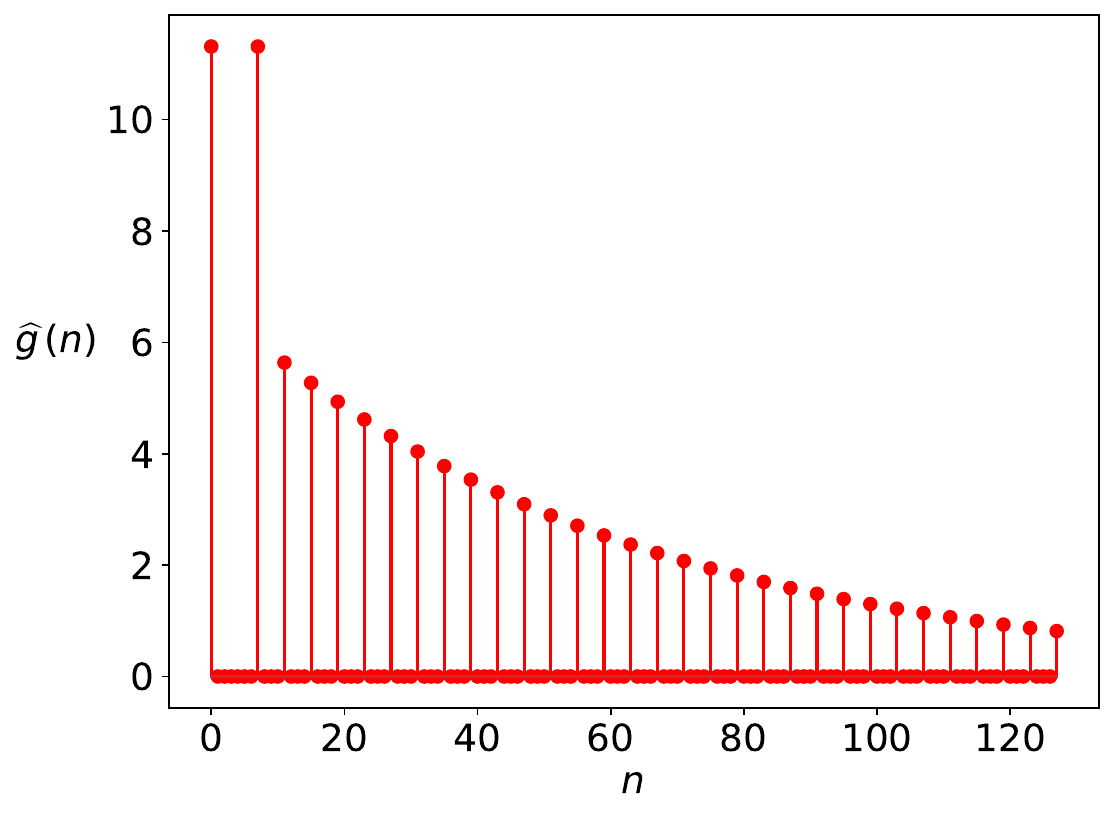}
		\caption{Sequency spectrum}
		 \label{fig:spectrum_g}
\end{subfigure}
    \caption{Signal $g(t)$ is shown on the left ((a)) and its sequency spectrum  $\widehat{g}(n)$ is shown on the right ((b)).  
    }
    \label{fig:signal_and_spectrum_g}
\end{figure}

\subsection{Low-pass and high-pass filtering} \label{Sec:low-pass-high-pass}
Let the sequency spectrum of the signal 
$
{\bf{f}} =\begin{bmatrix}
    f_0 & f_1 & \ldots & f_{N-1} 
\end{bmatrix}^T
$
be  
$
{\bf{f}} =\begin{bmatrix}
    \hat{f}_0 & \hat{f}_1 & \ldots & \hat{f}_{N-1}
\end{bmatrix}^T
$.
Suppose $c$ is the cutoff sequency such that the sequency components with sequency greater than or equal to  $c$ need to be removed. The low-pass filtered signal in the sequency domain is given by
\begin{align}
   {\hat{\bf{f}}}_{LP} =  \begin{bmatrix}
   \hat{f}_0 & \hat{f}_1 & \ldots & \hat{f}_{c-1} &  0 &  0 &  \ldots & 0
\end{bmatrix}^T.   
\end{align}
On the other hand, if sequency components below \( c \) are to be removed, then the corresponding high-pass filtered signal in the sequency domain is given by,
\begin{align}
   {\hat{\bf{f}}}_{HP} =  \begin{bmatrix}
   0 & 0 &\ldots & 0 & \hat{f}_{c} & \hat{f}_{c+1} & \ldots & \hat{f}_{N-1}
\end{bmatrix}^T.   
\end{align}

After removing the selected sequency components, the inverse Walsh-Hadamard transform is performed to convert the filtered sequency domain signals back to the time domain as
\begin{align}
    {\bf{\tilde{f}}}_{LP} =  H_N^S \,{\hat{\bf{f}}}_{LP},  \qquad {\bf{\tilde{f}}}_{HP} = H_N^S {\hat{\bf{f}}}_{HP},
\end{align}
where $  {\bf{\tilde{f}}}_{LP}$  and  $ {\bf{\tilde{f}}}_{HP} $ represent the low-pass and high-pass filtered signals in the time domain, respectively.

In the following, we describe (in Algorithm \ref{alg_filtering}) the steps of our proposed quantum approach for performing low-pass and high-pass filtering. One ancilla qubit is needed in this approach, and at the end of Algorithm \ref{alg_filtering}, the low-pass and high-pass filtered signals are obtained depending upon the state of the ancilla qubit. \\

	\begin{algorithm}[H] \label{alg_filtering}
			\DontPrintSemicolon
			\KwInput{The normalized input signal $\ket{\Psi}$. }
			\KwOutput{The signal $\ket{0} \otimes \ket{\widetilde{\Psi}_{l}} + \ket{1} \otimes \ket{\widetilde{\Psi}_{h}} $, where 
    $\ket{\widetilde{\Psi}_{l}}  $ and $\ket{\widetilde{\Psi}_{h}} $ are low-pass and high-pass filtered signals in the time domain, respectively.}
			\Fn{Filter $\ket{\Psi} $}{
                        Prepare the state $ \ket{0} \otimes \ket{\Psi} $ using $ n +1 $ qubits, where the ancilla qubit (the leftmost or the most significant qubit) is initialized to $\ket{0}$. \\
				    Apply  $ X \otimes H^{\otimes n}$ to the state  $\ket{0} \otimes \ket{\Psi}$ to obtain the state $\ket{1} \otimes \ket{\widehat{\Psi}}$, where $\ket{\widehat{\Psi}} = H^{\otimes n} \ket{\Psi}$ is the Walsh-Hadamard transform of the input signal $\ket{\Psi}$.  
        \\
                        Apply the unitary gate $I \otimes U_Z$ to $\ket{1} \otimes \ket{\widehat{\Psi}}$ to obtain the state $\ket{1} \otimes \ket{\widehat{\Psi_s}} $, where $\ket{\widehat{\Psi_s}} =  U_Z \ket{\widehat{\Psi}}$ is the Walsh-Hadamard transform of the input signal $\ket{\Psi}$ in sequency ordering. \\
                         Apply appropriate multi-controlled $X$ gates on the ancilla qubit, to split the state $\ket{1} \otimes \ket{\widehat{\Psi_s}}$ into low and high sequency components as shown below
                          \begin{align*}
                          \left(X \otimes \sum_{Z_n(k) < c } \ket{k}\bra{k} +  I \otimes \sum_{Z_n(k) \geq c  } \ket{k}\bra{k}\right) \left[\ket{1} \otimes \ket{\widehat{\Psi_s}}\right]
                          =   \ket{0} \otimes \ket{\widehat{\Psi}_{s,l}} + \ket{1} \otimes \ket{\widehat{\Psi}_{s,h}},
                        \end{align*}
                        where $\ket{\widehat{\Psi}_{s,l}} =  \sum_{Z_n(k) < c } \braket{k \, | \,\widehat{\Psi_s}} \ket{k}$ is the low sequency component and $ \ket{\widehat{\Psi}_{s,h}} = \sum_{Z_n(k) \geq c } \braket{k \, | \,\widehat{\Psi_s}} \ket{k}$ is the high sequency component, and $Z_n(k)$ denotes the sequency of $k$ (refer \meqref{eq:def_Z_k}). 
                        \\
                        This is an optional step. Although not part of the low-pass or high-pass filtering, the state $\ket{\widehat{\Psi}_{s,l}}$ and $\ket{\widehat{\Psi}_{s,h}}$ can be further processed by applying appropriate quantum gates as needed.  \\
                        Apply the unitary gate $ I \otimes U_Z^{-1}$ to obtain the filtered signal in natural ordering in the transformed domain.
                    \begin{align*}
                      \left( I \otimes U_Z^{-1}\right)  \left[ \ket{0} \otimes \ket{\widehat{\Psi}_{s,l}} + \ket{1} \otimes \ket{\widehat{\Psi}_{s,h}}\right] =  \ket{0} \otimes \ket{\widehat{\Psi}_{n,l}} + \ket{1} \otimes \ket{\widehat{\Psi}_{n,h}}.
                    \end{align*} \\
                        Apply Hadamard gates to $\ket{0} \otimes \ket{\widehat{\Psi}_{n,l}} + \ket{1} \otimes \ket{\widehat{\Psi}_{n,h}}$ to convert the  filtered signal $\ket{\widehat{\Psi}_{n,l}}$ (and $\ket{\widehat{\Psi}_{n,h}}$) back to the time domain,  
                        \begin{align*}
                            \left(  I   \otimes H^{\otimes n} \right)     \left[ \ket{0} \otimes \ket{\widehat{\Psi}_{n,l}} + \ket{1} \otimes \ket{\widehat{\Psi}_{n,h}} \right] =  \ket{0} \otimes \ket{\widetilde{\Psi}_{l}} + \ket{1} \otimes \ket{\widetilde{\Psi}_{h}},
                                \end{align*}
                            where $\ket{\widetilde{\Psi}_{l}} = H^{\otimes n}   \ket{\widehat{\Psi}_{n,l}} $ and $\ket{\widetilde{\Psi}_{h}} = H^{\otimes n}   \ket{\widehat{\Psi}_{n,h}} $ are low-pass and high-pass filtered signals in the time domain, respectively. \\
					\Return{  $\ket{0} \otimes \ket{\widetilde{\Psi}_{l}} + \ket{1} \otimes \ket{\widetilde{\Psi}_{h}} $.}
			}
			\caption{A quantum algorithm for low-pass and high-pass filtering.}
		\end{algorithm}
\begin{remark}
    We note that, if in Step 3 of Algorithm \ref{alg_filtering}, instead of  $\left(  X \otimes H^{\otimes n} \right)$, the unitary  $\left(  I \otimes H^{\otimes n} \right)$ is applied to the state  $\ket{0} \otimes \ket{\Psi}$, keeping all the remaining steps the same, then  Algorithm \ref{alg_filtering} returns the state  $\ket{0} \otimes \ket{\widetilde{\Psi}_{h}} + \ket{1} \otimes \ket{\widetilde{\Psi}_{l}} $. In other words, in this case, the high sequency component of the signal  $\ket{\widetilde{\Psi}_{h}}$ is associated with the state of ancilla qubit being $\ket{0} $ and the high sequency component of the signal is associated with the state of ancilla qubit being $\ket{1} $. 
\end{remark}
  
Some examples of low-pass and high-pass filtering, along with relevant quantum circuits, to illustrate Algorithm \ref{alg_filtering} are described below.

\subsubsection{Examples of low-pass and high-pass filtering}
Quantum circuits for low-pass and high-pass filtering for illustrating Algorithm \ref{alg_filtering}, for $n=7$ qubits, are shown in 
\mfig{fig:low_pass_circuit_filtering}. We note that the input to the circuit is $\ket{0} \otimes \ket{\Psi}$, where  $\ket{\Psi} = \ket{q_6} \otimes \ket{q_5} \otimes \cdots \otimes \ket{q_0} $ is the normalized signal and the ancilla qubit $\ket{q_7}$ is initialized to $\ket{0}$. This corresponds to Step 2 in Algorithm \ref{alg_filtering}. Next, $X \otimes H^{\otimes n} = X \otimes H^{\otimes 7}$ is applied to the state $\ket{0} \otimes \ket{\Psi}$ (this is shown before the barrier labeled A in the quantum circuit in \mfig{fig:low_pass_circuit_filtering}). This corresponds to Step 3 in Algorithm \ref{alg_filtering} and results in the state $\ket{1} \otimes \ket{\widehat{\Psi}}$, where $\ket{\widehat{\Psi}} = H^{\otimes 7} \ket{\Psi} $ is the Hadamard transform of the input signal $\ket{\Psi}$ in natural ordering. The part of the quantum circuit between barriers labeled A and B in \mfig{fig:low_pass_circuit_filtering}, implements the application of the unitary gate $I \otimes U_Z$ to $\ket{1} \otimes \ket{\widehat{\Psi}}$ in Step 4 of Algorithm \ref{alg_filtering}, and results in the quantum state $\ket{1} \otimes \ket{\widehat{\Psi_s}} $, where $\ket{\widehat{\Psi_s}} =  U_Z \ket{\widehat{\Psi}}$ is the Walsh-Hadamard transform of the input signal $\ket{\Psi}$ in sequency ordering. 

Next, the part of the quantum circuit between the barriers labeled B and C implements Step 5 of Algorithm \ref{alg_filtering}. In 
\mfig{fig:low_pass_circuit_filtering_fig_1}, a CNOT gate is applied on the target ancilla qubit $q_7$ when the control qubit $q_6$ is in the $\ket{0}$ state. We note that the condition $\ket{q_6} = \ket{0}$ corresponds to the selection of the sequency components of the signal that are less than $\frac{N}{2}$, where $N=2^n$. Therefore, the action of the CNOT gate results in flipping the target ancilla qubit to $\ket{0}$ state for sequency components of the signal that are less than $\frac{N}{2}$. 
As a result, the state  $\ket{1} \otimes \ket{\widehat{\Psi_s}}$ is split it into the low sequency component $\ket{\widehat{\Psi}_{s,l}}$ and the high sequency component $\ket{\widehat{\Psi}_{s,h}}$, as shown below
                          \begin{align*}
                          \left(X \otimes \sum_{Z_n(k) < c } \ket{k}\bra{k} +  I \otimes \sum_{Z_n(k) \geq c  } \ket{k}\bra{k}\right) \left[\ket{1} \otimes \ket{\widehat{\Psi_s}}\right]
                          =   \ket{0} \otimes \ket{\widehat{\Psi}_{s,l}} + \ket{1} \otimes \ket{\widehat{\Psi}_{s,h}},
                        \end{align*}
                        where $c = \frac{N}{2}$. 
                        Similarly, the quantum circuit in  
\mfig{fig:low_pass_circuit_filtering_fig_2} splits the signal into low and high sequency components with the cut-off sequency $c = \frac{N}{4}$, (i.e., where sequency components below the cut-off sequency  $c = \frac{N}{4}$ are regarded as low sequency components and above or equal to $c = \frac{N}{4}$ are regarded as high sequency components). We note that in \mfig{fig:low_pass_circuit_filtering_fig_2}, a multi-controlled $X$ gate is applied on the target ancilla qubit $q_7$ when both the control qubits $q_5$ and $q_6$ are in the $\ket{0}$ state. We note that this condition selects the sequency components of the signal that are less than $\frac{N}{4}$, where $N=2^n$. Therefore, in this case, the action of the multi-controlled $X$ gate results in flipping the target ancilla qubit to $\ket{0}$ state for sequency components of the signal that are less than $\frac{N}{4}$. 
In a similar way, one can observe that the quantum circuit in  
\mfig{fig:low_pass_circuit_filtering_fig_3} splits the signal into low and high sequency components with the cut-off sequency $c = \frac{3N}{4}$. (We note that in this case, the quantum circuit can be simplified by removing the two redundant $X$ gates used in the circuit, i.e., in this case, it is easier to start with the ancilla qubit $\ket{q_7}$ in the $\ket{0}$ state).

The next step, i.e., Step 6 in Algorithm \ref{alg_filtering} is optional and can be implemented between barriers labeled C and D. We observe that at this stage, the signal is already split into low and high sequency components in the Walsh (transformed) domain. Therefore, if needed, further processing and signal-shaping operations based on low and high sequency components can be carried out here. 

In Step 7 of Algorithm \ref{alg_filtering}, which is implemented between barriers labeled D and E,  
the unitary gate $ I \otimes U_Z^{-1}$ is applied to obtain the filtered signal in natural ordering in the Walsh (transformed) domain.
            \begin{align*}
                      \left( I \otimes U_Z^{-1}\right)  \left[ \ket{0} \otimes \ket{\widehat{\Psi}_{s,l}} + \ket{1} \otimes \ket{\widehat{\Psi}_{s,h}}\right] =  \ket{0} \otimes \ket{\widehat{\Psi}_{n,l}} + \ket{1} \otimes \ket{\widehat{\Psi}_{n,h}}.
            \end{align*} 
Finally, in Step 8 of Algorithm \ref{alg_filtering}, implemented after the barrier labeled E, Hadamard gates are applied to $\ket{0} \otimes \ket{\widehat{\Psi}_{n,l}} + \ket{1} \otimes \ket{\widehat{\Psi}_{n,h}}$ to convert the  filtered signal $\ket{\widehat{\Psi}_{n,l}}$ (and $\ket{\widehat{\Psi}_{n,h}}$) back to the time domain,  
                    \begin{align*}
                            \left(  I   \otimes H^{\otimes n} \right)     \left[ \ket{0} \otimes \ket{\widehat{\Psi}_{n,l}} + \ket{1} \otimes \ket{\widehat{\Psi}_{n,h}} \right] =  \ket{0} \otimes \ket{\widetilde{\Psi}_{l}} + \ket{1} \otimes \ket{\widetilde{\Psi}_{h}},
                    \end{align*}
                            where $\ket{\widetilde{\Psi}_{l}} = H^{\otimes n}   \ket{\widehat{\Psi}_{n,l}} $ and $\ket{\widetilde{\Psi}_{h}} = H^{\otimes n}   \ket{\widehat{\Psi}_{n,h}} $ are low-pass and high-pass filtered signals in the time domain, respectively.

\begin{remark} \label{Remark_complexity}
\leavevmode
\begin{enumerate}[(a)]
    \item  We observe that for the general case when $N=2^n$, Algorithm \ref{alg_filtering} requires $n+1$ qubits, say $q_n q_{n-1} \ldots q_0$, where the most significant qubit $q_n$ is the ancilla qubit.  
     \item In Step 5 of Algorithm \ref{alg_filtering}, when the cutoff sequency is $\frac{N}{2^r}$, with $1 \leq r \leq n $, one multiple control $X$ gate with $r$ open controls is needed. The control qubits are $q_{n-1}$, $q_{n-2}$, $\ldots$, $q_{n-r}$ and the target qubit is $q_n$ in this case. We recall that the use of open controls means that the target ancilla qubit state is flipped only when each of the control qubits is in the $\ket{0}$ state).
    We note that one multiple control $X$ gate with $r$ controls can be implemented using $O(r)$ CNOT gates  (see Lemma 7.2 in \cite{barenco1995elementary}).
    \item Similarly, it is easy to verify that, if the cutoff sequency is $N - \frac{N}{2^r}$, with $1 \leq r \leq n $, one multiple control $X$ gate with $r$ closed controls is needed in Step 5 of Algorithm \ref{alg_filtering}. The control qubits are $q_{n-1}$, $q_{n-2}$, $\ldots$, $q_{n-r}$ and the target qubit is $q_n$.
    Here the use of closed controls imply that the target ancilla qubit state is flipped only when each of the control qubits is in the $\ket{1}$ state.
\end{enumerate}
\end{remark}

\begin{figure}[H] 
\centering
      \begin{subfigure}{0.9\textwidth}
		\centering
	\includegraphics[width=\textwidth]{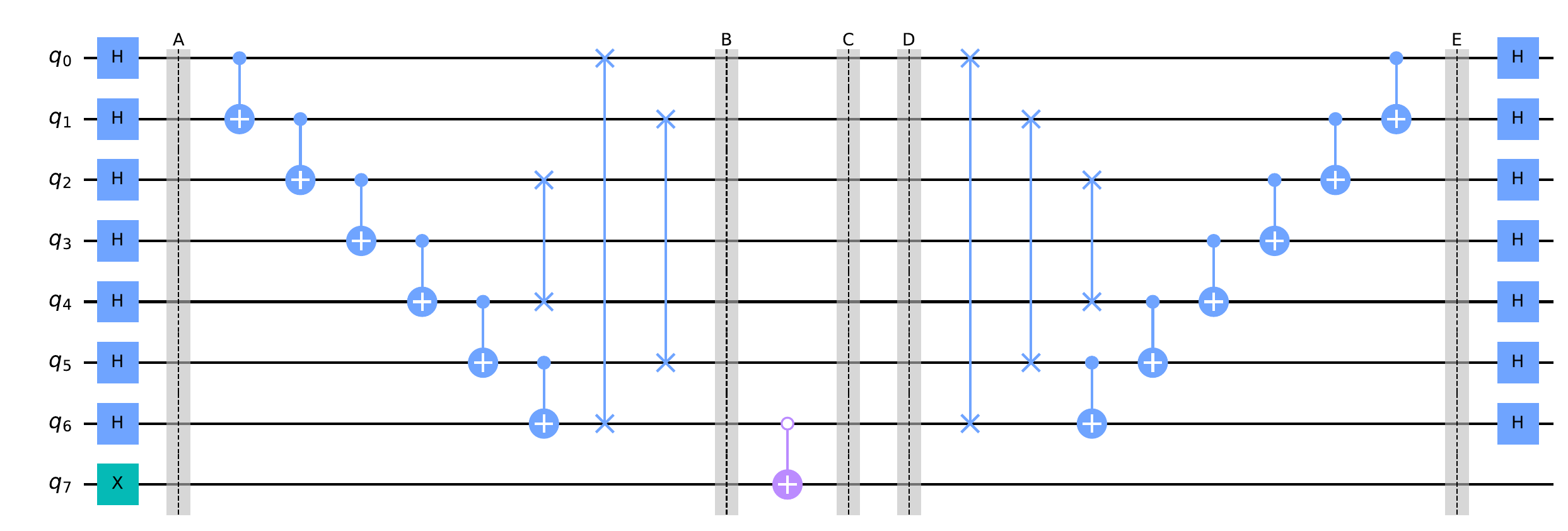}
		\caption{Quantum circuit: Cut-off sequency $= N/2$}
		\label{fig:low_pass_circuit_filtering_fig_1}
	\end{subfigure}
\\
\begin{subfigure}{0.9\textwidth}
		\centering
	\includegraphics[width=\textwidth]{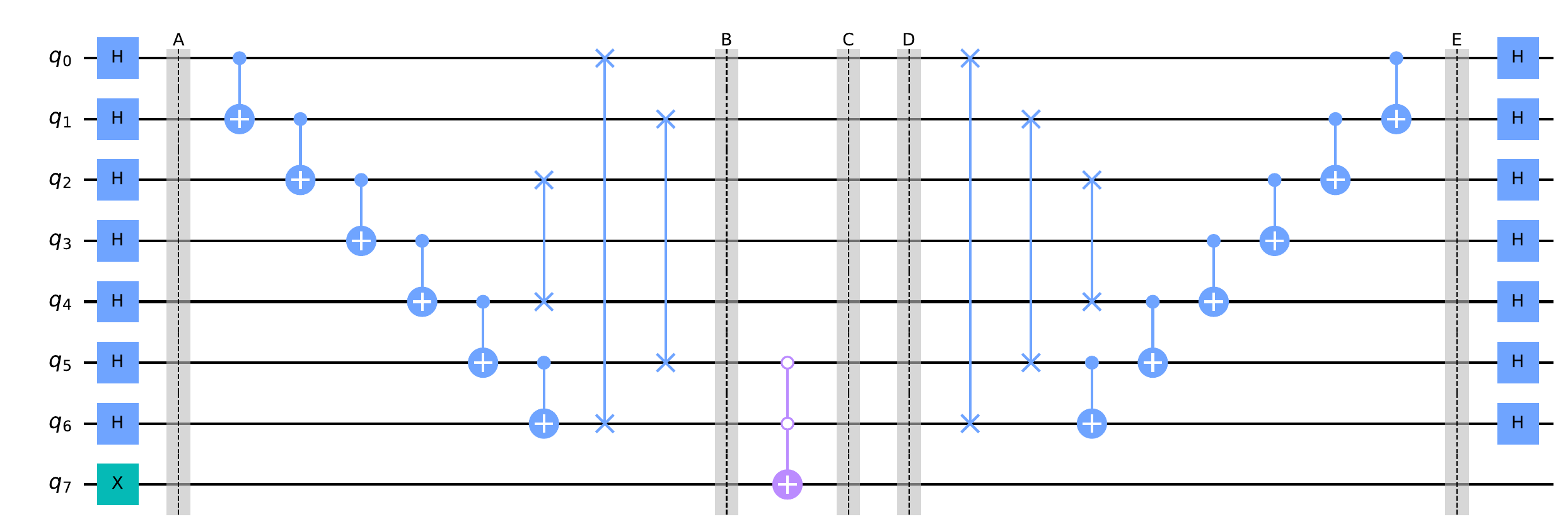}
		\caption{Quantum circuit: Cut-off sequency $= N/4$}
		\label{fig:low_pass_circuit_filtering_fig_2}
	\end{subfigure}\\
 \begin{subfigure}{0.9\textwidth}
		\centering
	\includegraphics[width=\textwidth]{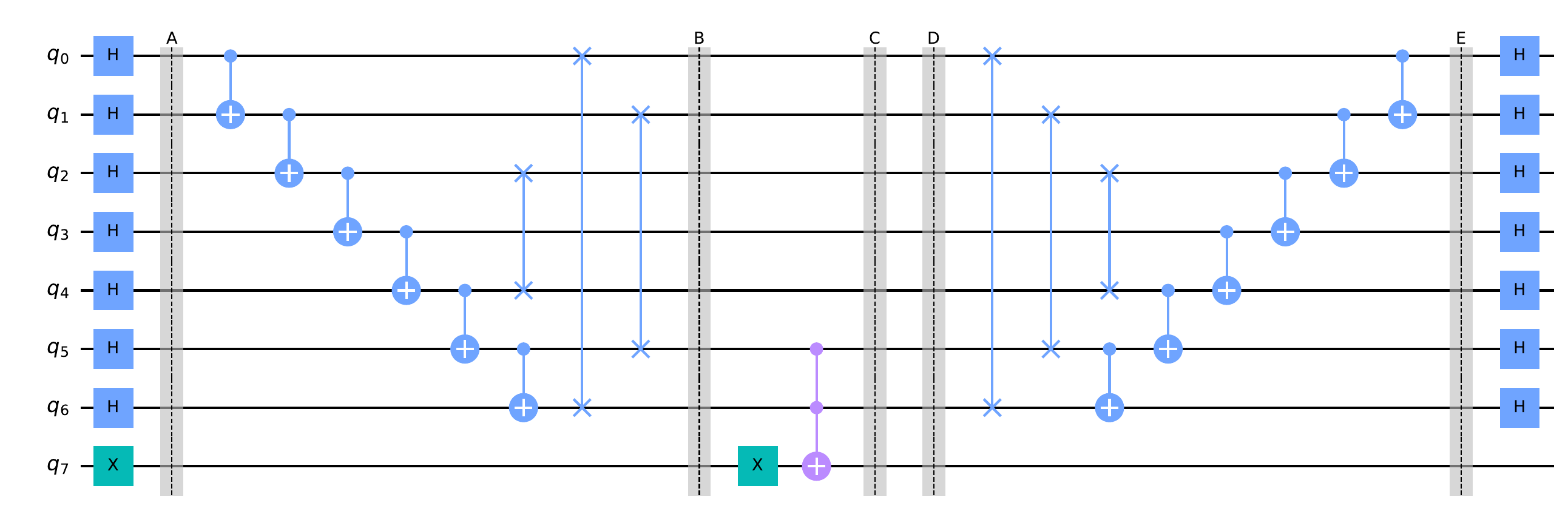}
		\caption{Quantum circuit: Cut-off sequency $= 3N/4$}
		\label{fig:low_pass_circuit_filtering_fig_3}
	\end{subfigure}
    \caption{Quantum circuits for low-pass (high-pass) filtering with cut-off sequencies of $N/2$ (top), $N/4$ (middle) and $3N/4$ (bottom). Information on low-pass (high-pass) filtered signal can be obtained by further measurements/processing when the ancilla (i.e., the most significant) qubit is in the $\ket{0}$ ($\ket{1}$ state). 
    }
    \label{fig:low_pass_circuit_filtering}
\end{figure}

In the following,  we present the results for low-pass filtering of the signals  $f(t)$  and $g(t)$ (shown in \mfig{fig:signal_and_spectrum_f} and \mfig{fig:signal_and_spectrum_g}, respectively) obtained by using Algorithm \ref{alg_filtering}. The quantum circuits shown in \mfig{fig:low_pass_circuit_filtering_fig_1}, \mfig{fig:low_pass_circuit_filtering_fig_2} and \mfig{fig:low_pass_circuit_filtering_fig_3} are applied to the normalized input states $\ket{\Psi_f}$ and $\ket{\Psi_g}$ corresponding to the discretized versions of the signals  $f(t)$  and $g(t)$, respectively. This results in low-pass (or high-pass) filtering of the input signals  $f(t)$  and $g(t)$, if the state of the ancilla qubit $\ket{q_7}$ is $\ket{0}$ (or $\ket{1}$) after the barrier labeled E. 

Application of the quantum circuit shown in \mfig{fig:low_pass_circuit_filtering_fig_2} to the normalized input state $\ket{\Psi_f}$ corresponding to the discretized version of the signal  $f(t)$ results in the low-pass filtered signal with cut-off sequency of $\frac{N}{4}$. The low-pass filtered signal and its spectra are shown in \mfig{fig:Low-pass_filtered_signal_f_N_by_4} and \mfig{fig:Low-pass_filtered_spectrum_f_N_by_4}, respectively. 
Similarly, application of the quantum circuit shown in \mfig{fig:low_pass_circuit_filtering_fig_1} to $\ket{\Psi_f}$ results in the low-pass filtered signal with cut-off sequency of $\frac{N}{2}$. The low-pass filtered signal and its spectra are shown in \mfig{fig:Low-pass_filtered_signal_f_N_by_2} and \mfig{fig:Low-pass_filtered_spectrum_f_N_by_2}, respectively. 

Application of the quantum circuits shown in \mfig{fig:low_pass_circuit_filtering_fig_2}, \mfig{fig:low_pass_circuit_filtering_fig_1} and \mfig{fig:low_pass_circuit_filtering_fig_3} to the normalized input state $\ket{\Psi_g}$ corresponding to the discretized version of the signal  $g(t)$ results in the low-pass (high-pass) filtered signals with cut-off sequencies of $\frac{N}{4}$, $\frac{N}{2}$ and $\frac{3N}{4}$, respectively, depending on the state of the ancilla qubit as discussed earlier. The low-pass filtered signals are shown in \mfig{fig:Low-pass_filtered_signal_g_N_by_4}, \mfig{fig:Low-pass_filtered_signal_g_N_by_2}, and \mfig{fig:Low-pass_filtered_signal_g_3_N_by_4} for cut-off sequencies of $\frac{N}{4}$, $\frac{N}{2}$ and $\frac{3N}{4}$, respectively. Their corresponding spectra are shown in \mfig{fig:Low-pass_filtered_spectrum_g_N_by_4}, \mfig{fig:Low-pass_filtered_spectrum_g_N_by_2} and \mfig{fig:Low-pass_filtered_spectrum_g_3_N_by_4}, respectively. Further, the high-pass filtered signals are shown in \mfig{fig:High-pass_filtered_signal_g_N_by_4}, \mfig{fig:High-pass_filtered_signal_g_N_by_2}, and \mfig{fig:High-pass_filtered_signal_g_3_N_by_4} for cut-off sequencies of $\frac{N}{4}$, $\frac{N}{2}$ and $\frac{3N}{4}$, respectively. Their corresponding spectra are shown in \mfig{fig:High-pass_filtered_spectrum_g_N_by_4}, \mfig{fig:High-pass_filtered_spectrum_g_N_by_2} and \mfig{fig:High-pass_filtered_spectrum_g_3_N_by_4}, respectively. These results for low-pass and high-pass filtering are in agreement with the expected results. We note that these results were obtained using simulations on IBM's open-source platform Qiskit by directly reading the state vectors. On an actual quantum computer, one can not directly obtain the state vector, but the global features of the filtered signals can be extracted efficiently.

\begin{figure}[H]
    \centering
    \begin{subfigure}{0.48\textwidth}
		\centering
\includegraphics[width=\textwidth]{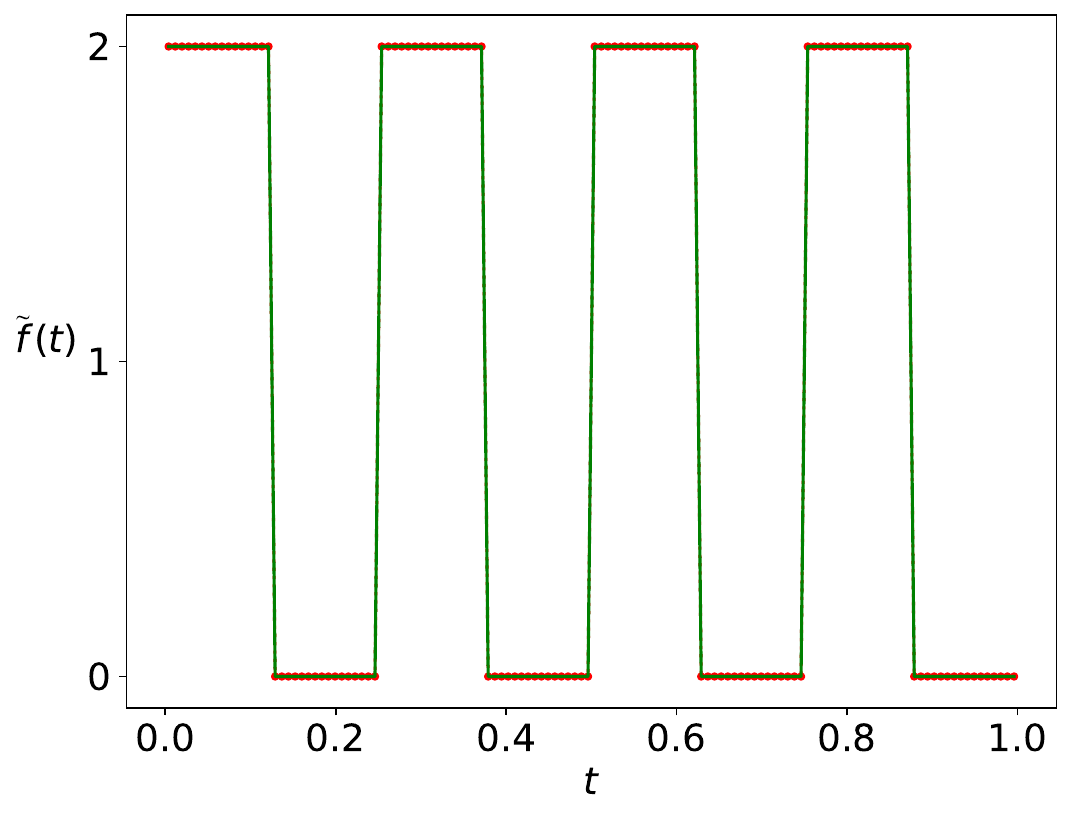}
		\caption{Low-pass filtered signal: \\ Cut-off sequency $= N/4$}
		\label{fig:Low-pass_filtered_signal_f_N_by_4}
	\end{subfigure}
  \begin{subfigure}{0.48\textwidth}
		\centering
\includegraphics[width=\textwidth]{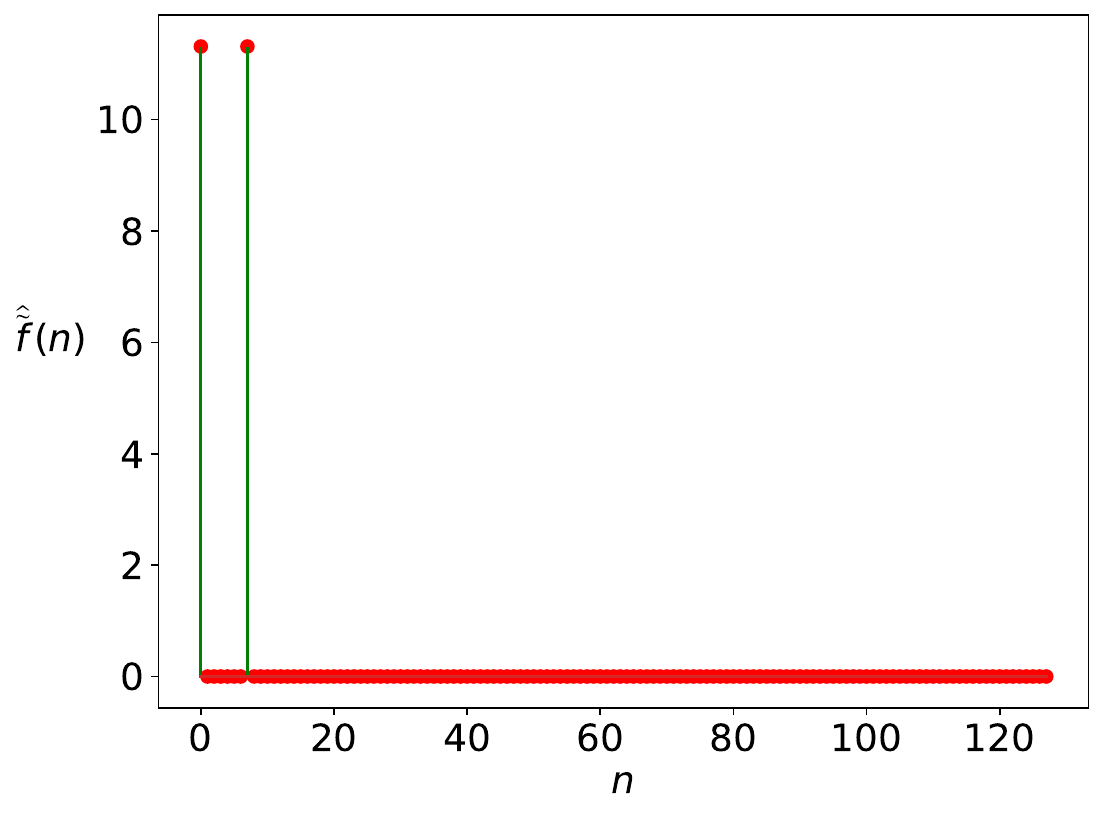}
		\caption{Low-pass filtered spectrum: \\ Cut-off sequency $= N/4$}
		\label{fig:Low-pass_filtered_spectrum_f_N_by_4}
	\end{subfigure}
\\
  \vspace{0.15cm}
   \begin{subfigure}{0.48\textwidth}
		\centering
	\includegraphics[width=\textwidth]{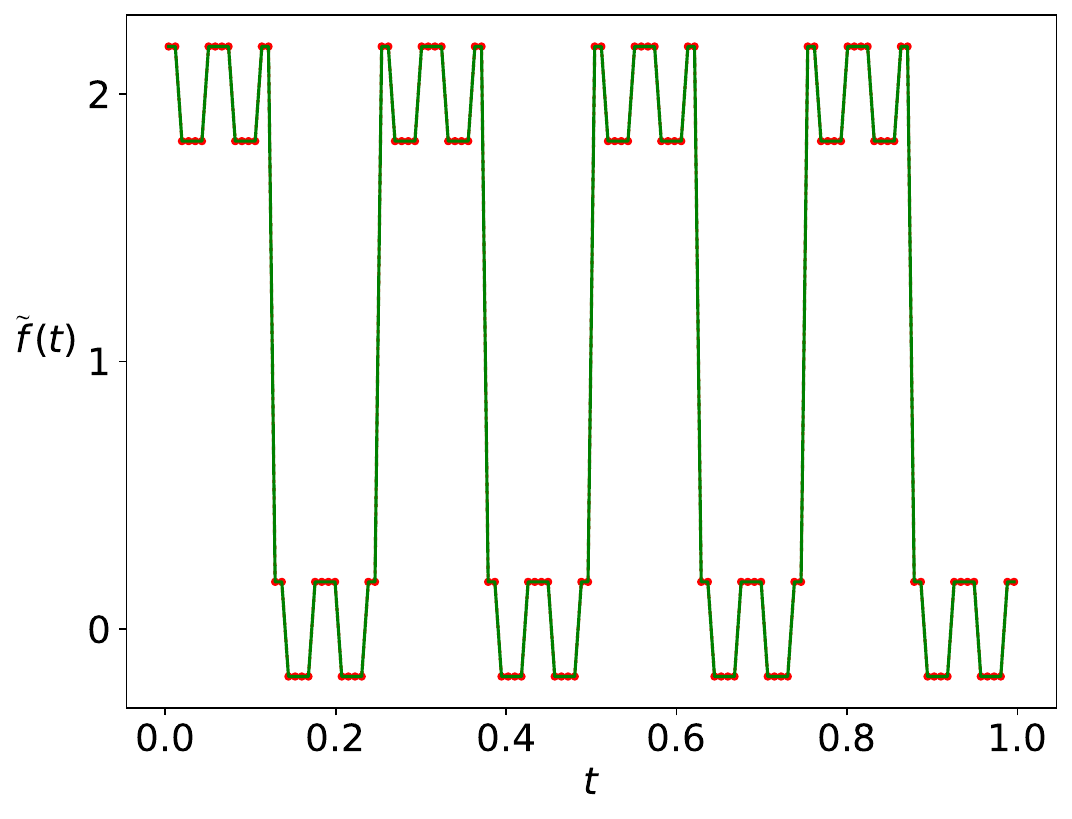}
		\caption{Low-pass filtered signal: \\ Cut-off sequency $= N/2$}
		\label{fig:Low-pass_filtered_signal_f_N_by_2}
\end{subfigure}
 \begin{subfigure}{0.48\textwidth}
		\centering
\includegraphics[width=\textwidth]{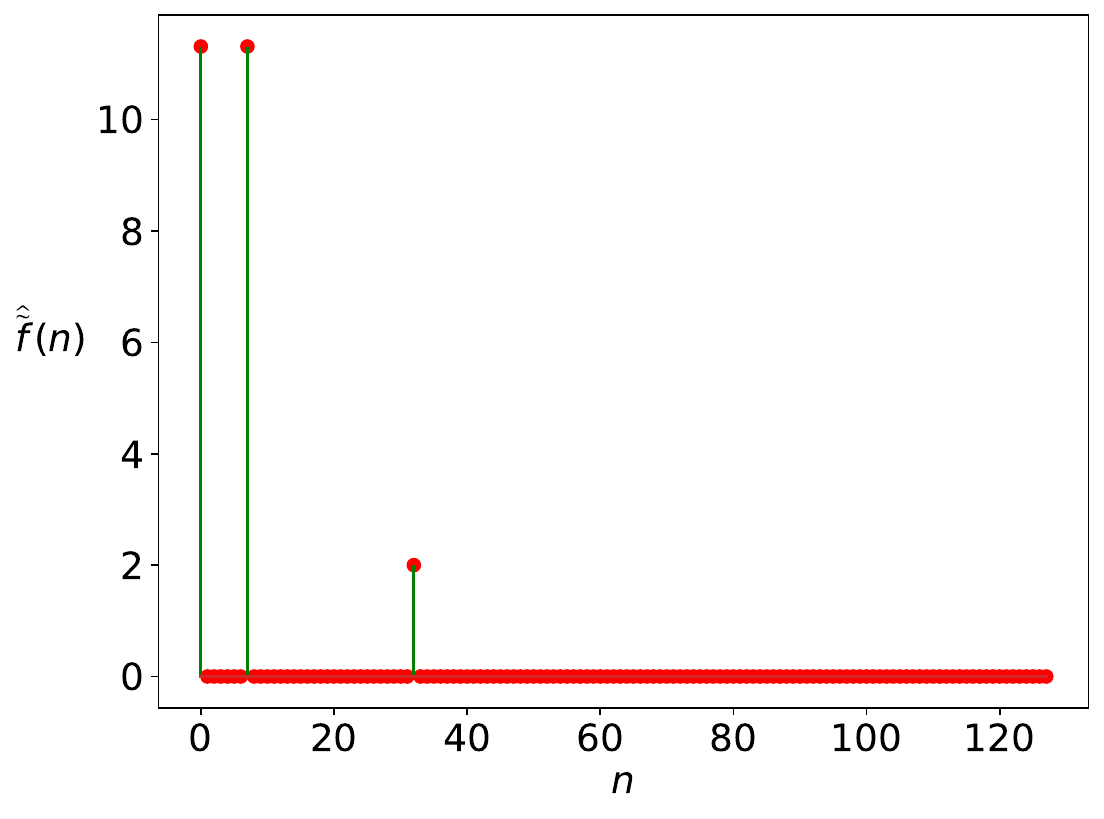}
		\caption{Low-pass filtered spectrum: \\ Cut-off sequency $= N/2$}
	\label{fig:Low-pass_filtered_spectrum_f_N_by_2}
\end{subfigure}
    \caption{Low-pass filtered signals ((a) and (c)) with cut-off sequencies of N/4 ((a)) and  N/2 ((c)) and corresponding spectra are shown on the right column ((b) and (d)). Filtered signals and corresponding spectra obtained from our proposed quantum approach (shown in green) match the expected results (shown in red).
    The input signal $f(t)$ and its sequency spectrum are shown in \mfig{fig:signal_and_spectrum_f}. 
    }
    \label{fig:examples_f_low_pass_filtering}
\end{figure}

\begin{figure}[H]
    \centering
    \begin{subfigure}{0.48\textwidth}
		\centering
\includegraphics[width=\textwidth]{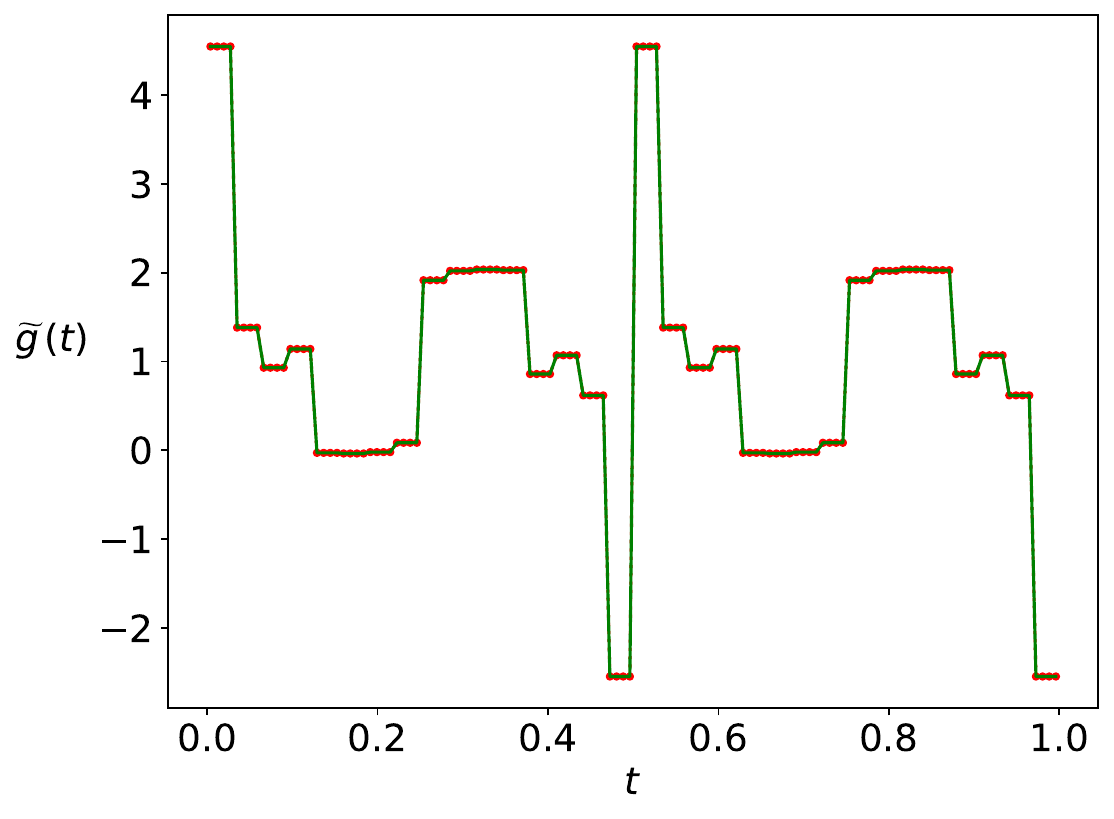}
		\caption{Low-pass filtered signal: \\ Cut-off sequency $= N/4$}
		\label{fig:Low-pass_filtered_signal_g_N_by_4}
	\end{subfigure}
  \begin{subfigure}{0.48\textwidth}
		\centering
\includegraphics[width=\textwidth]{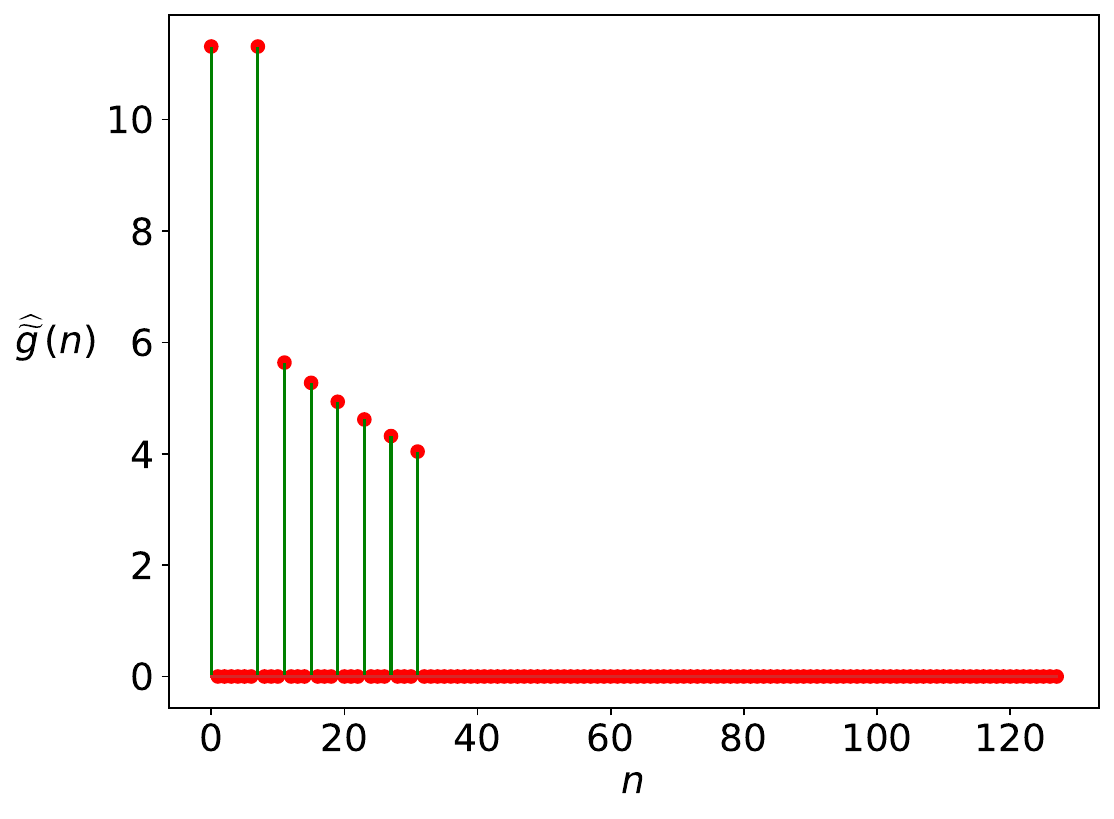}
		\caption{Low-pass filtered spectrum: \\ Cut-off sequency $= N/4$}
		\label{fig:Low-pass_filtered_spectrum_g_N_by_4}
	\end{subfigure}
\\
  \vspace{0.15cm}
   \begin{subfigure}{0.48\textwidth}
		\centering
	\includegraphics[width=\textwidth]{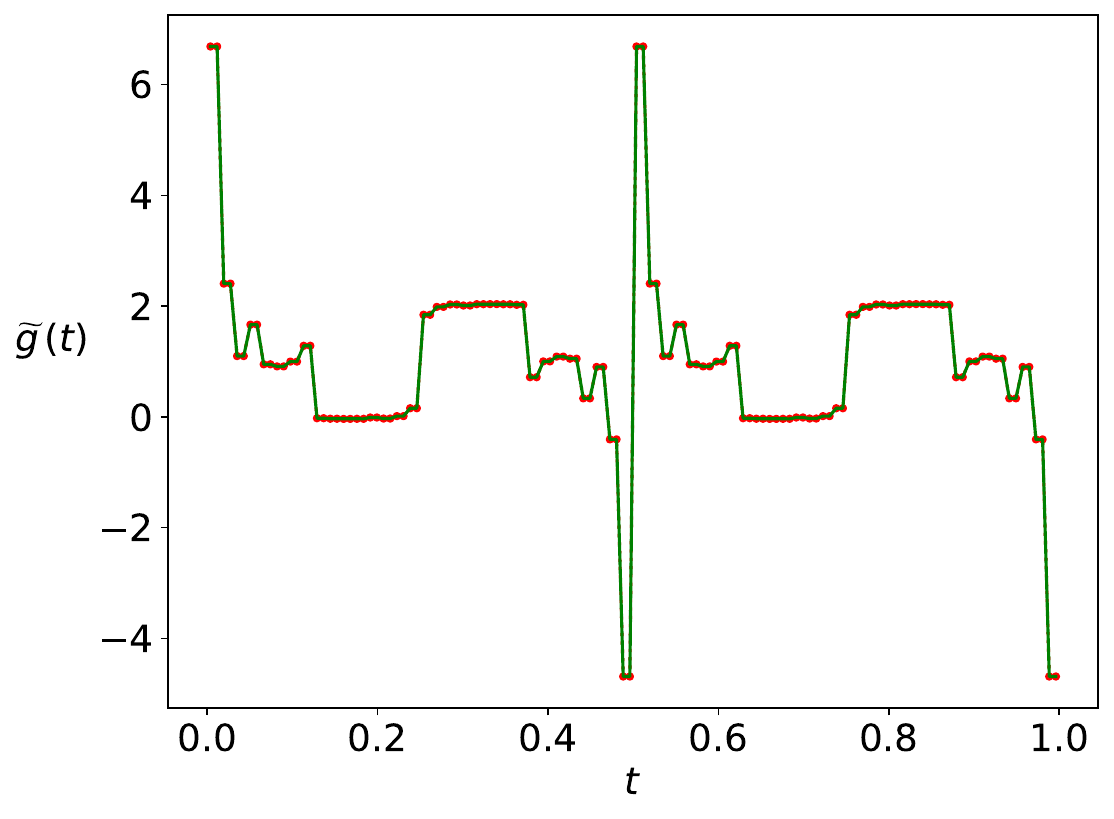}
		\caption{Low-pass filtered signal: \\ Cut-off sequency $= N/2$}
		\label{fig:Low-pass_filtered_signal_g_N_by_2}
\end{subfigure}
 \begin{subfigure}{0.48\textwidth}
		\centering
\includegraphics[width=\textwidth]{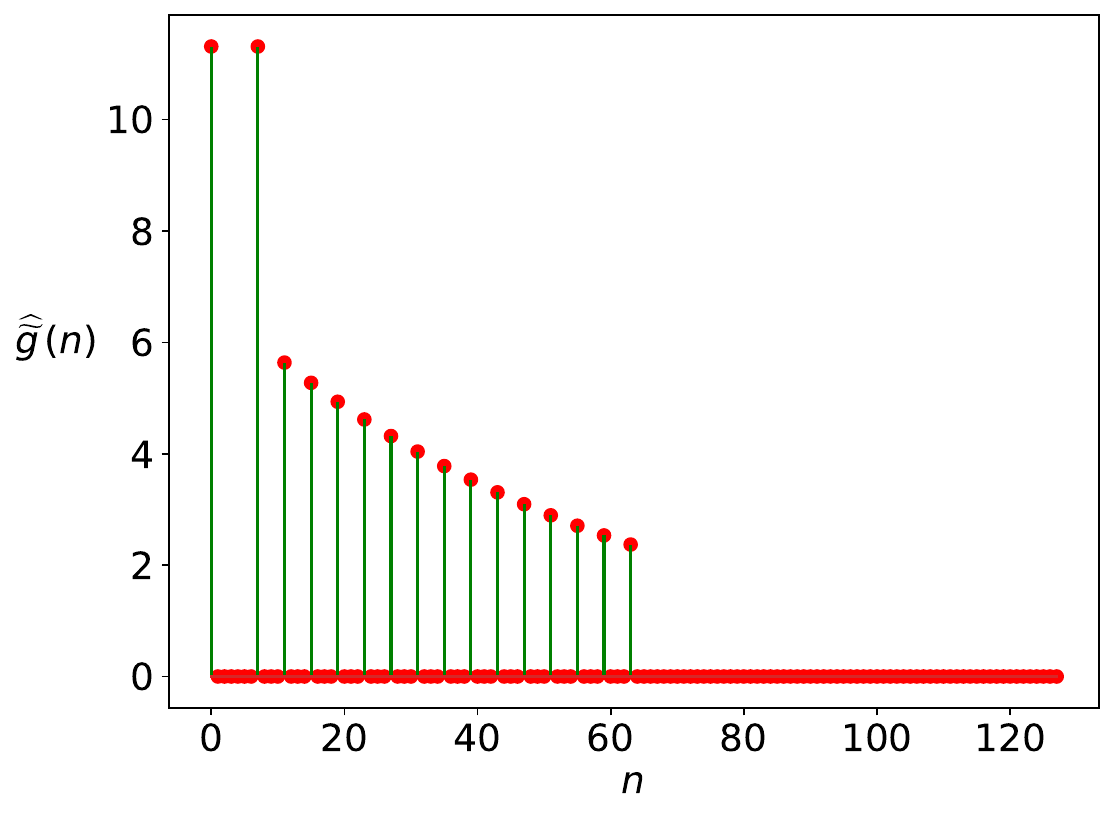}
		\caption{Low-pass filtered spectrum: \\ Cut-off sequency $= N/2$}
	\label{fig:Low-pass_filtered_spectrum_g_N_by_2}
\end{subfigure}
\\
  \vspace{0.15cm}
   \begin{subfigure}{0.48\textwidth}
		\centering
	\includegraphics[width=\textwidth]{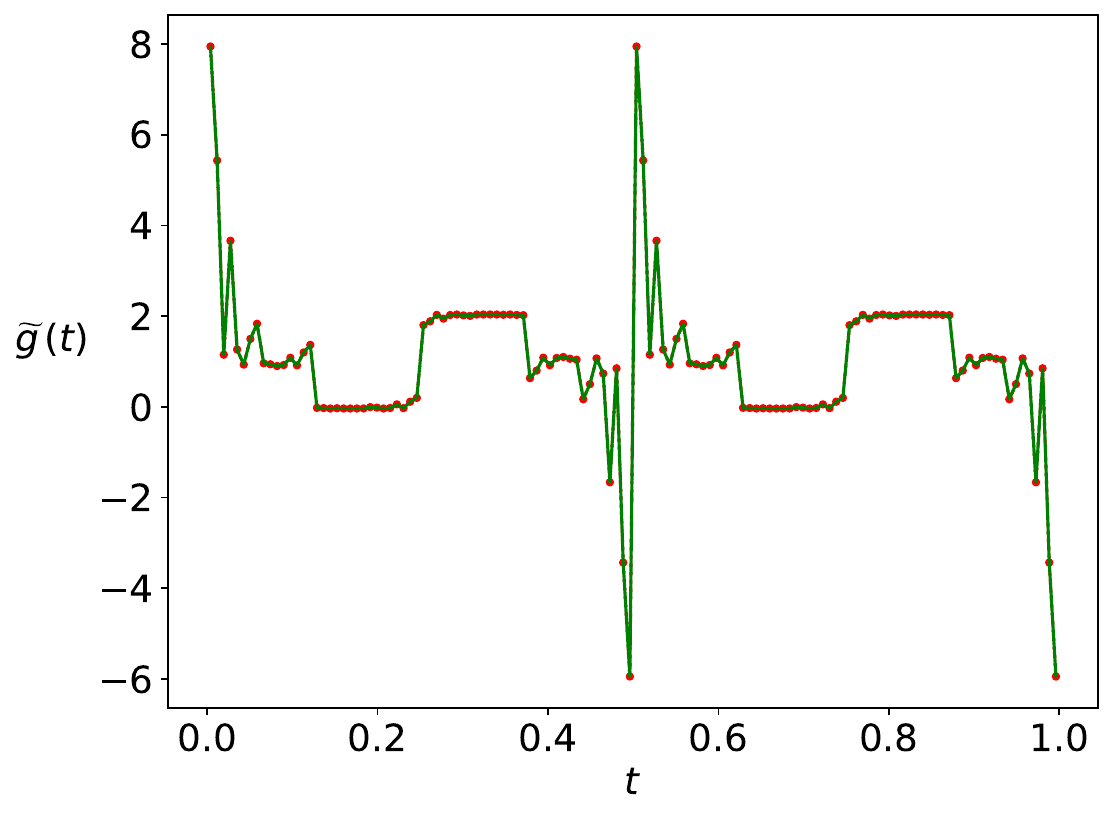}
		\caption{Low-pass filtered signal: \\ Cut-off sequency $= 3N/4$}
		\label{fig:Low-pass_filtered_signal_g_3_N_by_4}
\end{subfigure}
 \begin{subfigure}{0.48\textwidth}
		\centering
\includegraphics[width=\textwidth]{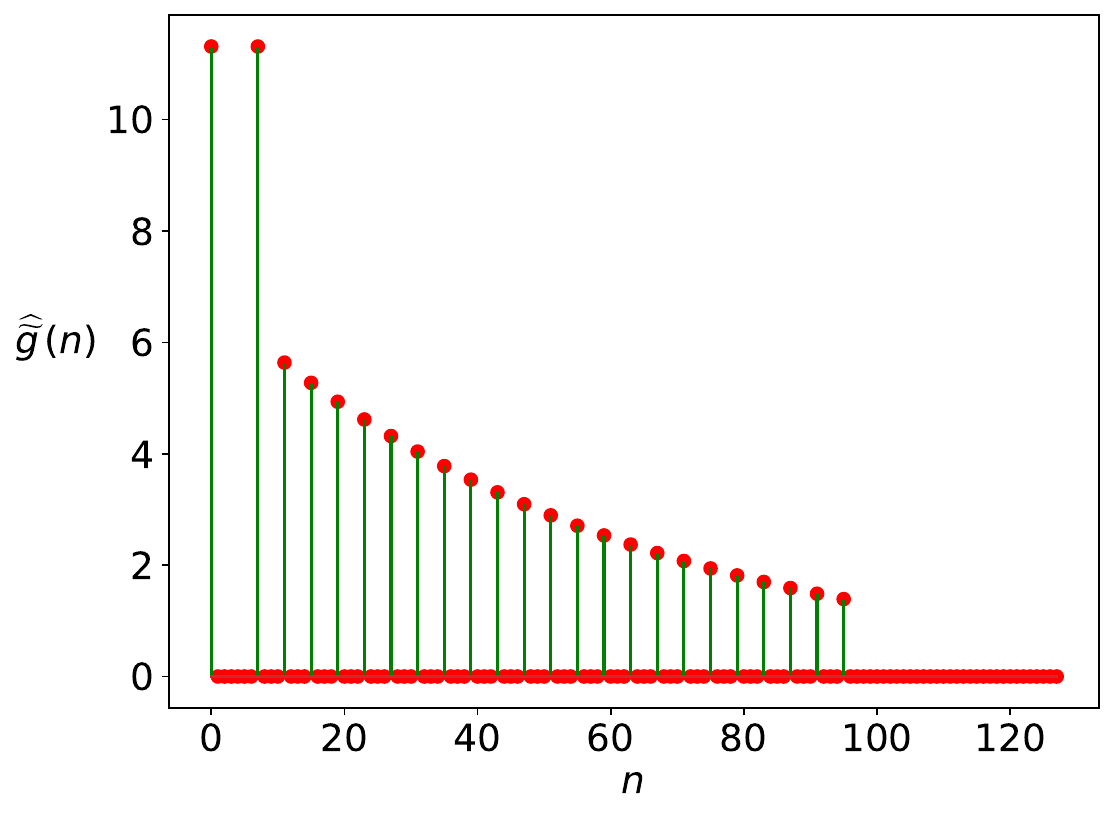}
		\caption{Low-pass filtered spectrum: \\ Cut-off sequency $= 3N/4$}
	\label{fig:Low-pass_filtered_spectrum_g_3_N_by_4}
\end{subfigure} 
    \caption{Low-pass filtered signals ((a) and (c)) with cut-off sequencies of N/4 ((a)) and  N/2 ((c)) and corresponding spectra are shown on the right column ((b) and (d)). Filtered signals and corresponding spectra obtained from our proposed quantum approach (shown in green) match the expected results (shown in red). The input signal $g(t)$ and its sequency spectrum is shown in \mfig{fig:signal_and_spectrum_g}. 
    }
    \label{fig:examples_g_low_pass_filtering}
\end{figure}

\begin{figure}[H]
    \centering
    \begin{subfigure}{0.48\textwidth}
		\centering
\includegraphics[width=\textwidth]{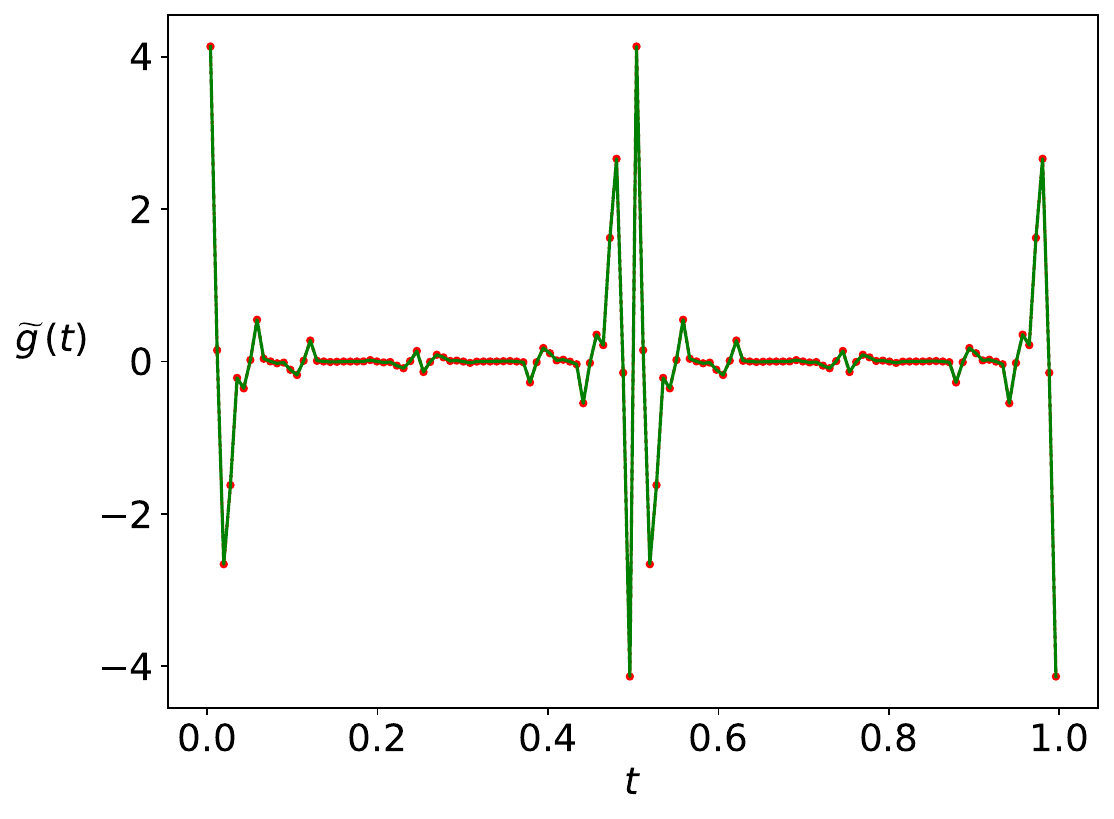}
		\caption{High-pass filtered signal: \\ Cut-off sequency $= N/4$}
		\label{fig:High-pass_filtered_signal_g_N_by_4}
	\end{subfigure}
  \begin{subfigure}{0.48\textwidth}
		\centering
\includegraphics[width=\textwidth]{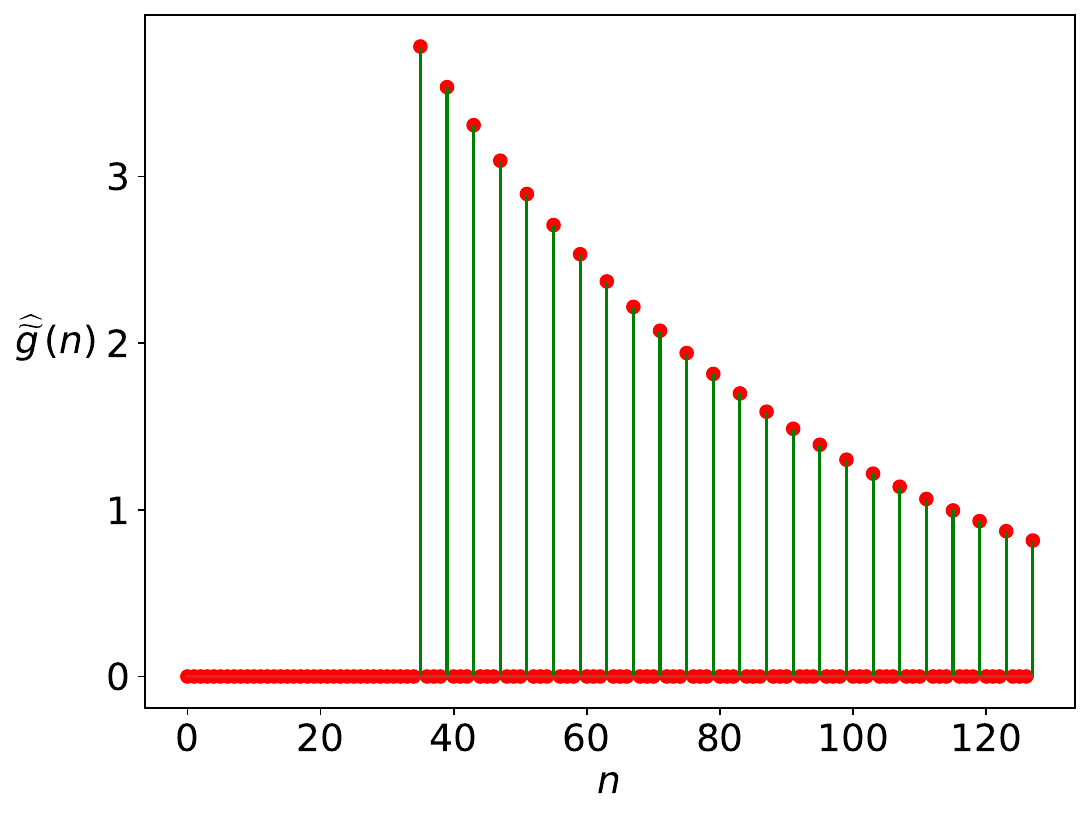}
		\caption{High-pass filtered spectrum: \\ Cut-off sequency $= N/4$}
		\label{fig:High-pass_filtered_spectrum_g_N_by_4}
	\end{subfigure}
\\
   \begin{subfigure}{0.48\textwidth}
		\centering
	\includegraphics[width=\textwidth]{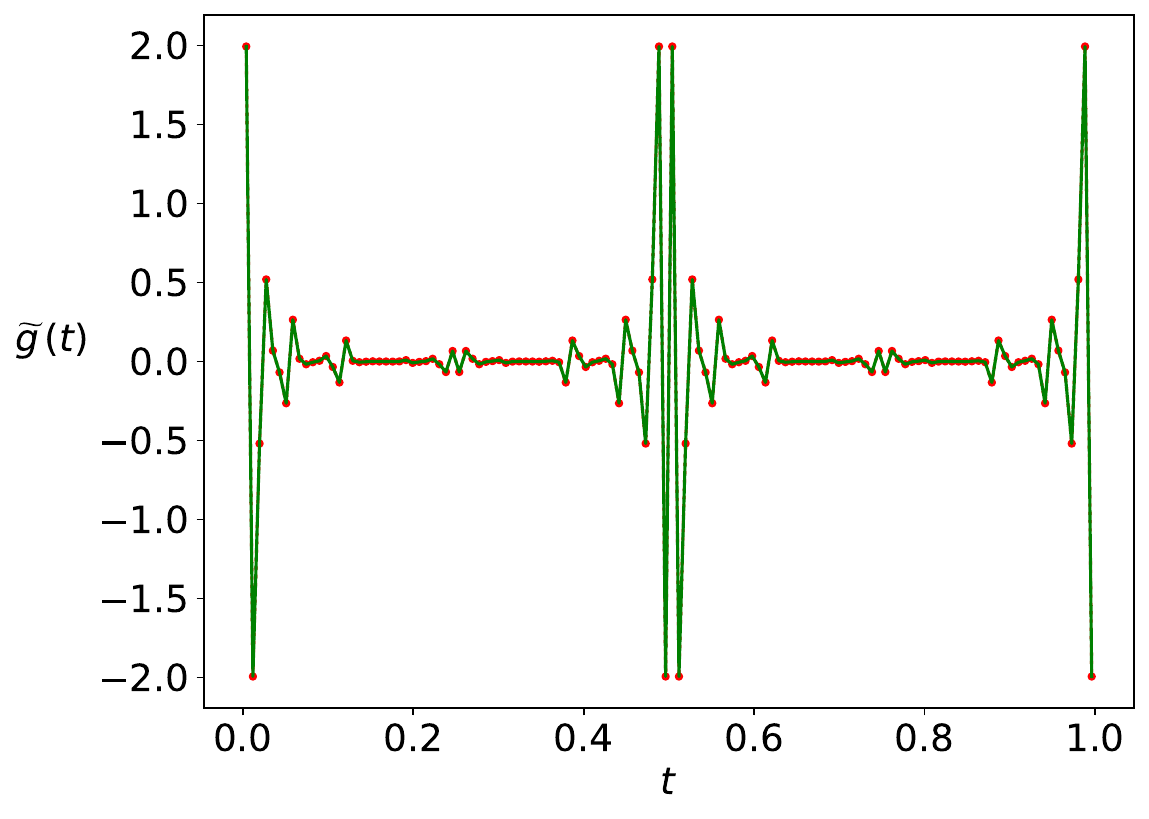}
		\caption{High-pass filtered signal: \\ Cut-off sequency $= N/2$}
		\label{fig:High-pass_filtered_signal_g_N_by_2}
\end{subfigure}
 \begin{subfigure}{0.48\textwidth}
		\centering
\includegraphics[width=\textwidth]{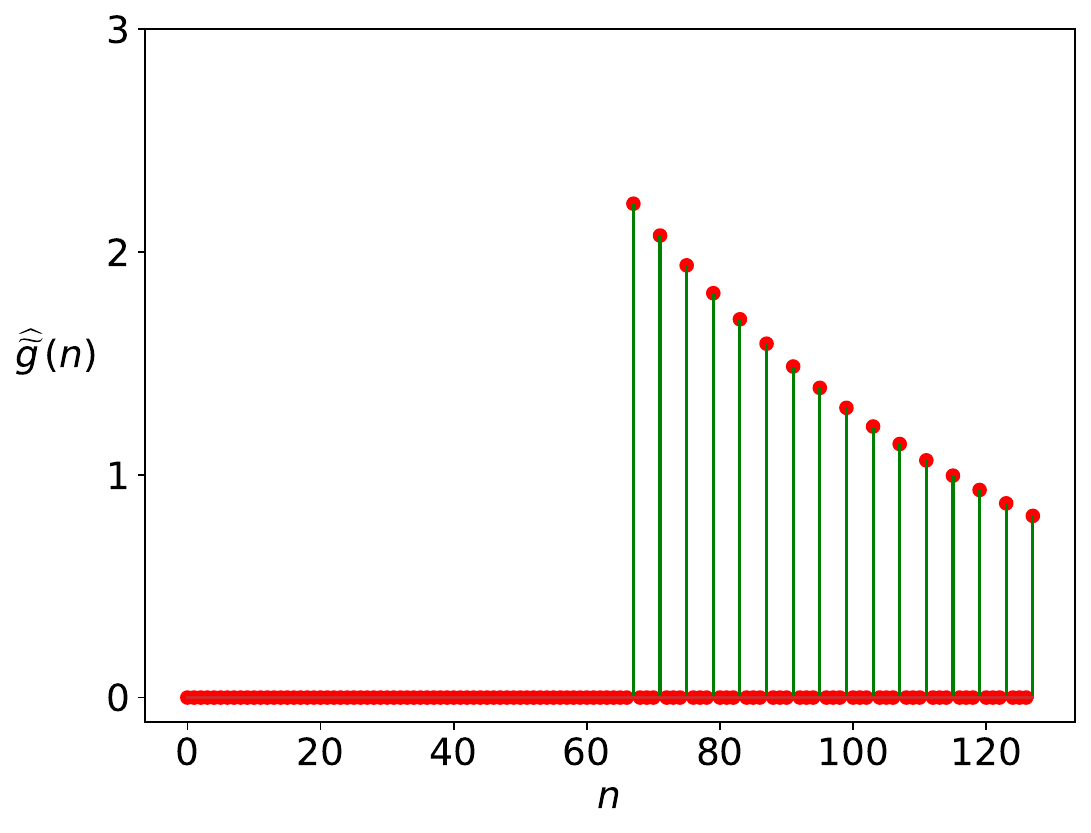}
		\caption{High-pass filtered spectrum: \\ Cut-off sequency $= N/2$}
	\label{fig:High-pass_filtered_spectrum_g_N_by_2}
\end{subfigure}
\\
   \begin{subfigure}{0.48\textwidth}
		\centering
	\includegraphics[width=\textwidth]{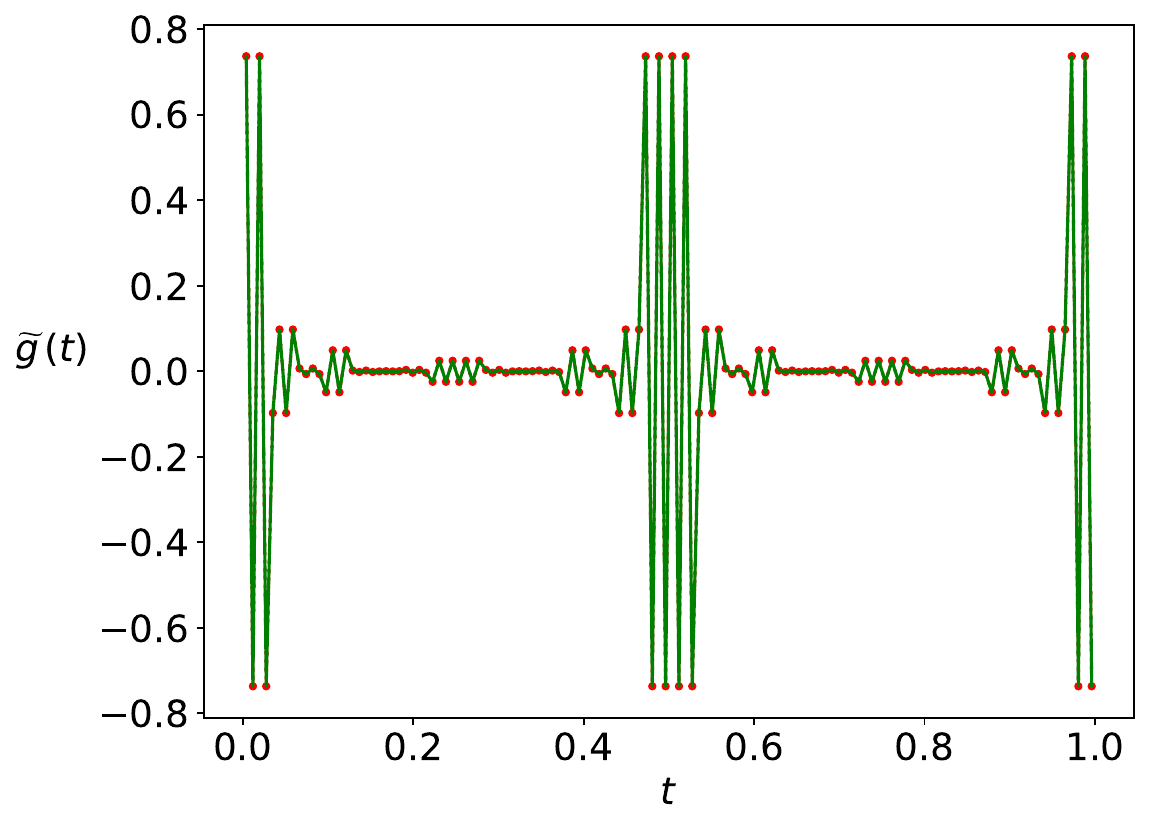}
		\caption{High-pass filtered signal: \\ Cut-off sequency $= 3N/4$}
		\label{fig:High-pass_filtered_signal_g_3_N_by_4}
\end{subfigure}
 \begin{subfigure}{0.48\textwidth}
		\centering
\includegraphics[width=\textwidth]{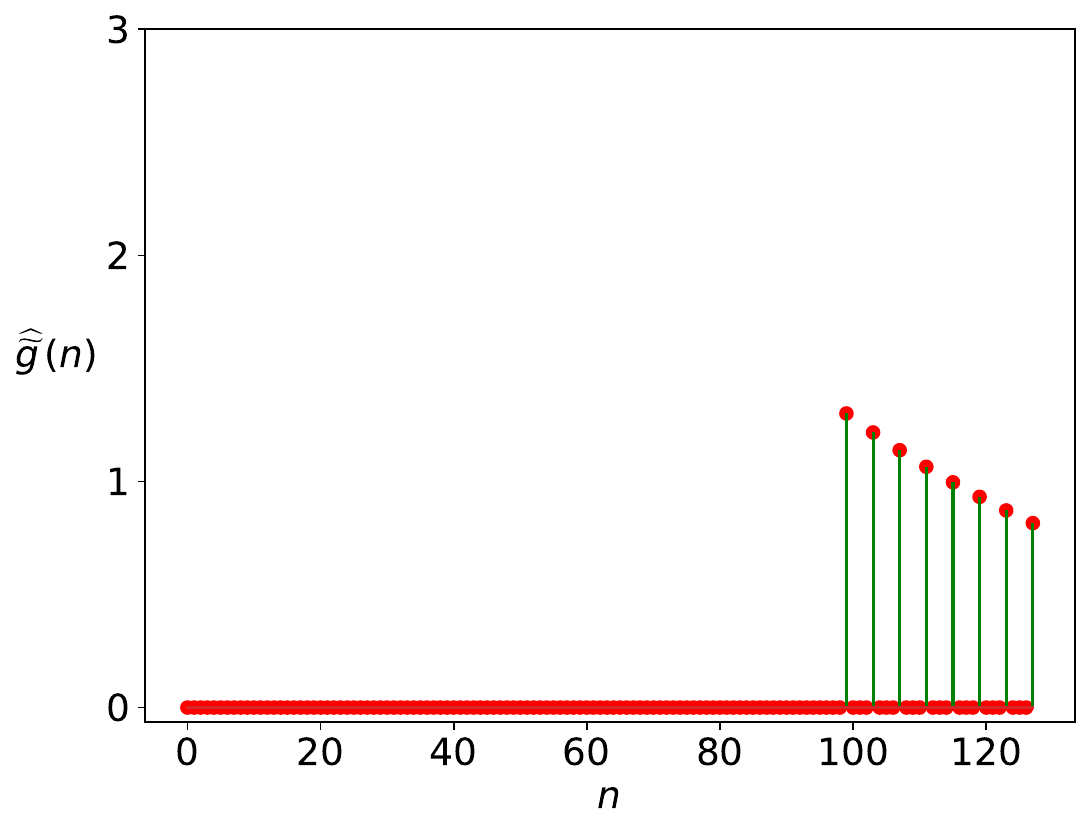}
		\caption{High-pass filtered spectrum: \\ Cut-off sequency $= 3N/4$}
	\label{fig:High-pass_filtered_spectrum_g_3_N_by_4}
\end{subfigure}
    \caption{High-pass filtered signals ((a) and (c)) with cut-off sequencies of N/4 ((a)) and  N/2 ((c)) and corresponding spectra are shown on the right column ((b) and (d)). Filtered signals and corresponding spectra obtained from our proposed quantum approach (shown in green) match the expected results (shown in red). The input signal $g(t)$ and its sequency spectrum are shown in \mfig{fig:signal_and_spectrum_g}.
    }
    \label{fig:examples_g_low_pass_filtering}
\end{figure}


\subsection{DC filtering} \label{Sec:DC}
DC filtering, also known as baseline correction or zero-offset removal, is a process used to remove or minimize the DC (direct current) component from a discrete signal. The DC component represents the average or constant value of the signal, and filtering it out is often necessary in signal processing applications to focus on the variations and fluctuations in the signal.

We note that the DC component of the signal corresponds to the zero sequency component of the signal. One can use Algorithm \ref{alg_filtering}, with the cutoff sequency  $c = 1$, to carry out high-pass filtering, which is equivalent to removing the DC component from the input signal, i.e., the DC filtering of the input signal. However, since the DC component (i.e., the zero sequency) component of the signal is the same in both sequency and natural orderings, computing Hadamard transform in natural order is sufficient for DC filtering, and the conversion from natural to sequency ordering (and its inverse) is not needed. Therefore, Steps $4$ and $7$ in Algorithm \ref{alg_filtering} may be skipped for DC filtering of an input signal. 

 A quantum circuit for DC filtering is shown in \mfig{fig:circuit_DC}. Here the qubit $\ket{q_7}$ is used as an ancilla qubit, and the remaining qubits (from $\ket{q_0}$ to $\ket{q_6}$) are initialized to the normalized input signal.  
Note that the two redundant X gates can be removed from the quantum circuit shown on the left to obtain an equivalent circuit shown on the right in \mfig{fig:circuit_DC}.
  This quantum circuit is used for DC filtering of discretized versions of the input signals $f(t)$ and $g(t)$. 
When the ancilla (i.e. the most significant) qubit is in $\ket{0}$ state, all sequency components other than the DC (i.e., sequency 0) component are allowed to pass through and the DC component is eliminated.
  The resulting DC-filtered signals (and their spectra) are shown in 
 \mfig{fig:dc_filtered_signal_and_spectrum_f} and \mfig{fig:dc_filtered_signal_and_spectrum_g}. The obtained results based on our quantum circuits are in good agreement with the expected results.

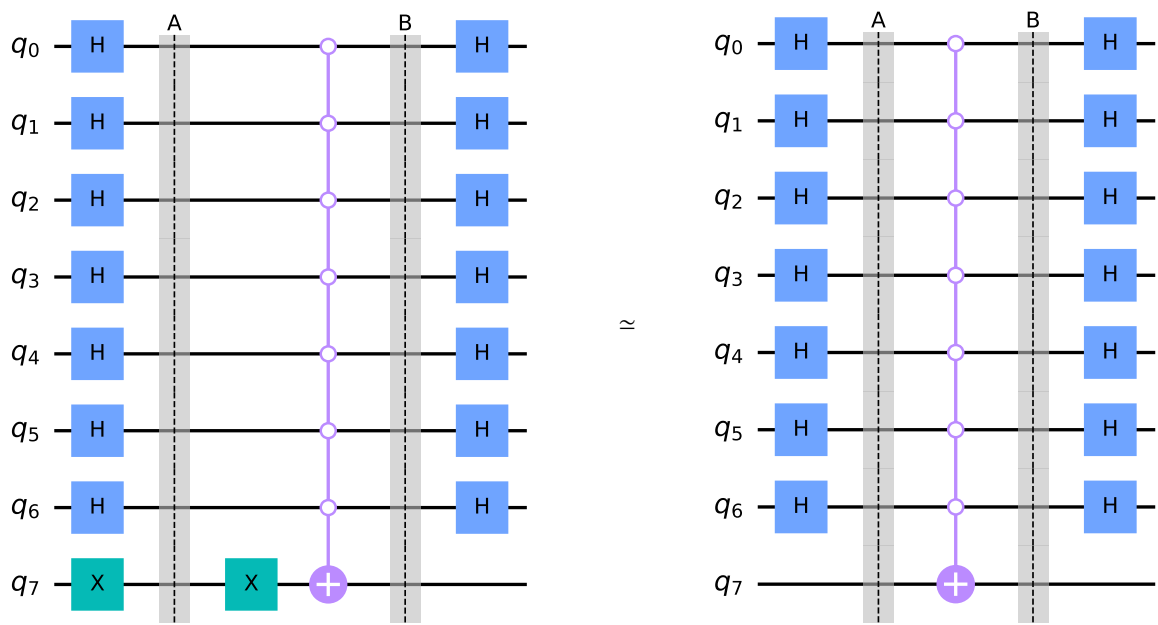
\begin{figure}[H]
    \centering
\includegraphics[width=0.48\textwidth]
{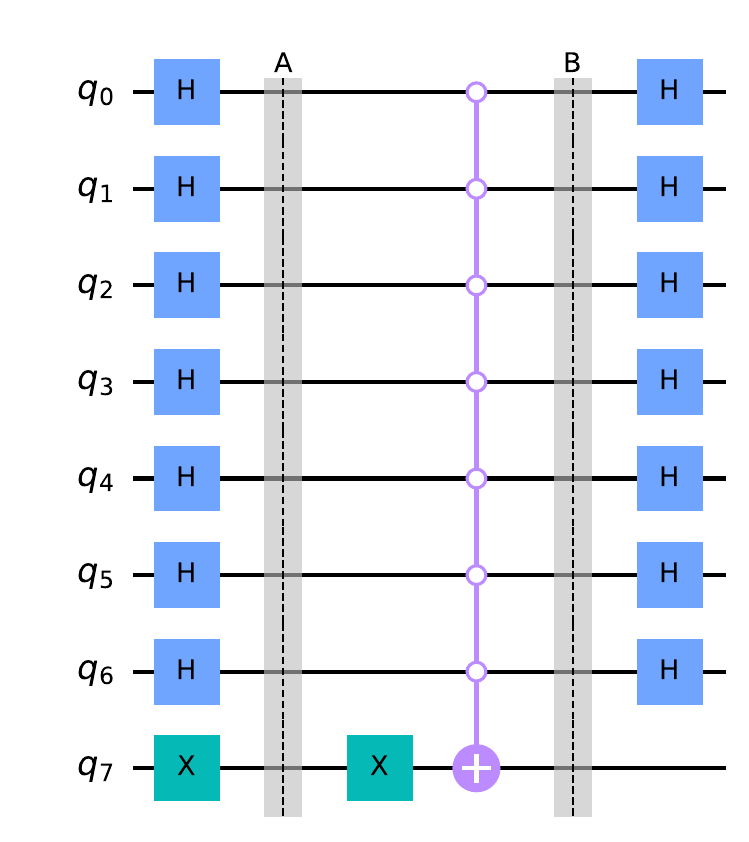}
\begin{tikzpicture}
%
\draw[draw opacity=0] (0,0) -- (1,4) node[above] {$\simeq$};
\end{tikzpicture}
\includegraphics[width=0.42\textwidth]{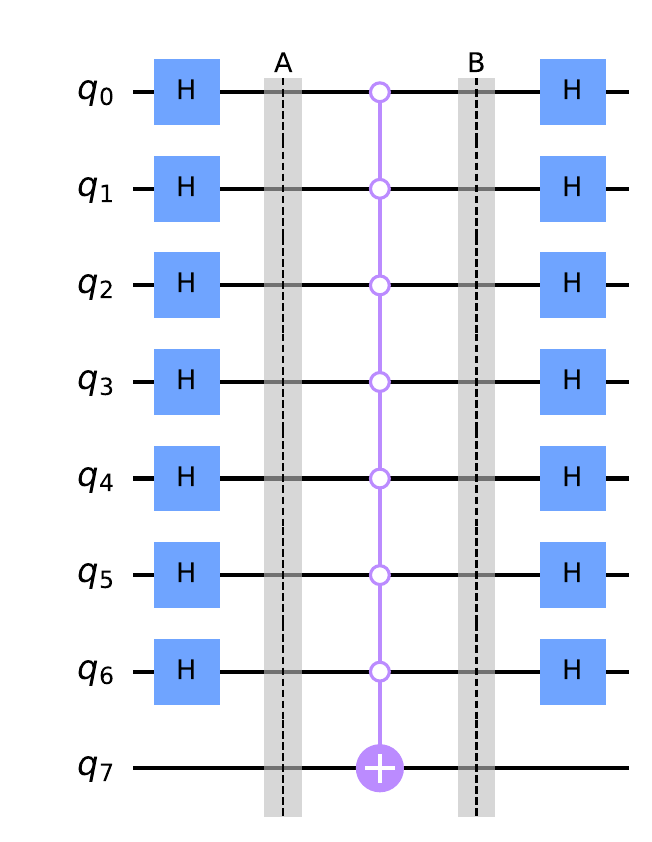}
    \caption{
    Quantum circuits for DC filtering. These circuits selectively allow all sequency components other than the DC (i.e., sequency 0) component when the ancilla (i.e. the most significant) qubit is in $\ket{0}$ state. The quantum circuit on the left is equivalent to the quantum circuit on the right (with two X gates removed).}
    \label{fig:circuit_DC}
\end{figure}

\begin{figure}[H] 
\centering
      \begin{subfigure}{0.48\textwidth}
		\centering
 \includegraphics[width=\textwidth]{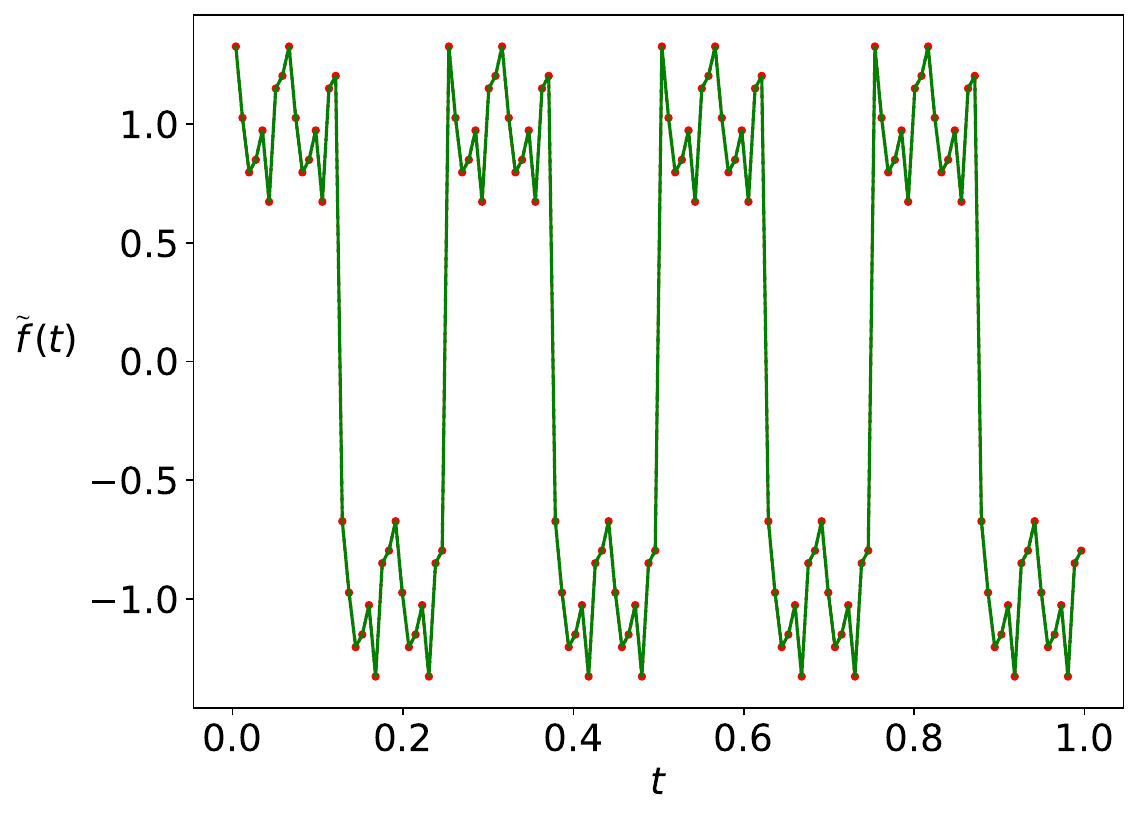}
		\caption{Signal}
		 \label{fig:signal_f}
	\end{subfigure}
\begin{subfigure}{0.48\textwidth}
		\centering
 \includegraphics[width=\textwidth]{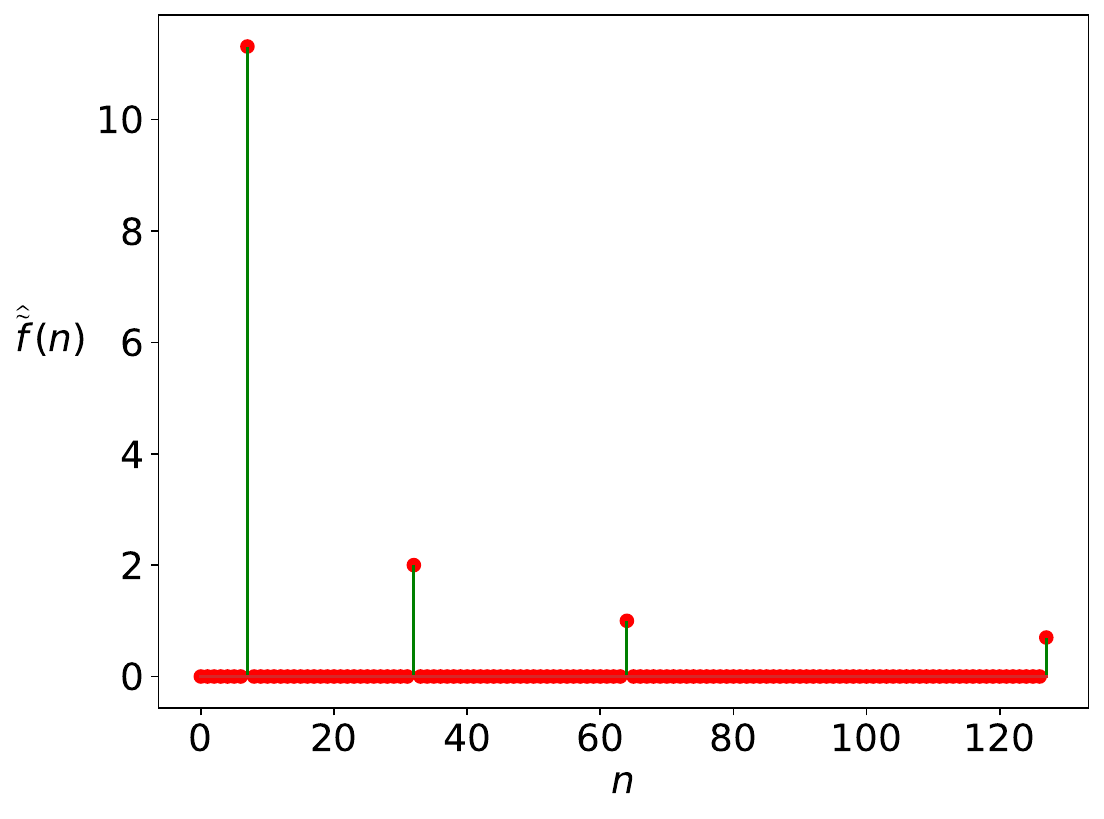}
		\caption{Sequency spectrum}
		 \label{fig:spectrum_f}
\end{subfigure}

    \caption{DC filtered signal $\widetilde{f}(t)$ is shown on left ((a)) and its sequency spectrum  $\widehat{\widetilde{f}}(n)$ is shown on right ((b)).  Filtered signals and corresponding spectra obtained from our proposed quantum approach (shown in green) match the expected results (shown in red). The input signal $f(t)$ and its sequency spectrum are shown in \mfig{fig:signal_and_spectrum_f}.
    }
    \label{fig:dc_filtered_signal_and_spectrum_f}
\end{figure}

\begin{figure}[H] 
\centering
      \begin{subfigure}{0.48\textwidth}
		\centering
 \includegraphics[width=\textwidth]{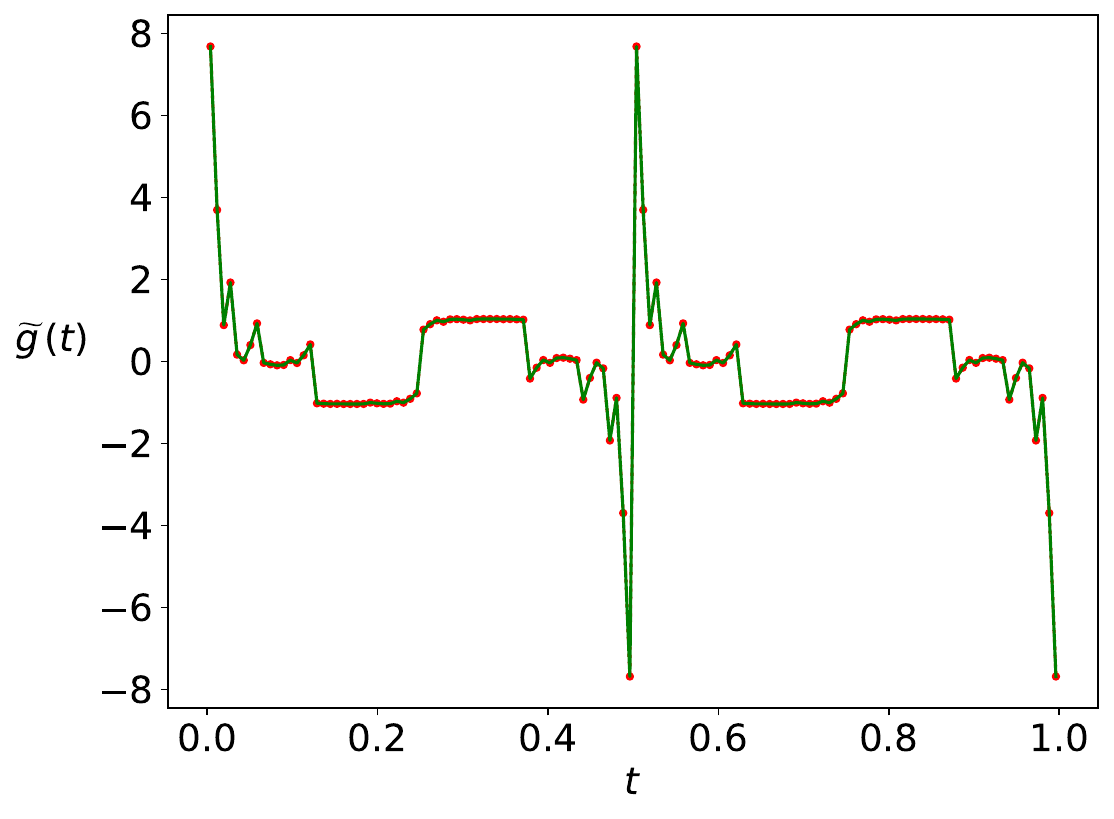}
		\caption{Signal}
		 \label{fig:signal_f_dc}
	\end{subfigure}
\begin{subfigure}{0.48\textwidth}
		\centering
 \includegraphics[width=\textwidth]{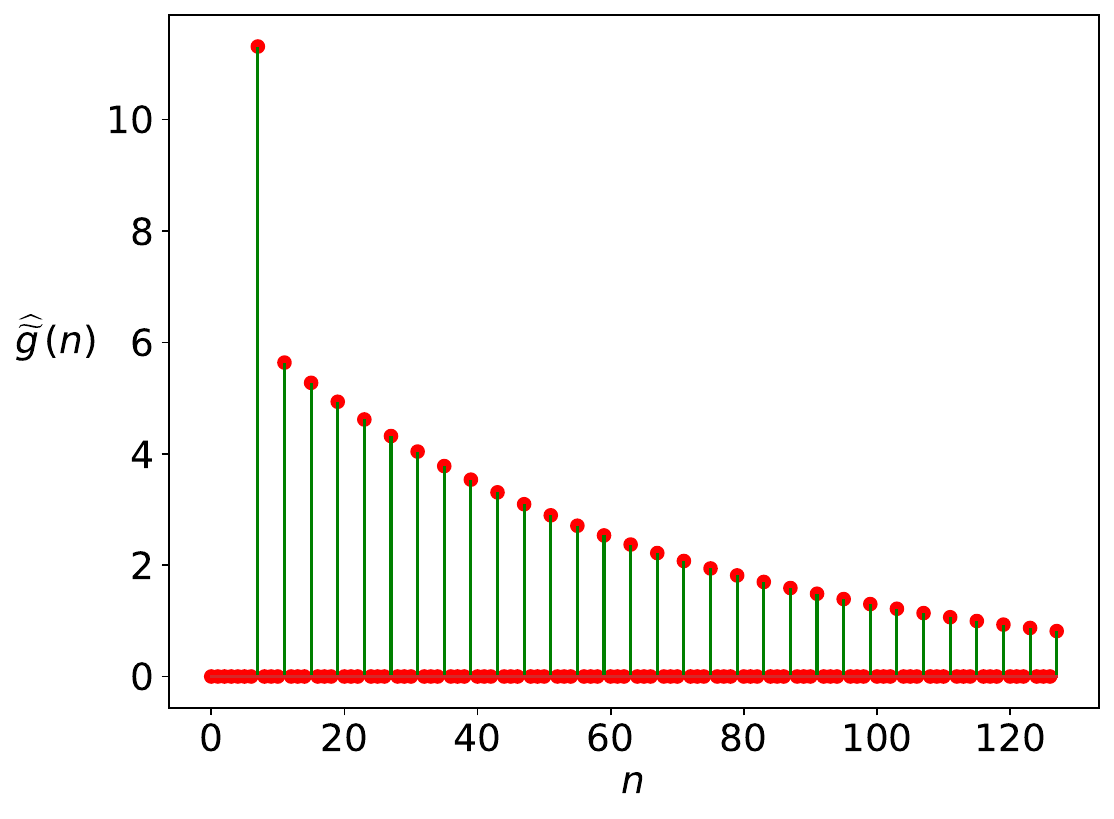}
		\caption{Sequency spectrum}
		 \label{fig:spectrum_f_dc}
\end{subfigure}
    \caption{DC filtered signal $\widetilde{g}(t)$ is shown on left ((a)) and its sequency spectrum  $\widehat{\widetilde{g}}(n)$ is shown on right ((b)).  Filtered signals and corresponding spectra obtained from our proposed quantum approach (shown in green) match the expected results (shown in red). The input signal $g(t)$ and its sequency spectrum are shown in \mfig{fig:signal_and_spectrum_g}. 
    }    \label{fig:dc_filtered_signal_and_spectrum_g}
\end{figure}

\subsection{Band-pass filtering} \label{Sec:bandpass}
In band-pass filtering, sequencies in a certain range are allowed to pass through, while sequencies outside that range are blocked. A band-pass filter is designed to have two cutoff sequencies: a lower cutoff frequency ($c_l$) and an upper cutoff sequency ($c_h$). Sequencies below $c_l$ and above $c_h$ are blocked, while sequencies within the range of $c_l$ to $c_h$ are allowed to pass.

One can perform band-pass filtering by using a combination of low-pass and high-pass filters in a cascade. The idea is to use a high-pass filter to block low sequency components and a low-pass filter to block high sequency components, effectively allowing only the desired band of sequencies to pass through. By adjusting the cutoff frequencies of the low-pass and high-pass filters, one can control the range of sequencies that are allowed to pass through, effectively achieving the desired band-pass filtering effect.

In \mfig{fig:circuit_band_pass_filtering}, a quantum circuit for band-pass filtering is shown. Similar to the quantum circuit diagram shown in \mfig{fig:circuit_DC}, the two redundant $X$ gates are removed from this circuit. The output of this circuit is 
 $\ket{0} \otimes \ket{\widetilde{\Psi}_{b}} + \ket{1} \otimes \ket{\widetilde{\Psi}_{b^{\prime}}}$, where $\widetilde{\Psi}_{b}$  denotes the filtered signal with sequency components within the interval $[N/4, 3N/4)$, and $\widetilde{\Psi}_{b^{\prime}}$ denotes the signal with sequency components outside the range $[N/4, 3N/4)$. We note that in the quantum circuit shown in \mfig{fig:circuit_band_pass_filtering}, within the barriers labeled B and C, the first multiple control $X$ gate with open controls flips the state of the target ancilla qubit to $\ket{1}$ for sequencies less than  $N/4$. The second multiple control $X$ gate with closed controls flips the state of the target ancilla qubit to $\ket{1}$ for sequencies greater than or equal to $3N/4$.

The quantum circuit shown in \mfig{fig:circuit_band_pass_filtering} is used for band-pass filtering of discretized versions of the input signal $g(t)$. The resulting band-pass filtered signal and its sequency spectrum are shown in \mfig{fig:signal_band_pass_filtering}. The obtained results are as expected.

\begin{figure}[H]
    \centering
\includegraphics[width=0.99\textwidth]{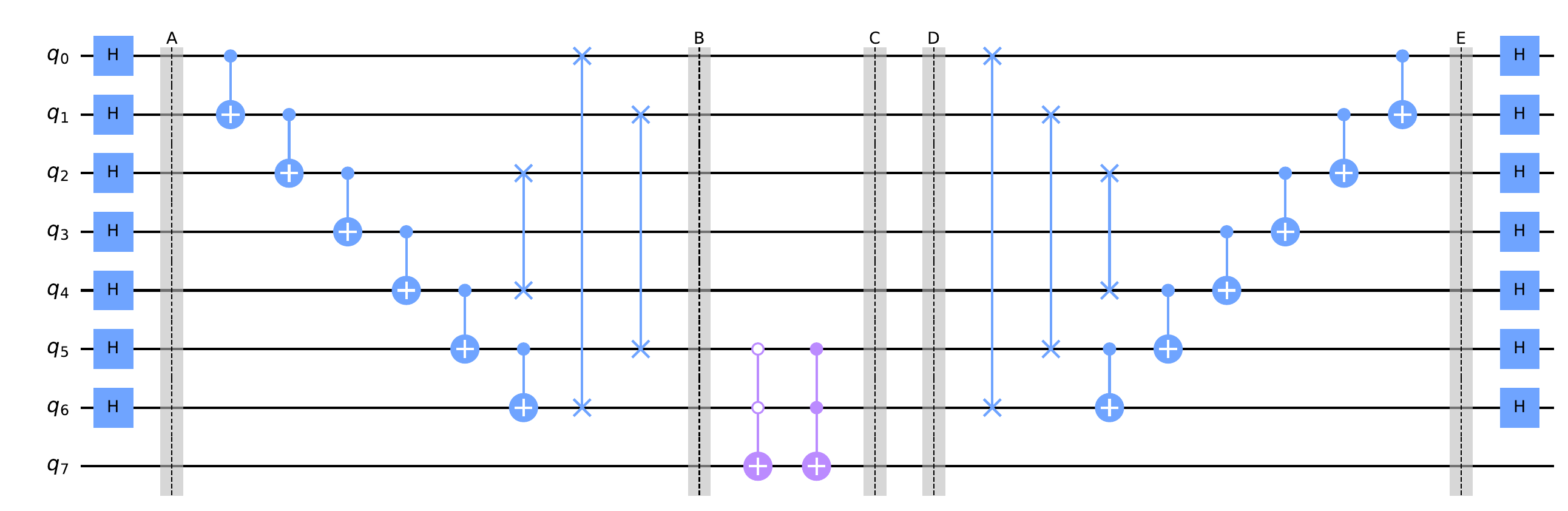}
    \caption{Quantum Circuit for band-pass filtering: This circuit selectively allows to pass sequency components within the range $[N/4,\, 3N/4)$, if the output state of the ancilla qubit $\ket{q_7}$ is $\ket{0}$.}
    \label{fig:circuit_band_pass_filtering}
\end{figure}

\begin{figure}[H]
    \centering
\includegraphics[width=0.48\textwidth]{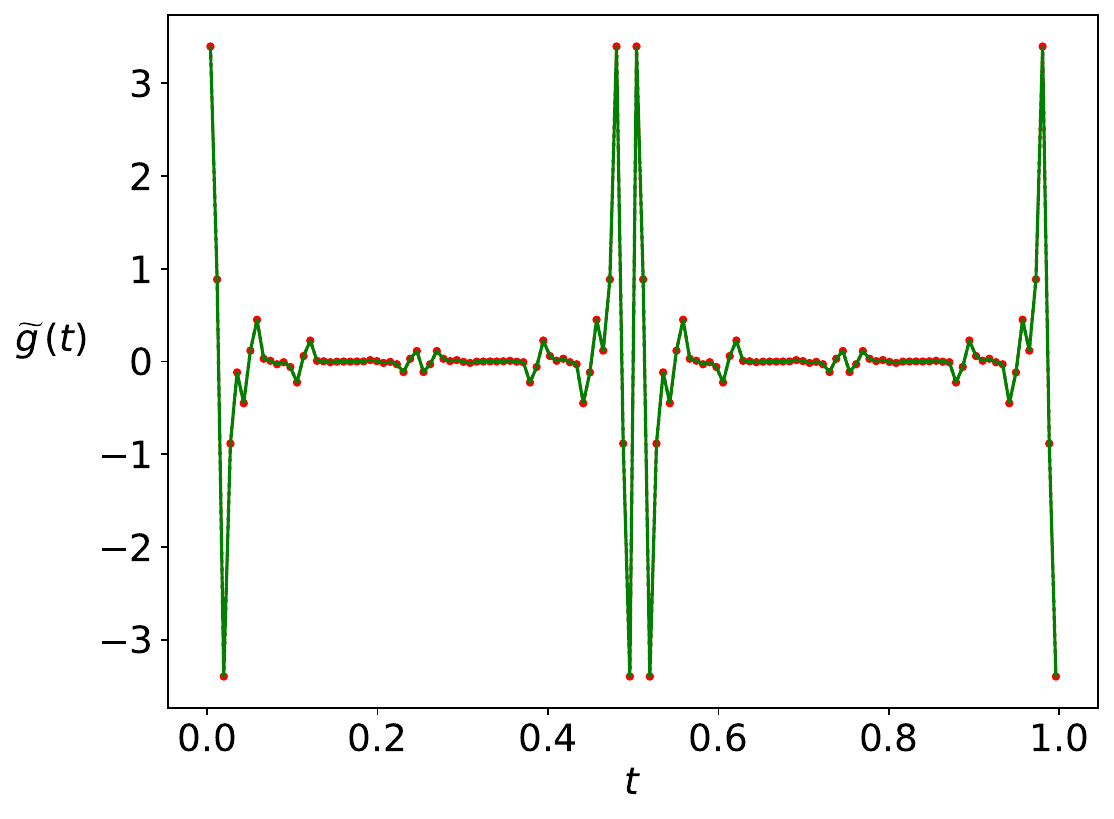}
\includegraphics[width=0.48\textwidth]{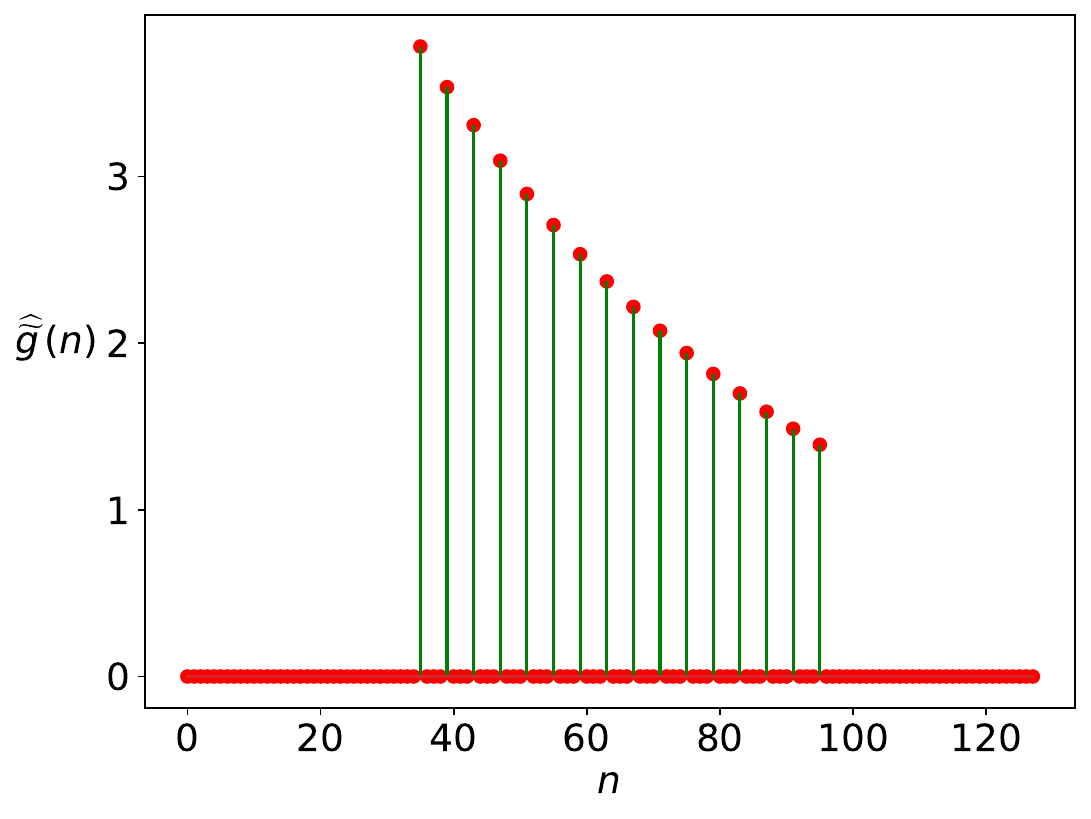}
    \caption{Band-pass filtered signal (left) with sequencies in the range [N/4, 3N/4) and corresponding spectrum (right). Filtered signals and corresponding spectra obtained from our proposed quantum approach (shown in green) match the expected results (shown in red). The input signal $g(t)$ and its sequency spectrum are shown in \mfig{fig:signal_and_spectrum_g}.
    }
    \label{fig:signal_band_pass_filtering}
\end{figure}

\subsection{Complexity Analysis} \label{Sec:complexity}
Assume that the discretized version of the input signal is represented using a vector of length $N=2^n$. Algorithm \ref{alg_filtering}
requires $ (\log_2 N)+1$ qubits for low-pass and high-pass filtering, where the most-significant qubit is used as the ancilla qubit. Step 3 in Algorithm \ref{alg_filtering}, needs $\log_2 N $ Hadamard gates and one $X$ gate (shown before the barrier labeled A in the quantum circuits given in \mfig{fig:low_pass_circuit_filtering}). Steps 4 and 7 in  Algorithm \ref{alg_filtering}, used for conversion from natural ordering to sequency ordering and 
from sequency ordering to natural ordering, respectively, require $O (\log_2 N) $ CNOT and Swap gates (shown between the barriers labeled A and B, and between the barriers labeled D and E, in \mfig{fig:low_pass_circuit_filtering}). Step 6 in  Algorithm \ref{alg_filtering} is optional and can be ignored for the present discussion. The gate complexity and the circuit depth of Step 5 of Algorithm \ref{alg_filtering} depends upon the desired cutoff sequency (this step corresponds to the part of the quantum circuit shown between the barriers labeled B and C in \mfig{fig:low_pass_circuit_filtering}). 
We note that for the examples considered in \mfig{fig:low_pass_circuit_filtering}, for the cutoff sequencies of $\frac{N}{2}$ and $\frac{N}{4}$, one CNOT gate with an open control and one multiple controlled $X$ gate with two open controls are needed. 
It is easy to see that for cutoff sequencies of $\frac{N}{2^r}$, with $1 \leq r \leq n $, one multiple control $X$ gate with $r$ open controls is needed in Step 5 of Algorithm \ref{alg_filtering}. For cutoff sequences of the form  $N - \frac{N}{2^r}$, with $1 \leq r \leq n $, one multiple control $X$ gate with $r$ closed controls is needed in Step 5 of Algorithm \ref{alg_filtering}. This was noted earlier in Remark \ref{Remark_complexity} (b) and (c).
We assume that the cutoff sequency is such that the Step 5 of Algorithm \ref{alg_filtering} can be implemented using  $O (\log_2 N) $ gates. 
Finally, Step 7 in  Algorithm \ref{alg_filtering}, which converts the filtered signal in the sequency domain back to the time domain, requires  $\log_2 N$ Hadamard gates. 

It is clear from the above discussion that the gate complexity of Algorithm \ref{alg_filtering} is  $O (\log_2 N ) $. Here the assumption is that Step 5 of Algorithm \ref{alg_filtering} can be implemented using  $O (\log_2  N ) $ gates, which is the case for the cutoff sequencies of the form $\frac{N}{2^r}$, with $1 \leq r \leq n $ or $N - \frac{N}{2^r}$, with $1 \leq r \leq n $. It can be easily checked that the depth of the corresponding quantum circuit needed  is also $O (\log_2 N ) $.

We note that the filtering approach presented in Algorithm \ref{alg_filtering} can be modified to use QFT (Quantum Fourier Transform) instead of WHT (Walsh-Hadamard Transform) by considering filtering in the frequency domain in place of the sequency domain. If we neglect the cost of state preparation and measurement, the QFT-based filtering approach will have the gate complexity of $O (\left(\log_2 N\right)^2) $ (as the gate complexity of QFT circuit with $n$ qubits is $O (n^2)$, where $n=\log_2 N $). The circuit depth for the QFT-based filtering approach will also be $O (\left(\log_2 N \right)^2) $. 
Since the gate complexity and also the circuit depth of our proposed approach for filtering in the sequency domain is $O (\log_2 N ) $, our proposed approach presented in Algorithm \ref{alg_filtering} (based on Walsh-Hadamard transforms) offers a quadratic improvement in gate complexity and circuit depth in comparison to the QFT-based filtering approach. 

We note that the classical FFT-based approach for filtering is of $O (N \log_2 N ) $. The QFT-based filtering approach already provides an exponential improvement compared to the classical FFT-based filtering approach (where the cost of state preparation and measurement is not considered). As noted above, our proposed approach gives a quadratic improvement in gate complexity and circuit depth compared to the QFT-based filtering approach. 

For many quantum algorithms, including Algorithm \ref{alg_filtering} presented in this work, state preparation and measurement costs remain key challenges in realizing their quantum advantage over their corresponding classical counterparts. However, the proposed algorithm based on Walsh-Hadamard transforms (or the QFT-based filtering approach) will be helpful in situations where the input comes from some other quantum subroutine (as its output), making state preparation unnecessary. Moreover, one can extract global features of the filtered signal without incurring the measurement costs associated with the complete determination of the output filtered signal.

\section{Conclusion}\label{sec:conclusion}

In this work, we introduced a novel quantum approach, including Algorithm \ref{alg_filtering}, for signal filtering, using the Walsh-Hadamard transform in sequency ordering. We effectively addressed DC, low-pass, high-pass, and band-pass filtering tasks using our proposed approach. A quantum circuit for obtaining the sequency-ordered Walsh-Hadamard transform was presented, along with a proof of its correctness. This circuit was used as a key component of quantum circuits for performing low-pass, high-pass, and band-pass filtering using Algorithm \ref{alg_filtering}. Through computational examples, the performance and accuracy of the proposed quantum approach was demonstrated. By converting time domain signals to the sequency domain, selectively preserving desired sequency components, and transforming the filtered signal back to the time domain, the proposed approach consistently achieved the expected results. 

One of the notable advantages of the proposed algorithm lies in its computational efficiency. The gate complexity and circuit depth of the quantum circuit based on Algorithm \ref{alg_filtering} were shown to be  $O (\log_2 N)$, representing a significant improvement over alternative quantum filtering methods based on QFT. Traditional QFT-based signal filtering methods require at least $O ((\log_2 N)^2)$ gates and circuit depth (associated with performing QFT, when the state preparation and measurement costs are ignored), while classical FFT-based approaches exhibit a complexity of $O (N \log_2 N )$. 
Our proposed approach holds promise for faster and more resource-efficient computations in a wide range of signal processing applications owing to the reduced gate complexity and circuit depth.

%

\begin{thebibliography}{10}

\bibitem{deutsch1992rapid}
David Deutsch and Richard Jozsa.
\newblock Rapid solution of problems by quantum computation.
\newblock {\em Proceedings of the Royal Society of London. Series A:
  Mathematical and Physical Sciences}, 439(1907):553--558, 1992.

\bibitem{bernstein1993quantum}
Ethan Bernstein and Umesh Vazirani.
\newblock Quantum complexity theory.
\newblock In {\em Proceedings of the twenty-fifth annual ACM symposium on
  Theory of computing}, pages 11--20, 1993.

\bibitem{shukla2023generalization}
Alok Shukla and Prakash Vedula.
\newblock A generalization of {B}ernstein--{V}azirani algorithm with multiple
  secret keys and a probabilistic oracle.
\newblock {\em Quantum Information Processing}, 22(6):244, 2023.

\bibitem{simon1997power}
Daniel~R Simon.
\newblock On the power of quantum computation.
\newblock {\em SIAM journal on computing}, 26(5):1474--1483, 1997.

\bibitem{grover1996fast}
Lov~K Grover.
\newblock A {F}ast {Q}uantum {M}echanical {A}lgorithm for {D}atabase {S}earch.
\newblock In {\em Proceedings of the Twenty-eighth Annual ACM Symposium on
  Theory of Computing}, pages 212--219. ACM, 1996.

\bibitem{shor1999polynomial}
Peter~W Shor.
\newblock Polynomial-time algorithms for prime factorization and discrete
  logarithms on a quantum computer.
\newblock {\em SIAM review}, 41(2):303--332, 1999.

\bibitem{SHUKLA2022127708}
Alok Shukla and Prakash Vedula.
\newblock A hybrid classical-quantum algorithm for solution of nonlinear
  ordinary differential equations.
\newblock {\em Applied Mathematics and Computation}, page 127708, 2022.

\bibitem{kuklinski1983fast}
WS~Kuklinski.
\newblock Fast {W}alsh transform data-compression algorithm: {ECG}
  applications.
\newblock {\em Medical and Biological Engineering and Computing},
  21(4):465--472, 1983.

\bibitem{zarowski1985spectral}
C~Zarowski and Maurice Yunik.
\newblock Spectral filtering using the fast {W}alsh transform.
\newblock {\em IEEE transactions on acoustics, speech, and signal processing},
  33(5):1246--1252, 1985.

\bibitem{Shukla2022}
Alok Shukla and Prakash Vedula.
\newblock A hybrid classical-quantum algorithm for digital image processing.
\newblock {\em Quantum Information Processing}, 22(1):3, Dec 2022.

\bibitem{lu2016walsh}
Yi~Lu and Yvo Desmedt.
\newblock Walsh transforms and cryptographic applications in bias computing.
\newblock {\em Cryptography and Communications}, 8(3):435--453, 2016.

\bibitem{beer1981walsh}
Tom Beer.
\newblock Walsh transforms.
\newblock {\em American Journal of Physics}, 49(5):466--472, 1981.

\bibitem{ahner1988walsh}
Henry~F Ahner.
\newblock Walsh functions and the solution of nonlinear differential equations.
\newblock {\em American Journal of Physics}, 56(7):628--633, 1988.

\bibitem{gnoffo2014global}
Peter~A Gnoffo.
\newblock Global series solutions of nonlinear differential equations with
  shocks using {W}alsh functions.
\newblock {\em Journal of Computational physics}, 258:650--688, 2014.

\bibitem{gnoffo2015unsteady}
Peter~A Gnoffo.
\newblock Unsteady solutions of non-linear differential equations using {W}alsh
  function series.
\newblock In {\em 22nd AIAA Computational Fluid Dynamics Conference}, page
  2756, 2015.

\bibitem{gnoffo2017solutions}
Peter~A Gnoffo.
\newblock Solutions of nonlinear differential equations with feature detection
  using fast {W}alsh transforms.
\newblock {\em Journal of Computational Physics}, 338:620--649, 2017.

\bibitem{beauchamp1975walsh}
Kenneth~George Beauchamp.
\newblock {\em Walsh functions and their applications}.
\newblock Academic Press, 1975.

\bibitem{shukla2022quantum}
Alok Shukla.
\newblock A quantum algorithm for counting zero-crossings.
\newblock {\em arXiv preprint arXiv:2212.11814}, 2022.

\bibitem{barenco1995elementary}
Adriano Barenco, Charles~H Bennett, Richard Cleve, David~P DiVincenzo, Norman
  Margolus, Peter Shor, Tycho Sleator, John~A Smolin, and Harald Weinfurter.
\newblock Elementary gates for quantum computation.
\newblock {\em Physical review A}, 52(5):3457, 1995.

\end{thebibliography}

	\end{document}